\documentclass[11pt]{article}
\usepackage{fullpage}
\usepackage[T1]{fontenc}
\usepackage[utf8]{inputenc}
\usepackage{xcolor}
\usepackage[colorlinks=true]{hyperref}
\definecolor{purple}{HTML}{7068A4}
\definecolor{blue}{HTML}{3F70B2}
\definecolor{orange}{HTML}{A47458}
\definecolor{yellow}{HTML}{9B7A3C}
\hypersetup{
    linkcolor=blue
    ,citecolor=orange
    ,filecolor=purple
    ,urlcolor=orange
    ,menucolor=purple
    ,runcolor=purple
}
\usepackage{amsmath, amssymb, amsthm, amsfonts, graphicx, color, subcaption, enumerate, bm, cleveref, enumitem, array, mathtools, braket, csquotes}
\usepackage[colorinlistoftodos,prependcaption,textsize=tiny]{todonotes}
\usepackage{multicol, multirow}
\usepackage{thmtools, thm-restate}
\usepackage{ragged2e}
\usepackage{lineno}
\usepackage{framed}
\usepackage[framemethod=tikz]{mdframed}
\usepackage[ruled,linesnumbered]{algorithm2e}
\SetKwComment{Comment}{/* }{ */}

\usepackage{subfiles}
\graphicspath{{graphics/}}

\crefname{claim}{Claim}{Claims}
\crefname{property}{Property}{Properties}
\crefname{algocf}{Algorithm}{Algorithms}
\Crefname{algocf}{Algorithm}{Algorithms}

\usepackage[style=trad-alpha,natbib=true,maxcitenames=4]{biblatex}
\newtheorem{theorem}{Theorem}
\newtheorem{question}{Question}

\newtheorem{lemma}{Lemma}[section]
\newtheorem{claim}[lemma]{Claim}

\theoremstyle{definition}
\newtheorem{definition}[lemma]{Definition}

\theoremstyle{remark}

\newtheorem*{remark*}{Remark}

\def\polylog{\operatorname{polylog}}
\newcommand{\eps}{\varepsilon}
\renewcommand{\epsilon}{\varepsilon}
\newcommand{\poly}{\operatorname{poly}}

\newcommand{\R}{\mathbb{R}}

\newcommand{\cA}{\mathcal{A}}

\renewcommand{\emptyset}{\varnothing}

\title{Improved Massively Parallel Shortest Path Computation}

\author{Yi-Jun Chang\footnote{National University of Singapore. ORCID: 0000-0002-0109-2432. Email: cyijun@nus.edu.sg} \and Shaun Quek\footnote{National University of Singapore. ORCID: 0009-0004-6682-742X. Email: shaun.quek@u.nus.edu}}

\date{}

\begin{document}

\maketitle
\thispagestyle{empty}

\begin{abstract}

We present new algorithms for approximate shortest path computation in weighted undirected graphs within the Massively Parallel Computation (MPC) model. Specifically, we construct an $a$-approximate \emph{Single-Source Shortest Path (SSSP) oracle}: Given a query consisting of a source vertex $u \in V$, the oracle returns $a$-approximate distances $\delta(u,v)$ for all $v \in V$.

Fix an $n$-vertex weighted undirected graph. We consider two regimes of the MPC model: the \emph{near-linear} regime, where each machine has $\tilde{\Theta}(n)$ memory, and the \emph{heterogeneous} regime, where all machines have $\Theta(n^\gamma)$ memory for some constant $\gamma\in (0,1]$, except for a single machine with $\tilde{\Theta}(n)$ memory. We develop randomised algorithms that construct $O(1)$-approximate SSSP oracles in $O(\log\log n)$ rounds in the near-linear regime and in $O(\log^2\log n)$ rounds in the heterogeneous regime. These oracles answer distance queries in $O(1)$ rounds. Our results improve upon the recent work of Dory and Matar (SPAA~2024) which, for the same total memory, requires $\left(\log\log n\right)^{2+O(1/\gamma)}$ rounds and applies only to \emph{unweighted} graphs in the heterogeneous setting.

We say that $\tilde N_k(u)$ is the approximate $k$-nearest neighbourhood of a vertex $u$ if there exists an approximate distance estimate $\delta$ such that $\tilde N_k(u)$ consists of the $k$ vertices closest to $u$ according to $\delta$, breaking ties by vertex ID. The main technical ingredient is a new method for efficiently computing approximate $n^{O(\gamma)}$-nearest neighbourhoods in the heterogeneous MPC model in $O(\log^2\log n)$ rounds.  
\end{abstract}

\newpage
\bigskip
\tableofcontents
\bigskip
\thispagestyle{empty}

% \newpage

% \listoftodos

\newpage
\pagenumbering{arabic}

\section{Introduction}
Shortest-path computation is one of the most fundamental problems in computer science. In this work, we study the problem in the \emph{Massively Parallel Computation (MPC)} model~\cite{karloff2010model}, a theoretical abstraction of widely used large-scale data processing frameworks such as Hadoop~\cite{White:2012}, MapReduce~\cite{dg04}, and Spark~\cite{ZahariaCFSS10}. The MPC model has received significant attention across a wide range of computational problems, including graph coloring~\cite{chang2021,czumaj2020,ghaffari2019hardness,ghaffari2020}, independent set~\cite{behnezhad2019mis,ghaffari2018improved,ghaffari2019hardness,ghaffari2019sparsifying,ghaffari2020}, matching~\cite{assadi2019coreset,behnezhad2019mis,behnezhad2023exponentially,czumaj2018compression,ghaffari2018improved,ghaffari2019hardness,ghaffari2020,lattanzi2011filtering}, set cover~\cite{ji2025setcover,laxman2024-set-cover,harvey2018-ratio-algorithms}, graph connectivity~\cite{adoni2018connectivity,assadi2019-sparse-graph-connectivity,behnezhad2019near,coy2022-connectivity,gan2025-grid-connectivity}, minimum spanning tree~\cite{lattanzi2011filtering,nowicki2021,azarmehr2025-mst}, clustering~\cite{bateni2021clustering,bhaskara2018clustering,chaudhuri2022clustering,cohen2021clustering,coy2023clustering,czumaj2024clustering}, subgraph finding~\cite{biswas2020graphcounting,chu2011,pagh2011,suri2011}, string problems~\cite{boroujeni2021approximating,golan_et_al:LIPIcs.CPM.2025.6,HajiaghayiSS19,im2017efficient,koo2024optimal}, and geometric problems~\cite{andoni2014,chang2024fully,HolmT23,jayaram2024massively}.

\paragraph{The MPC Model.} The MPC model is a synchronous distributed computation model in which all machines can directly communicate with one another. The input is initially distributed between machines, where each machine has limited local memory $S$.  Computation proceeds in synchronous message-passing rounds, where the total amount of information a machine can send or receive each round is bounded by $S$. The objective is to minimise the number of communication rounds, which we refer to as the \emph{round complexity}. The local computation power of each machine is usually ignored in the algorithm analysis as communication is the main bottleneck in most settings; in this work, we require at most polynomial (in $S$) running time per machine between rounds. 

In this work, we study the MPC model for graph problems on weighted undirected graphs. Specifically, the input is a graph $G=(V,E)$ with $n$ vertices and $\poly(n)$-bounded positive weights $\omega_G:V\times V\to \mathbb Z$. The edge set $\set{(u,v,\omega_G(u,v)):\set{u,v}\in E}$ is initially arbitrarily distributed across all machines; each vertex is represented by a unique vertex ID of $O(\log n)$ bits.

MPC graph algorithms are typically classified based on the local memory size $S$, where memory is measured in words of $O(\log n)$ bits. The following three regimes are well-studied: 
\begin{description}
    \item[Super-Linear:] $S \in \Theta(n^{1 + c})$, where $c > 0$ is any constant.
    \item[Near-Linear:] $S \in \tilde{\Theta}(n)$.
    \item[Sublinear:] $S \in \Theta(n^\gamma)$, where $\gamma \in (0,1)$ is any constant.
\end{description}
We note that the round complexity can vary significantly across different memory regimes. For instance, the state-of-the-art upper bounds on round complexity for maximal matching, maximal independent set, $(1-\epsilon)$-approximation of maximum matching, and $(2+\epsilon)$-approximation of minimum vertex cover in the MPC model are as follows: $O(1)$ in the  superlinear regime~\cite{lattanzi2011filtering,ahn2018access,harvey2018-ratio-algorithms}, $O(\log \log n)$ in the near-linear regime~\cite{behnezhad2023exponentially,ghaffari2018improved}, and $O(\sqrt{\log n \log \log n})$ in the sublinear regime~\cite{ghaffari2019sparsifying}.

Typically, the total memory available across all machines is $\tilde O(|E|)$, near-linear in the total input size (we use the notations $\tilde{O}(\cdot)$, $\tilde{\Theta}(\cdot)$, and $\tilde{\Omega}(\cdot)$ to suppress factors of $\log^{\pm O(1)} n$). However, following the work of ~\citet{dory2024apsp}, we will be considering a small amount of additional total memory, i.e., $\tilde O(|E| + n^{1+c})$ for a small constant $c\in (0,1]$. Alternatively, one might view our results as applying to graphs that are slightly dense ($|E|\in \Omega(n^{1+c})$).

\paragraph{Distance Oracles and SSSP Oracles.} In this work, we are interested in the \emph{All-Pairs Shortest Paths (APSP)} problem on weighted undirected graphs. A typical formulation of the problem is to output the shortest-path distances between all pairs of vertices. However, a subtle but important challenge arises in the MPC model: Explicitly storing all-pairs shortest distances requires $\Omega\left(\binom n2\right)$ total memory. Unless the graph is sufficiently dense, the available total memory is insufficient to store all $\binom n2$ distances.

To overcome this, a data structure approach is usually taken and we instead construct a \emph{distance oracle}, a sparse structure that approximately preserves distances and supports efficient queries. Specifically, an \emph{$a$-approximate distance oracle} returns, for any query pair $(u, v)$, an estimate of the distance between $u$ and $v$ that is within a factor $a$ of the true shortest-path distance. The number of rounds needed to build the oracle is called the \emph{preprocessing time}, and the number of rounds needed to answer a query is the \emph{query time}. We will be mainly concerned with constant-round query times, and instead focus on optimising the approximation ratio and preprocessing time. 

In fact, our work considers a slightly stronger data structure: an \emph{$a$-approximate single-source shortest paths (SSSP) oracle} which, given a query source vertex $u \in V$, returns $a$-approximate distances $\delta(u, v)$ for all $v \in V$. We emphasise that the query received by the oracle is a \emph{single} source vertex $u$, while the output comprises approximate shortest-path distances from $u$ to every vertex in the graph $G$, so we refer to it as an SSSP oracle. Observe that an SSSP oracle immediately yields a distance oracle.

\paragraph{Approaches in the Near- and Sub-Linear Regimes.} Given a graph~$G$, a \emph{$t$-spanner}~$H$ is a spanning subgraph of~$G$ in which the distance between every pair of vertices is at most~$t$ times their distance in~$G$~\cite{peleg1989}. It is known that every weighted undirected graph over $n$ vertices has a $(2t-1)$-spanner with $O(n^{1+1/t})$ edges~\cite{althofer1993-sparse}, and this bound is tight assuming the Erdős girth conjecture~\cite{erdos1965}. In the near-linear MPC model, \citet{dory2021} showed a near-optimal spanner construction: $O(t)$-spanners of size $\tilde O(n^{1+1/t})$ in constant rounds for any positive integer $t$. Setting $t=\Theta(\log n)$, we can build an $\tilde O(n)$-size spanner and store it on a single machine. The spanner acts as an $O(\log n)$-approximate distance oracle, where queries are answered by a local computation of shortest distances over the spanner. Similar results have been shown by~\citet{bis2021, fischer2022massively}.

While spanner constructions exist for the sublinear MPC model~\cite{bis2021}, it is not clear how to transform a spanner into a constant-round query-time distance oracle when each machine only has sublinear memory. Alternatively,~\citet{din2020} adopted a different approach based on classical techniques for \emph{distance sketches}~\cite{thorup-zwick2005}, and showed that constant-approximate, constant-round query-time distance oracles can be constructed in $\poly(\log n)$ rounds in the sublinear MPC model. 

A natural question that arises in the sublinear MPC model is whether we can achieve both $O(1)$ distance approximations and $o(\log n)$ rounds simultaneously. However, it is understood that the polylogarithmic round complexity of~\cite{din2020} is \emph{necessary} under the widely-believed \emph{1-vs-2-cycles conjecture}~\cite{Beame13,kiveris2014connected,karloff2010model,nanongkai2022equivalence,rastogi2013finding,roughgarden2018shuffles}, which asserts that distinguishing between a single $n$-cycle and two disjoint $(n/2)$-cycles requires $\Omega(\log n)$ rounds in this model. Assuming this conjecture, even basic tasks such as computing $s$-$t$ distances or $s$-$t$ connectivity require $\Omega(\log n)$ rounds in the sublinear regime. 

\paragraph{The Heterogeneous MPC Model.}
Consequently, achieving faster algorithms necessitates allowing one machine to have near-linear local memory. Motivated by the 1-vs-2-cycles conjecture, \citet{fischer2022massively} recently introduced the \emph{heterogeneous} MPC model: All machines have local memory $S \in \Theta(n^\gamma)$ for some constant $\gamma \in (0,1]$, except for a single machine with local memory of size $\tilde{\Theta}(n)$. 

Intuitively,  the {heterogeneous} MPC model is a hybrid between the sublinear and near-linear regimes and allows for a possibility to break the logarithmic barrier given by the 1-vs-2-cycles conjecture. 
For instance, while the best-known graph connectivity algorithm in the sublinear regime runs in $O(\log D + \log \log n)$ rounds~\cite{behnezhad2019near} (where $D$ is the graph diameter), the same problem can be solved in $O(1)$ rounds in the heterogeneous regime~\cite{fischer2022massively,ahn2012analyzing}.

In a recent work,~\citet{dory2024apsp} presented a $\poly(\log\log n)$-round construction of constant-approximate distance oracles for \emph{unweighted} undirected graphs in the heterogeneous MPC regime, where the total memory is $\tilde O(|E|+n^{1+c})$ for a small constant $c > 0$. Their construction relies on two separate distance-preserving structures to handle different distance scales: A distance sketch for short distances and a linear-sized emulator for long distances. The distance sketches are based on the classical techniques of \citet{thorup-zwick2005}, while the emulator builds on the more recent methods of \citet{elkin-neiman2016emulator}. A key limitation of this approach, however, is that the emulator construction only applies to unweighted graphs. Although there have been efforts to extend this to weighted graphs~\cite{elkin2022emulator}, these do not achieve the same stretch guarantees and introduce an undesirable dependence on the edge weights of the graph. Consequently, the approach of \citet{dory2024apsp} remains restricted to unweighted graphs, with no immediate extension to the weighted setting.

This leads us to consider the following question: 
\begin{question}
Does there exist a $\poly(\log\log n)$-round algorithm for constructing $O(1)$-approximate distance oracles for \underline{weighted} graphs in the heterogeneous regime?
\end{question}

\subsection{Our Contributions}
We answer the above question in the affirmative by presenting an algorithm for constructing constant-approximate SSSP oracles on weighted undirected graphs that runs in $O(\log^2 \log n)$ rounds in the heterogeneous MPC model. Furthermore, in the near-linear MPC regime, we improve the round complexity to $O(\log \log n)$.

Throughout this paper, we say that an event occurs \emph{with high probability} (w.h.p.) if it holds with probability at least $1 - 1/\poly(n)$, where $n = |V|$ is the number of vertices in the input graph $G = (V, E)$. Moreover, we reserve the symbol $\gamma$ to refer to the exponent of the local memory $\Theta(n^\gamma)$ of the machines in the sublinear or heterogeneous MPC model, where the model being referenced is made clear in context.

\begin{theorem}[SSSP Oracles in Near-Linear MPC]\label{theorem:near-linear}
    Let $G=(V,E)$ be a weighted undirected graph with $n$ vertices, and let $c \in (0,1]$ be a constant. In the near-linear MPC model with $\tilde O(|E| + n^{1+c})$ total memory, there exists a randomised algorithm that w.h.p.~constructs an $O(1/c^2)$-approximate SSSP oracle in $O(\log\log n)$ rounds. Every subsequent query can then be answered in $O(1)$ rounds.
\end{theorem}

\begin{theorem}[SSSP Oracles in Heterogeneous MPC]\label{theorem:hetero}
    Let $G=(V,E)$ be a weighted undirected graph with $n$ vertices. In the heterogeneous MPC model with $\tilde O(|E| + n^{1+\gamma})$ total memory, there exists a randomised algorithm that w.h.p.~constructs an $O(1/\gamma^2)$-approximate SSSP oracle in $O(\log^2\log n)$ rounds. Every subsequent query can then be answered in $O(1)$ rounds.
\end{theorem}

Our oracle improves upon that of~\citet{dory2024apsp} in two aspects. First, it applies to \emph{weighted} graphs, whereas their construction is limited to \emph{unweighted} graphs. Second but less significantly, it achieves a polynomial (in $\log\log n$) speed-up in preprocessing time in the heterogeneous regime, reducing the complexity from $\left(\log\log n\right)^{2+O(1/\gamma)}$ to $O(\log^2\log n)$ when the total memory is $\tilde O(|E| + n^{1+\gamma})$. This completely removes the dependence on the parameter~$\gamma$ in the exponent. 

We remark that while both our work and~\cite{dory2024apsp} obtain a constant approximation of distances, the oracle of~\cite{dory2024apsp} provides better stretch guarantees. When the total memory is $\tilde O(|E| + n^{1+\gamma})$, while both our distance oracle and that of~\cite{dory2024apsp} achieve $O(1/\gamma^2)$-approximations,~\cite{dory2024apsp} was able to achieve a better $O(1/\gamma)$-approximation if we allow a $\poly(\log\log n)$ query time. Moreover,~\cite{dory2024apsp} showed that with more total memory, their distance oracles can achieve $3(1+\eps)$-approximations under constant-round query times (and $\poly(\log\log n)$ preprocessing time), or $(1+\eps)$-approximations under $\poly(\log\log n)$ query times (and $O(1)$ preprocessing time) for any $0 < \eps < 1/2$. 

We also remark that our algorithms can be modified to report both the approximate distance as well as a path that witnesses this distance. Specifically, we construct a constant-approximate \emph{path-reporting distance oracle}: Given a query pair $(u,v)$, output a path $P$ connecting $u$ and $v$ such that the length of $P$ (i.e., the sum of the weights of all edges in $P$) is within a constant factor of the true shortest-path distance between $u$ and $v$. We refer to \Cref{appendix:paths} for more details. 

In view of these results, a natural question to ask is if increasing the \emph{total} memory of the MPC model to, say, $\poly(n)$ leads to faster APSP algorithms. We first note that once we have $\Omega(n^2)$ total memory, we are no longer restricted to constructing a distance oracle and can instead explicitly output approximate distances for all vertex pairs. Namely in the near-linear MPC regime, $\Omega(n^2)$ total memory also allows us to directly apply congested-clique algorithms (see below), for which faster constant-approximate APSP algorithms are known (e.g.,~\cite{bui2024improved}). In the heterogeneous MPC regime, however, it is unclear to us whether additional total memory alone can lead to faster algorithms. We believe this is an interesting direction for future research.

\subsection{Related Work}\label{subsec:related-works}

\paragraph{Bypassing the 1-vs-2-cycles conjecture.} It is well known that the logarithmic barrier imposed by the 1-vs-2-cycles conjecture can be bypassed using geometric information~\cite{andoni2014}. Concretely, an \emph{embedded planar graph} is a planar graph equipped with a straight-line drawing, where the coordinates of each vertex in the plane are provided as input. For such graphs, in the sublinear MPC model, a $(1+\epsilon)$-approximate SSSP can be computed in $O(1)$ rounds using $O(n)$ total memory, and a $(1+\epsilon)$-approximate APSP can be computed in $O(1)$ rounds using $O(n^2)$ total memory~\cite{chang2024fully,HolmT23}.

\paragraph{Congested Clique.} Of related interest is the distributed model known as the \emph{congested clique}. In this model, the communication network forms a fully connected clique of $n$ vertices, where each vertex can send $\Theta(\log n)$-bit messages to every other vertex in synchronous rounds. Given an $n$-vertex graph $G$ as input, each vertex initially knows only the edges incident to it in $G$. By the end of an APSP algorithm, every vertex should know the (approximate) distances to all other vertices in $G$. While $G$ can be an arbitrary graph on $n$ vertices, the communication network itself is a clique.

The MPC and congested clique models share several similarities. For instance, any congested clique algorithm for dense graphs with $O(n^2)$ edges that uses at most $O(n)$ local memory per vertex can be simulated directly in the near-linear MPC model with nearly the same round complexity~\cite{behnezhad2018mpc-cc}. In fact, it is common for algorithms developed for one model to be adapted for the other~\cite{dory2021,bis2021,czumaj2020}.

APSP in the congested clique has received considerable attention in recent years. Of particular relevance is the work of \citet{dory2022exponentially}, which presented a $\poly(\log\log n)$-round algorithm for $(2+\varepsilon)$-approximate APSP in unweighted graphs, and the more recent work of \citet{bui2024improved}, which achieved an $O(\log\log\log n)$-round algorithm for $O(1)$-approximate APSP in weighted graphs. Notably, both the heterogeneous MPC algorithm of \citet{dory2024apsp} and the congested clique algorithm of \citet{dory2022exponentially} rely on the emulator of \citet{elkin-neiman2016emulator} as a central component.

% \yijun{I think it is quite different, although they have the same name... so this is potentially confusing. Only the part of sampling a subset of nodes and then build a virtual graph over it is the same, the detail of how this virtual graph is built is different. This should be clarified. However, indeed this is the right citation for this very general concept of skeleton graphs.}

\paragraph{Skeleton Graphs.} Our work adapts several ideas from \citet{bui2024improved}, the most significant of which is the \emph{skeleton graph} (\Cref{definition:skeleton-graph}). The general skeleton graph concept refers to the notion of building a new graph over a sampled subset of the original vertices. As far as we know, \cite[Algorithm 2.3]{ullman91} is the earliest application of this technique; we however refer to the particular version used in \cite{bui2024improved}, which differs significantly from that of \cite{ullman91}. 

At a high level, our work involves a hierarchical construction of the skeleton graph to build a polylog-stretch low-hopbound emulator (\Cref{definition:emulator}) as a key step towards constructing the SSSP oracle. We note that this manner of hierarchical construction also shares similarities with the work of \citet{pram-subemulator}. We defer the comparisons to \cite{bui2024improved} and \cite{pram-subemulator} to \Cref{subsec:comparison-with-prior-work}.

%\paragraph{From Randomised to Deterministic.} A very recent work by \citet{cho2026-spaa} presented techniques for de-randomising known randomised shortest-path algorithms for the MPC and the congested clique models. This included the aforementioned $O(1)$-APSP congested clique algorithm of \cite{bui2024improved} as well as the shortest-path MPC algorithms of \cite{dory2024apsp}. Given the similar nature of the randomisation used in our work to that of previous works, we strongly believe that the techniques of \cite{cho2026-spaa} can be applied to de-randomise our algorithms and obtain deterministic variants of \Cref{theorem:near-linear} and \Cref{theorem:hetero}. We leave this to be verified by the reader.

\paragraph{From Randomised to Deterministic.}
A very recent work by \citet{cho2026-spaa} introduced techniques for de-randomising randomised shortest-path algorithms in the MPC and congested clique models. Their approach is based on de-randomising hitting-set computations, and applies, in particular, to the aforementioned $O(1)$-approximate APSP congested clique algorithm of \cite{bui2024improved} and the shortest-path MPC algorithms of \cite{dory2024apsp}. Given the similar nature of the randomisation used in our work to that of previous works, we expect that the techniques of \cite{cho2026-spaa} can be adapted to obtain deterministic variants of \Cref{theorem:near-linear} and \Cref{theorem:hetero}.

\subsection{Organisation}
In \Cref{sect:prelim}, we introduce the terminology and tools used throughout the paper. In \Cref{section:tech-overview}, we provide a technical overview of our algorithms. In \Cref{section:SSSP-oracles}, we present our main results, the SSSP oracle algorithms. In \Cref{section:k-nearest-hopset,section:extracting-neighbourhoods,section:skeleton-graph}, we describe the three building blocks underlying these algorithms: constructing $k$-nearest hopsets, extracting $k$-nearest neighbourhoods, and constructing skeleton graphs, respectively. In \Cref{appendix:mpc-tools}, we review the basic communication primitives for MPC algorithms. In \Cref{appendix:paths}, we extend our construction to output the paths corresponding to the approximate distances.

\section{Preliminaries}\label{sect:prelim}

In this section, we present the terminology and the tools used throughout the paper.

\subsection{Graph Notations}\label{subsect:graph_notation}

For a weighted graph $G = (V, E)$, we let $d_G(u, v)$ denote the distance between vertices $u$ and $v$ in $G$; if there is no path connecting $u$ and $v$, then we denote $d_G(u,v)=\infty$. The \emph{weighted diameter} of $G$ is defined as the maximum distance between any pair of vertices, i.e., $\max_{u, v \in V} d_G(u, v)$. We write $d_G^{(\beta)}(u, v)$ to denote the distance between $u$ and $v$ via a path consisting of at most $\beta$ edges. 

Unless otherwise specified, all graphs considered in this paper are assumed to be \emph{simple} and \emph{undirected}. Each vertex has a distinct ID of $O(\log n)$ bits. As is standard, we assume that the input graph $G$ has positive edge weights that are polynomially bounded integers.

For any edge $e = \{u, v\} \in E$, we denote its weight by $\omega_G(u, v) = \omega(e)$. As convenient shorthands, for any pair $\{u, v\} \notin E$ with $u \ne v$, we let $\omega_G(u, v) = \infty$; for any $u\in V$, we let $\omega_G(u, u)=0$. When convenient, we may use the subscript to indicate the specific edge set under consideration. For any path $P$, we use $\omega(P)$ to denote its length, defined as the sum of the weights of all edges in $P$.

\paragraph{Approximations.} 
We say that $\Delta \in \mathbb{R}$ is an \emph{$a$-approximation} of a value $d$ if $d \leq \Delta \leq a \cdot d$. Similarly, we call a function $\delta : V \times V \to \mathbb{R}\cup\set{\infty}$ an \emph{$a$-approximate} distance estimate for the graph $G$ if
$$
d_G(u, v) \leq \delta(u, v) \leq a \cdot d_G(u, v)
$$
for all pairs of vertices $u, v \in V$.

\paragraph{$k$-Nearest Neighbourhoods.} Consider a graph $G = (V, E)$. For any vertex $v \in V$ and positive integer $k$, we define $N_{G,k}(v) \subseteq V$ as the \emph{$k$-nearest neighbourhood} of $v$, consisting of the $k$ vertices in $G$ that are closest to $v$, including $v$ itself. We often omit $G$ from the subscript when the graph is clear in context. 

We also define the edge set
$$
E_k := \left\{ \{u, t\} : u \in V \wedge t \in N_k(u) \right\} \quad \text{with} \quad \omega_{E_k}(u, t) = d_G(u, t).
$$
We refer to $E_k$ as the \emph{$k$-nearest edge set} of $G$. In most instances throughout this work, computing the $k$-nearest neighbourhoods of all the vertices in $G$ amounts to computing the $k$-nearest edge set of $G$ -- in such cases, we use the terms interchangeably.

We also consider additional useful variants of $k$-nearest neighbourhoods. We denote $N_k^{(h)}(v) \subseteq V$ as the \emph{$k$-nearest $h$-hop neighbourhood} of $v$, which consists of the $k$ vertices closest to $v$ according to the distances given by the function $d_G^{(h)}$. We say that $\tilde{N}_k(v)$ is an \emph{$a$-approximate $k$-nearest neighbourhood} of $v$ if there exists an $a$-approximate distance estimate $\delta$ for $G$ such that $\tilde{N}_k(v)$ consists of the $k$ vertices closest to $v$ according to $\delta$. 

Unless otherwise specified, all references to $k$-nearest neighbourhoods and their variants in this work assume that ties are broken by vertex ID.

% \subsection{The MPC Model} The \emph{Massively Parallel Computation} (MPC) model~\cite{karloff2010model} is a synchronous distributed computation model in which all machines can directly communicate with one another. Memory is measured in words of $O(\log n)$ bits, and we use $S$ to denote the memory available to each machine. Computation proceeds in synchronous rounds, where each machine can send and receive at most $O(S)$ words per round. The total memory available across all machines depends on the algorithm and is typically almost linear in the input size.

% We identify machines in the MPC model by their integer IDs in $\{1, 2, \ldots, M\}$, where $M$ is the total number of machines. We refer to a set of machines as a \emph{contiguous segment} if their IDs form a consecutive sequence.
% The number of communication rounds is the standard measure of time for an MPC algorithm, which we refer to as the \emph{round complexity}.

\subsection{Distance-Preserving Structures} Throughout this paper, we consider three types of distance-preserving structures.
\begin{definition}[Spanners]
    Given a graph $G=(V,E)$, a \emph{$t$-spanner} is a subgraph $H\subseteq G$ over the same vertex set $V$ such that for any two vertices $u,v\in V$, 
    $$
        d_G(u,v)\leq  d_H(u,v)\leq t\cdot d_G(u,v).
    $$
The parameter $t$ is called the \emph{stretch} of the spanner.
\end{definition}

\begin{definition}[Hop-bounded Emulators]\label{definition:emulator}
    Given a graph $G=(V,E)$, a \emph{$t$-stretch $\beta$-hopbound emulator} is an edge set $H$ over the vertices $V$ such that for any two vertices $u,v\in V$,  
    $$d_G(u,v) \leq d_{H}(u,v)\leq d_{H}^{(\beta)}(u,v)\leq t\cdot d_G(u,v).
    $$
\end{definition}

\begin{definition}[$k$-Nearest Hopsets]\label{definition:hopsets}
    Given a graph $G = (V, E)$, a \emph{$k$-nearest $\beta$-hopbound hopset} is an edge set $H$ over the vertices $V$ such that for any two vertices $u, v \in V$, 
    $$
        d_{G \cup H}(u, v) = d_G(u, v),
    $$
    i.e., distances are not shortened. Moreover, if $u \in N_{G,k}(v)$ or $v \in N_{G,k}(u)$, then
    $$
        d_{G \cup H}^{(\beta)}(u, v) = d_G(u, v).
    $$
    Here, $G \cup H = (V, E \cup H)$, where the smallest-weight edge is taken in the case of parallel edges.
\end{definition}

\subsection{Useful Tools}

We identify machines in the MPC model by their integer IDs in $\{1, 2, \ldots, M\}$, where $M$ is the total number of machines. We refer to a set of machines as a \emph{contiguous segment} if their IDs form a consecutive sequence.

We use the following tools for various purposes throughout this work. The first tool is a constant-round $O(t)$-spanner construction in the heterogeneous MPC model, which also trivially extends to the near-linear MPC model. 

\begin{lemma}[Spanners in Heterogeneous MPC~\cite{dory2021},\cite{fischer2022massively}]\label{lemma:near-linear-spanner}
Let $G = (V, E)$ be a weighted graph with $n$ vertices and let $t \leq \log n$ be a parameter. In the heterogeneous MPC model with $\tilde{O}(|E|)$ total memory, there exists a randomised algorithm that, w.h.p., constructs an $O(t)$-spanner for $G$ of size $\tilde O(n^{1+1/t})$ in $O(1)$ rounds.
\end{lemma}

The next tool comprises four subroutines that will be useful in the design of our MPC algorithms. At a high level, each subroutine involves constructing a tree with branching factor $O(n^\gamma)$ and depth $O(1/\gamma)$ over a contiguous segment of machines. %A key requirement is that these subroutines are executed on a contiguous segment of machines.

\begin{lemma}[Basic MPC Tools~\cite{din2020,goodrich2011}]\label{lemma:mpc-toolbox}
    Consider the sublinear MPC model with $\tilde O(|E|)$ total memory and $O(n^\gamma)$ memory per machine. Suppose we have a contiguous segment of machines and a set of comparable items $S=\set{x_1,x_2,\ldots,x_n}$ stored across these machines. The following tasks can be accomplished in $O(1)$ rounds: 
    \begin{description}
        \item[Sorting:] Sort the items in $S$ across the contiguous segment of machines.
        \item[Find Minimum:] Find the minimum value in $S$. 
        \item[Prefix Sum:] If $S$ is a set of numbers, compute the value $p_i=\sum_{j=1}^{i}x_j$ for each $x_i\in S$.
        \item[Broadcast:] Broadcast a message of size $b\in O(n^{\gamma})$ across the contiguous segment of machines.
    \end{description}
\end{lemma}

\subsection{A Connectivity Assumption}
In this work, we will implicitly assume throughout that the input graph to our distance oracle algorithms is \emph{connected}. To justify this assumption, we show that a distance oracle for disconnected graphs can be reduced to a distance oracle for connected graphs. Specifically, we do this by connecting high-weight edges between different connected components. 

\begin{lemma}\label{lemma:connectivity-reduction}
    Let $G=(V,E)$ be a weighted graph with $n$ vertices, let $T\in\tilde \Omega(|E|)$, $R:=R(n)$, $Q:=Q(n)$ be functions, and let $a\in \poly(n)$. Suppose that in the heterogeneous MPC model with $T$ total memory, there exists a randomised algorithm $\mathcal A$ that, when $G$ is connected, w.h.p.~constructs an $a$-approximate SSSP oracle in $O(R)$ rounds which answers queries in $O(Q)$ rounds.

    Then in the same model of computation, there exists another randomised algorithm $\mathcal A'$ that, \underline{regardless if $G$ is connected or disconnected}, w.h.p.~constructs an $a$-approximate SSSP oracle in $O(R + 1)$ rounds which answers queries in $O(Q +1)$ rounds.
\end{lemma}
\begin{proof}
    We suppose the existence of $\mathcal A$ and describe the algorithm $\mathcal A'$. For readability, we describe only the essential details.

    Given the input graph $G=(V,E)$, the near-linear-space machine $M^*$ first chooses an arbitrary vertex $z\in V$. The entire set of incident edges on $z$ is collected onto $M^*$ in $O(1)$ rounds using the tools from \Cref{lemma:mpc-toolbox}. Also using the same set of tools, we find the heaviest edge weight $W$ of $G$ in $O(1)$ rounds which is communicated to $M^*$. Locally, $M^*$ now adds the following additional edges to the graph $G$: For each vertex $v\in V$ such that $\set{z,v}\notin E$, add the edge $\set{z,v}$ with weight $a(n)\cdot W(n-1) + 1$. This constructs a new weighted graph $G'$. Crucially, $G'$ is also a connected graph.

    We run the algorithm $\mathcal A$ on $G'$ (instead of $G$). This requires $O(R)$ rounds and constructs an $a(n)$-approximate SSSP oracle $\mathcal O_{G'}$ for $G'$ which answers queries in $O(Q)$ rounds. For a query vertex $u$, let $\mathcal O_{G'}(u)[v]$ be the $a(n)$-approximate distance from $u$ to $v$ returned by $\mathcal O_{G'}$. We claim that $u$ and $v$ are connected in $G$ if and only if $\mathcal O_{G'}(u)[v] \leq a(n)\cdot W(n-1)$. Using this claim, we adapt $\mathcal O_{G'}$ into an SSSP oracle $\mathcal O_G$ for the original graph $G$ as follows:
    \begin{equation*}
    \mathcal O_G(u)[v] := 
        \begin{cases}
            \mathcal O_{G'}(u)[v] & \text{if } \mathcal O_{G'}(u)[v] \leq a(n)\cdot W(n-1)\\
            \infty & \text{otherwise.}
        \end{cases}
    \end{equation*}
    It is quick to verify that $\mathcal O_G$ satisfies the following: (1) it requires at most $O(R + 1)$ rounds of preprocessing, and (2) it is an $a(n)$-approximate SSSP oracle for $G$ that answers queries in $O(Q + 1)$ rounds.

    It now suffices to justify the earlier claim. If $u$ and $v$ are connected in $G$, we then have that $d_{G'}(u,v)=d_G(u,v) \leq W(n-1)$ (as there are at most $n - 1$ edges in any shortest path between $u$ and $v$). Therefore, we must have that $\mathcal O_{G'}(u)[v] \leq a(n)\cdot d_{G'}(u,v) \leq a(n)\cdot W(n-1)$. Conversely, suppose that $u$ and $v$ are not connected in $G$. By the definition of $G'$, we have that $d_{G'}(u,v) \geq a(n)\cdot W(n-1) + 1$. Therefore, $\mathcal O_{G'}(u)[v] \geq d_{G'}(u,v) \geq a(n)\cdot W(n-1) + 1 > a(n) \cdot W(n-1)$. The claim follows.
\end{proof}

\section{Technical Overview}\label{section:tech-overview}
We first provide an informal, high-level overview of our SSSP oracle construction. Given the input graph $G$ with $n$ vertices,
\begin{enumerate}
    \item Build an $O(1)$-spanner graph $G'$ of $G$.
    \item Determine approximate $k$-nearest neighbourhoods for all vertices in $G'$.
    \item Transform the approximate $k$-nearest neighbourhoods into exact $k$-nearest neighbourhoods, building the $k$-nearest edge set $E_k$.
    \item Using $E_k$, build the \emph{skeleton graph} $G'_S$ for $G'$. 
    \item Build an $\tilde O(n)$-sized constant-stretch spanner $F$ for $G'_S$, which is stored on a near-linear memory machine.
    \item $F$ and $E_k$ together act as our distance oracle for all future distance queries.
\end{enumerate}

We present the technical overview as follows. We first
describe the three ``building blocks'' in \Cref{section:building-blocks}, which cover Steps 3 and 4 above. We then give a more precise technical description of the oracle construction in \Cref{section:SSSP-oracle-overview}, with a particular focus on the construction of approximate $k$-nearest neighbourhoods (Step 2 above) for the heterogeneous MPC regime. We conclude with a comparison with prior work in \Cref{subsec:comparison-with-prior-work}.

\subsection{Building Blocks}\label{section:building-blocks}
The first of the building blocks is a skeleton graph construction (\Cref{subsec:skel}), which forms the foundation of our SSSP oracle. This is supported by two additional tools: A $k$-nearest hopset construction and an algorithm for extracting $k$-nearest neighbourhoods (\Cref{subsec:knear}).

\subsubsection{Skeleton Graphs}\label{subsec:skel} 

Informally, in the context of distance computation, a skeleton graph $G_S=(V_S,E_S)$ is a weighted graph over a subset $V_S \subseteq V$ of vertices such that approximate distances in the original graph $G$ can be efficiently inferred from the approximate distances in $G_S$. 

In this work, we consider a specific skeleton graph definition from \citet{bui2024improved}, which is based on the concept of $k$-nearest neighbourhoods. Recall that the $k$-nearest edge set $E_k$ is defined in \Cref{sect:prelim}.

% Suppose that together with the input graph $G=(V,E)$, we also have the $k$-nearest edge set $E_k$ for some chosen parameter $k$. Roughly speaking, the skeleton graph $G_S=(V_S,E_S)$ is a weighted graph with two important properties:
% \begin{itemize}
%     \item The number of vertices in $V_S$ is (roughly) a $1/k$-fraction of the number of vertices in $V$. 
%     \item For any vertex pair $u,v\in V$, there is a path in $E_k\cup E_S$ that gives a constant-approximation for the distance $d_G(u,v)$.
% \end{itemize}
% These two properties were also discussed in the prior work by \cite{bui2024improved}. However in our work, we also use another important property (which was glossed over in the previous work): 
% \begin{itemize}
%     \item Consider the path in $E_k\cup E_S$ which witnesses the constant-approximation: The first and last edge of this path belong to $E_k$, and the remaining edges belong to $E_S$. 
% \end{itemize}

\begin{restatable}[Skeleton Graphs~\cite{bui2024improved}]{definition}{skeletonGraphDef}\label{definition:skeleton-graph}
Let $G = (V, E)$ be a weighted graph with $n$ vertices, and let $k > 1$ be an integer.  
A \emph{skeleton graph} $G_S = (V_S, E_S)$ is a weighted graph on a vertex set
\[
V_S \subseteq V  \quad \text{with} \quad |V_S| \leq \frac{C_S \cdot n \log k}{k},
\]
where $C_S > 0$ is a universal constant. This is subject to the following conditions:
\begin{itemize}
    \item (\emph{No shrinking distances}) For any pair of vertices $u, v \in V_S$, we have $d_{G_S}(u, v) \geq d_G(u, v)$.

    \item (\emph{Clusters as good approximations}) For all $u,v\in V$ such that $u \notin N_k(v)$ and $v \notin N_k(u)$, there exists a path $P$ in $E_k \cup E_S$ from $u$ to $v$ satisfying:
    \begin{description}
        \item[(P1)] The length of $P$ is a 7-approximation of $d_G(u, v)$.
        \item[(P2)] The first and last edges of $P$ belong to $E_k$, while all intermediate edges belong to $E_S$.
    \end{description}

    \item(\emph{Small edge weights}) If the edge weights of $G$ are bounded above by a function $f(n)$, then the edge weights of $G_S$ are bounded above by $g(n) := 3k \cdot f(n)$.

    \item(\emph{Connectivity}) If $G$ is connected, then so is $G_S$.
\end{itemize}
\end{restatable}

Intuitively, the path $P$ in the above definition enables us to derive $O(1)$-approximate distances in the original graph $G$ from $O(1)$-approximate distances in the skeleton graph $G_S$. The approximate distances in $G_S$ can be efficiently obtained via spanners constructed using \Cref{lemma:near-linear-spanner}. In particular, when $k \in n^{\Theta(1)}$, the entire spanner can be stored on a machine with near-linear memory, allowing us to obtain the desired $O(1)$-approximate SSSP oracle.

Following the definition of skeleton graphs, two questions remain: First, how do we construct the edge set $E_k$; second, how do we build a skeleton graph? The former question requires the two remaining building blocks. We address the latter question in the discussion below.

\paragraph{Existing Construction.} We begin by reviewing the construction of \citet{bui2024improved} in the congested clique model. The first step is to compute, via random sampling, a hitting set of size $O\left(\frac{n \log k}{k}\right)$ for the $k$-nearest neighbourhoods $N_k(v)$ for all $v \in V$ (i.e., a set that intersects every $k$-nearest neighbourhood). These sampled vertices serve as cluster centers, and each vertex $v \in V$ joins the cluster of its nearest center. Let $c(v)$ denote the cluster center assigned to vertex $v$. The set of cluster centers forms the vertex set of the skeleton graph.

To construct the edge set of the skeleton graph, consider all paths of the form
\[
c(u) \rightsquigarrow u \rightsquigarrow t \rightarrow v \rightsquigarrow c(v),
\quad \text{where} \quad t \in N_k(u) \quad \text{and} \quad \set{t, v} \in E.
\]
For each such path, add the edge $\{c(u), c(v)\}$ to the skeleton graph, with weight equal to the total length of the path. In the case of parallel edges, retain the edge with the minimum weight.

While this construction is sufficient for the congested clique algorithm of \citet{bui2024improved}, it is not suitable for the MPC model, as the total number of edges can reach $O(n^2 / k^2)$ in the worst case, exceeding the total memory available. 

Namely, consider the scenario where a vertex $t$ belongs to the $k$-nearest neighbourhood of every other vertex and is adjacent to all of them. In this case, $t$ leads to an edge between every pair of cluster centers, resulting in a total of $O\left(|V_S|^2\right) \subseteq O((n \log k)^2 / k^2)$ edges in the skeleton graph. Moreover, the construction of \citet{bui2024improved} relies on a sparse matrix multiplication algorithm that critically exploits the high total bandwidth of the congested clique model (i.e. $O(n^2)$ total memory across all machines), and therefore does not extend to the MPC model.

\paragraph{Modified Construction.} We address the issue by modifying the construction to directly include any vertex $t \in V$ that would otherwise induce an excessive number of edges in the skeleton graph. Specifically, for each such $t$ and for each candidate path of the form $c(u) \rightsquigarrow u \rightsquigarrow t \rightarrow v \rightsquigarrow c(v)$, rather than adding the edge $\{c(u), c(v)\}$ to the skeleton graph, we instead add the edges $\{t, c(u)\}$ and $\{t, c(v)\}$. 

This modification ensures that the total number of edges in the skeleton graph is bounded by $\tilde{O}(|E|k^2)$. We also provide an MPC implementation that runs in $O(1)$ rounds. Crucially, we show that this change preserves all structural guarantees required by \Cref{definition:skeleton-graph}. The construction is formalized in the following lemma.

\begin{restatable}[Skeleton Graph Construction]{lemma}{skeletonGraphImp}\label{lemma:skeleton-implement}
Let $G = (V, E)$ be a connected, weighted graph with $n$ vertices, and let $k > 1$ be an integer. Suppose the $k$-nearest edge set $E_k$ for $G$ is given. If $k^2 \in O(n^\gamma)$, then in the sublinear MPC model with $\tilde{O}(|E|k^2)$ total memory, there exists a randomised algorithm that w.h.p.~constructs a skeleton graph $G_S = (V_S, E_S)$ in $O(1)$ rounds.
\end{restatable}

We note that in all applications of \Cref{lemma:skeleton-implement} in this paper, we do not apply it on the input graph directly, but on a sparse $O(1)$-spanner of the input graph -- this reduces the total memory requirement from $\tilde{O}(|E|n^{\gamma/2})$ to $\tilde{O}(|E|+ n^{1+\gamma})$.

\subsubsection{Constructing the $k$-Nearest Edge Set}\label{subsec:knear}

We now describe how to construct the $k$-nearest edge set $E_k$. However, instead of constructing $E_k$ directly, we will show how to transform any given knowledge of the \emph{approximate} $k$-nearest neighbourhoods into the \emph{exact} $k$-nearest edge set $E_k$.

\paragraph{Constructing the $k$-Nearest Hopset.} 
The first step of the transformation is to construct a $k$-nearest hopset, based on an MPC implementation of the corresponding congested clique algorithm by \citet{bui2024improved}. The algorithm starts with an $a$-\emph{approximate} $k$-nearest neighbourhood $\tilde N_k(u)$ for each vertex $u \in V$. By requesting the $k$ shortest edges for each vertex in $\tilde N_k(u)$, vertex $u$ can recover exact distances to a sufficiently large fraction of its true $k$-nearest neighbourhood $N_k(u)$. Shortcutting paths using this distance information yields a $k$-nearest hopset with hopbound $O(a \log d)$, where $d$ is the weighted diameter of the graph.

% \paragraph{Constructing the $k$-Nearest Hopset.} 
% The first step of the transformation is to construct a $k$-nearest hopset, which is obtained by an MPC implementation of the corresponding congested clique algorithm by \citet{bui2024improved}, which works as follows. Start with an $a$-\emph{approximate} $k$-nearest neighbourhood $\tilde N_k(u)$ for every vertex $u\in V$. By simply requesting the $k$ shortest edges for every vertex in $\tilde N_k(u)$, the vertex $u$ is able to know the exact distances to a sufficiently large fraction of vertices in $N_k(u)$ i.e.,~the \emph{exact} $k$-nearest neighbourhood. Short-cutting paths using these exact distances then yields a $k$-nearest $O(a\log d)$-hopbound hopset, where $d$ is the weighted diameter of the graph. 

%While this construction is not novel, implementing this construction in the MPC is not immediately obvious. Namely, we introduce two new MPC tools -- the duplication and labelling lemmas in \Cref{appendix:mpc-tools} -- to aid in the execution. Using these tools, we show that this construction can be done in $O(1)$ rounds. We summarise this in the following lemma.

\begin{restatable}[$k$-Nearest Hopset]{lemma}{kNearestHopset}\label{lemma:k-nearest-hopset}
Let $G = (V, E)$ be a connected, weighted graph with $n$ vertices, and let $k > 0$ be an integer. Suppose that for every vertex $u \in V$, we are given the set $\{(u, t) : t \in \tilde N_k(u)\}$, where $\tilde N_k(u)$ is an $a$-approximate $k$-nearest neighbourhood of $u$. If $k^2 \in O(n^\gamma)$, then in the sublinear MPC model with $\tilde{O}(|E| + nk^2)$ total memory, there exists a deterministic algorithm that constructs a $k$-nearest $O(a\log d)$-hopbound hopset $H$ in $O(1)$ rounds, where $d$ is the weighted diameter of $G$.
\end{restatable}

\paragraph{Constructing the $k$-Nearest Edge Set.}  
Building on the $k$-nearest hopset construction from \Cref{lemma:k-nearest-hopset}, we now turn to the problem of \emph{extracting} the $k$-nearest neighbourhoods. Specifically, given a graph $G = (V, E)$ and a hopset with hopbound $h$, the goal is to compute the $k$-nearest $h$-hopbound neighbourhoods $N_k^{(h)}(u)$ for all vertices $u \in V$ in $O(\log h)$ rounds in the MPC model.

This task has been studied in the congested clique model and appears as a key subroutine in several APSP algorithms~\cite{bui2024improved,censorhillel2019fast,dory2022exponentially}. It can be solved via a standard \emph{graph exponentiation} approach. Initially, each vertex $u$ knows its 1-hop $k$-nearest neighbourhood $N_k^{(1)}(u)$. In iteration $i$, each vertex $u$ collects the sets $N_k^{(2^{i-1})}(t)$ for all $t \in N_k^{(2^{i-1})}(u)$, which is sufficient to compute $N_k^{(2^{i})}(u)$. Repeating this process for $O(\log h)$ iterations yields the desired $k$-nearest $h$-hop neighbourhoods.

We show that each iteration can be implemented in $O(1)$ rounds in the sublinear MPC model, and thus the overall procedure completes in $O(\log h)$ rounds. The result is summarised in the following lemma.

\begin{restatable}[Extracting $k$-Nearest Neighbourhoods]{lemma}{exactKNearest}\label{lemma:exact-k-nearest}
Let $G = (V, E)$ be a connected, weighted graph with $n$ vertices, and let $k > 0$ and $r > 0$ be integers, setting $h = 2^r$. Suppose that $k^2 \in O(n^\gamma)$. Then, in the sublinear MPC model with $\tilde{O}(|E| + nk^2)$ total memory, there exists a deterministic $O(r)$-round algorithm that constructs the set
\[
\left\{ \left(u, t, d_G^{(h)}(u, t)\right) : u \in V \ \text{and} \ t \in N_k^{(h)}(u) \right\}.
\]
\end{restatable}

In this work, we apply the above lemma in two contexts. First, in conjunction with the $k$-nearest hopset from \Cref{lemma:k-nearest-hopset}, it enables an efficient construction of the exact $k$-nearest edge set $E_k$, as previously discussed. Second, it allows us to extract \emph{$c$-approximate} $k$-nearest neighbourhoods from a given \emph{$c$-stretch} $h$-hopbound emulator.

We remark that the round complexity of \Cref{lemma:exact-k-nearest} constitutes one of the main bottlenecks in our overall algorithm. In particular, an improved algorithm for \Cref{lemma:exact-k-nearest} would directly lead to a faster implementation in the near-linear MPC model.

% In the work, we use this lemma in two ways: The first is to use it in conjunction with the $k$-nearest hopset in \Cref{lemma:k-nearest-hopset} to yield a fast construction for the $k$-nearest edge set $E_k$, as already discussed. The second is to use it to extract out $c$-approximate $k$-nearest neighbourhoods from a given a $c$-stretch $h$-hopbound emulator. We remark the round complexity of \Cref{lemma:k-nearest-hopset}  is the  bottleneck in the overall round complexity. A faster algorithm here will yield a faster implementation of the overall algorithm. 

\subsection{Building the SSSP Oracle}\label{section:SSSP-oracle-overview}
In this section, we provide a technical overview of our SSSP oracle, organised in two parts. In \Cref{subsect:overview1}, we describe how to construct the oracle assuming access to approximate $k$-nearest neighbourhoods. This part follows by combining the previously introduced building blocks in a straightforward manner. In \Cref{subsect:overview2}, we address the main technical challenge: efficiently computing approximate $k$-nearest neighbourhoods in the heterogeneous MPC model. Our approach is based on constructing an emulator via a hierarchical application of the skeleton graph construction.

We remark that in the following discussion, we hide the dependence on $c$ and $\gamma$ in the stretch terms for ease of presentation. This is justified as we consider both $c$ and $\gamma$ as constants. We treat this more carefully in the full proofs in \Cref{section:SSSP-oracles}.

\subsubsection{From Approximate $k$-Nearest Neighbourhoods to SSSP Oracles}\label{subsect:overview1} 
Using \Cref{lemma:k-nearest-hopset,lemma:exact-k-nearest}, we construct the set $E_k$ from the given approximate $k$-nearest neighbourhoods and apply \Cref{lemma:skeleton-implement} to build a skeleton graph $G_S = (V_S, E_S)$. We then compute an $O(1)$-spanner of $G_S$ and store it on a machine with near-linear memory. Together, the spanner and $E_k$ yield the desired $O(1)$-approximate SSSP oracle.

We now describe the setup in more detail. Suppose we are given $a$-approximate $k$-nearest neighbourhoods $\tilde N_k(u)$ for every vertex $u \in V$, where $\log a = O(\log\log n)$ and $k = n^{c/4}$ for some constant $c > 0$. We also assume $|E| = O(n^{1 + c/2})$, which can be ensured by replacing $G$ with an $O(1)$-spanner, which is acceptable since it only stretches the distances by a constant factor.

% we use $k=n^{c/4}$ for some constant $c > 0$, and suppose we have $a$-approximate $k$-nearest neighbourhoods $\tilde N_k(u)$ for every vertex $u\in V$, such that $O(\log a)\subseteq O(\log\log n)$. Assume that $G$ has a polynomially-bounded weighted diameter. Additionally, suppose that the $G$ has size $O(n^{1+c/2})$ -- this assumption can be relaxed by an initial sparsification of $G$ using a spanner.
\begin{enumerate}
    \item Construct a $k$-nearest $O(a \log n)$-hopbound hopset $H$ for $G$ using \Cref{lemma:k-nearest-hopset}, which takes $O(1)$ rounds.

    \item Apply \Cref{lemma:exact-k-nearest} to the graph $G \cup H$, setting the parameter $r$ so that $2^r = h \in O(a \log n)$ is at least the hopbound of the hopset. This yields the exact $k$-nearest neighbourhoods for all vertices in $O(\log (a \log n)) = O(\log\log n)$ rounds and also the edge set $E_k$.

    \item Using $E_k$, construct a skeleton graph $G_S = (V_S, E_S)$ for $G$ via \Cref{lemma:skeleton-implement} in $O(1)$ rounds.

    \item For a constant $t \geq 4/c$, construct an $O(t)$-spanner $F$ of $G_S$ with $\tilde O(|V_S|^{1 + 1/t}) \subseteq \tilde O(n)$ edges using \Cref{lemma:near-linear-spanner} in $O(1)$ rounds. Collect the spanner on a  machine $M^*$ with near-linear memory.
\end{enumerate}
Overall, the algorithm completes in $O(\log\log n)$ rounds and uses $\tilde{O}(n^{1+c})$ total memory.

\paragraph{SSSP Queries.} To compute an $O(1)$-approximate SSSP from a source vertex $u \in V$, we consider two cases:
\begin{description}
    \item[Case 1: $v \in N_k(u)$ or $u\in N_k(v)$.] Return the exact distance $\omega_{E_k}(u, v) = d_G(u, v)$.
    
    \item[Case 2: $v \notin N_k(u)$ and $u\notin N_k(v)$.] Return the minimum length over all paths of the form: $u \rightarrow s \in V_S$ via an edge in $E_k$, followed by a path $s \rightsquigarrow s' \in V_S$ in the spanner $F$, and then $s' \rightarrow v$ via an edge in $E_k$.
\end{description}
In Case 2, the returned value is an $O(1)$-approximation of the true distance $d_G(u, v)$, guaranteed by the properties of the skeleton graph. Moreover, the path lengths to all vertices $v \in V$ can be computed \emph{simultaneously} in $O(1)$ rounds. In other words, using the edge set $E_k$ and the spanner $F$ of $G_S$, we can answer $O(1)$-approximate SSSP queries in $O(1)$ rounds.

\subsubsection{Computing Approximate $k$-Nearest Neighbourhoods}\label{subsect:overview2}

 %In this work, we solve this challenge in two different ways, depending on the MPC regime being considered. 
We now describe the construction of the required $a$-approximate $k$-nearest neighbourhoods.

\paragraph{The Near-Linear MPC Model.} In the near-linear MPC model, the problem can be solved easily by constructing an $O(\log n)$-spanner with $\tilde O(n)$ edges using \Cref{lemma:near-linear-spanner}. This spanner can then be broadcast to all machines, allowing each machine to compute an $O(\log n)$-approximate $k$-nearest neighbourhood for each vertex locally. Since the spanner construction takes only $O(1)$ rounds, the overall round complexity for constructing the SSSP oracle is $O(\log\log n)$.

\paragraph{The Heterogeneous MPC Model.} The spanner-based approach used in the near-linear MPC model does not extend to the heterogeneous regime, as we cannot afford to store a linear-sized spanner on each machine. To overcome this, we construct a $2^{O(\log\log n)}$-stretch, $O(\log\log n)$-hopbound emulator via a recursive application of skeleton graph constructions. This emulator is built in $O(\log^2 \log n)$ rounds and constitutes the main bottleneck in the overall SSSP oracle construction. 

Once the emulator is constructed, we apply \Cref{lemma:exact-k-nearest} to extract the desired $a = 2^{O(\log\log n)}$-approximate $k$-nearest neighbourhoods, for $k$ being a small polynomial in $n$, using an additional $O(\log\log\log n)$ rounds. This suffices for our purposes.

% \paragraph{The Emulator Construction.} The key idea behind the construction is straightforward: We recursively construct skeleton graphs on-top of one another, while collecting the edge set $E_k$ into our emulator during each iteration. The high-level description is as follows: We first start with an empty edge set $H=\emptyset$. At the start of the $i$th iteration, we begin with the graph $G_i$, (where $G_0=G$).
% \begin{enumerate}
%     \item We find approximate $k$-nearest neighbourhoods for $G_i$.
%     \item We find the $k$-nearest edge set $E_k$ for $G_i$.
%     \item Add $E_k$ to $H$.
%     \item Construct the skeleton graph $G_S$ from $G_i$. 
%     \item Set $G_{i+1}:=G_S$.
% \end{enumerate}
% With each iteration, the number of vertices reduces (roughly) by a factor of $k$. We repeat this until the final graph is small enough such that all-pair distances can be computed locally on a single machine, and with one last step, we add these all-pairs distances into $H$ as well. The final edge set $H$ then becomes our emulator. The properties of our emulator depends mainly on the number of iterations: Namely, if we the algorithm runs for $T$ iterations, we show that we obtain a $2^{O(T)}$-stretch, $O(T)$-hopbound emulator for $G$ by the properties of the skeleton graph.\yijun{The discussion above this point is somewhat unmotivated... I think the subsection can start right here without the discussion above.} 

% However, the important question still remains: How do we construct approximate $k$-nearest neighbourhoods in the heterogeneous MPC? 
\paragraph{The Emulator Construction.} The key idea behind the emulator construction is that the spanner-based approach used in the near-linear regime still works in the heterogeneous regime when $k$ is \emph{polylogarithmic}. Once the spanner is stored on the near-linear memory machine $M^*$, it can locally compute the desired approximate $k$-nearest neighbourhoods for all vertices using the spanner. Since $k$ is polylogarithmic, the total output size is $nk \in \tilde{O}(n)$ and fits within the local memory of $M^*$. 

We apply the skeleton graph construction to $G$ with parameter $k$, obtaining a slightly smaller skeleton graph $G_1$ with $O\left(\frac{n \log k}{k}\right)$ vertices. For simplicity, in this technical overview, we assume the bound is $O(n/k)$.

Next, we apply the same skeleton graph construction  to $G_1 = (V_1, E_1)$. The reduced size of $G_1$ allows us to increase the parameter to $k^2$, yielding a new skeleton graph with $O\left(|V_1| / k^2\right) = O(n / k^3)$ vertices. This, in turn, enables us to increase $k$ even further in the next iteration. Due to the growth rate of the parameter $k$, the recursion terminates after $O(\log\log n)$ iterations. We show that the union of the edge sets $E_k$ across all levels of the recursion forms the desired emulator.

%While it is not clear how we might find the approximate $k$-nearest neighbourhoods for some \emph{polynomial} $k$, if $k$ is \emph{logarithmic}, this becomes more feasible: We first store an $O(\log n)$-spanner for $G=G_0$ on the single near-linear machine $M^*$, which subsequently determines the $k\in O(\log n)$-nearest neighbourhood for every vertex \emph{in this spanner}. We now can apply the skeleton graph construction and obtain a slightly smaller skeleton graph $G_1$ of (roughly) $O(n/\log n)$ vertices. This smaller graph allows us to increase $k$ to $O(\log^2 n)$, obtaining a new skeleton graph of $O(n/\log ^3n)$ vertices, allowing us to increase $k$ even further and so on. 

%We now describe the idea more concretely: We begin with $H:=\emptyset$ and work in iterations. At the start of the $i$th iteration, we have the graph $G_i=(V_i,E_i)$ and the parameter $k_i$ (where $G_0=G$ and $k_0=\log n$). Assume that $k_i\cdot |V_i|\in O(n\log n)$, and that the weighted diameter of $G_i$ is bounded by $\poly(n)$. 

We now describe the idea more concretely. We initialize $H := \emptyset$ and proceed in iterations. At the beginning of the $i$th iteration, we maintain a graph $G_i = (V_i, E_i)$ and a parameter $k_i$, where $G_0 = G$ and $k_0 = \log n$. We assume that $k_i \cdot |V_i| \in O(n \log n)$ and that the weighted diameter of $G_i$ is bounded by $\poly(n)$.
\begin{enumerate}
    \item We begin by sparsifying $G_i$ by constructing an $O(1)$-stretch spanner $G_i'$ using \Cref{lemma:near-linear-spanner}, to ensure that the number of edges remains bounded.

    \item Next, we construct an $O(\log n)$-stretch spanner $S$ of $G_i'$ with $\tilde O(n)$ edges using \Cref{lemma:near-linear-spanner}. We store $S$ on the  machine $M^*$ with near-linear memory. Since $k_i \cdot |V_i| \in O(n \log n)$, $M^*$ can locally compute the $k_i$-nearest neighbourhoods for all vertices in the spanner $S$. This yields $a = O(\log n)$-approximate $k_i$-nearest neighbourhoods for all vertices in $V_i$, and takes $O(1)$ rounds.

    \item As in \Cref{subsect:overview1}, we apply \Cref{lemma:exact-k-nearest,lemma:k-nearest-hopset} to construct the edge set $E_{k_i}$ in $O(\log(a \log n)) \subseteq O(\log\log n)$ rounds.

    \item We add $E_{k_i}$ to the emulator edge set $H$.

    \item Finally, we construct the skeleton graph $G_S$ for $G_i'$ and set $G_{i+1} := G_S$ and $k_{i+1} \in O(k_i^2)$.
\end{enumerate}

We continue the process until the graph becomes small enough that all-pairs distances can be computed locally on a single machine. A minor technical detail is that once $k_i$ reaches a small polynomial in $n$, we keep $k_i$ fixed in all subsequent iterations to ensure that the local memory requirements of the building blocks are satisfied. Since there are $O(\log\log n)$ iterations, where each iteration takes $O(\log\log n)$ rounds, the overall emulator construction completes in $O(\log^2 \log n)$ rounds.

\subsection{Comparison with Prior Work}\label{subsec:comparison-with-prior-work}
\paragraph{Comparison with \cite{bui2024improved}.} While both our algorithm and that of \citet{bui2024improved} rely on skeleton graphs, they are used in different ways. We apply the skeleton graph construction recursively to build an emulator with small stretch and hopbound. In contrast, the $O(\log\log\log n)$-round congested clique algorithm of \citet{bui2024improved} uses a skeleton graph with $k \in 2^{\Theta(\sqrt{a})}$ as an intermediate step to improve an $a$-approximate APSP to an $O(\sqrt{a})$-approximate one in $O(1)$ rounds -- starting with $a=O(\log n)$, they repeat this for $O(\log\log\log n)$ iterations to obtain constant-approximate APSP.

Their approach is infeasible in our setting for two main reasons. First, constructing such a skeleton graph in $O(1)$ rounds requires a fast $k$-nearest neighbourhood computation algorithm that leverages the high total bandwidth of the congested clique model. Second, obtaining the $O(\sqrt{a})$-approximation relies on broadcasting an $O(\sqrt{a})$-spanner of the skeleton graph to all machines, which demands linear local memory per machine.

\paragraph{Comparison with \cite{pram-subemulator}.} In the work of \citet{pram-subemulator}, they demonstrated a similar recursive construction of a hop-bounded emulator in the \emph{PRAM model of computation} using a structure known as a \emph{subemulator} (as opposed to our use of the skeleton graph).

% \yijun{so now there is no discussion about the fact that we can actually use subemulator instead of skeleton graph? I am ok with it, as it become quite messy if one tries to discuss it, so it is a good idea to avoid... However, we should also expect that most likely we will encounter the same reviewers regardless of the venue we resubmit the paper} 

Although both \citet{pram-subemulator} and our work employ similar recursive constructions of emulators, the motivations for the recursion and the roles played by the resulting emulators differ substantially. In \citet{pram-subemulator}, the emulator is used to construct Bourgain’s embedding, metric tree embeddings, and low-diameter decompositions, whereas in our work it is used to find the required $a$-approximate $k$-nearest neighbourhoods. Moreover, \citet{pram-subemulator} adopt the recursive construction, in which a smaller vertex set permits a larger choice of $k$, to satisfy the work-efficiency constraint. In contrast, we arrive at the same condition to ensure that the $k$-nearest neighbourhoods of all vertices can be stored on a single machine with near-linear memory.

%We continue until we can compute all-pair distances locally on a single machine. A missing technical detail is that once $k_i$ becomes a small polynomial in $n$, we keep the value of $k_i$ the same for all remaining iterations, so that the local memory requirement when applying the building blocks is met. 
%As there are $O(\log\log n)$ iterations and each iteration requires $O(\log\log n)$ rounds, the emulator construction finishes in $O(\log^2\log n)$ rounds. 

%Subsequently, with an additional $O(\log\log\log n)$ rounds, we can construct $2^{O(\log\log n)}$-approximate $k$-nearest neighbourhoods for all vertices in $V$ by applying \Cref{lemma:exact-k-nearest}, where $k$ is now some small polynomial in $n$. Therefore, constructing the approximate neighbourhoods requires $O(\log^2\log n)$ rounds in total. Overall, the algorithm for constructing the SSSP oracle in the heterogeneous MPC runs in $O(\log\log n + O(\log^2\log n))= O(\log^2\log n)$ rounds. 

\section{$O(1)$-Approximate SSSP Oracles}\label{section:SSSP-oracles}
In this section, we present the main algorithms of the paper in detail: $O(1)$-approximate SSSP oracles with $O(1)$ query time. The construction requires $O(\log\log n)$ rounds in the near-linear MPC model and $O(\log^2 \log n)$ rounds in the heterogeneous MPC model.

We first state the following lemma, which is a tool to ``convert'' approximate $k$-nearest neighbourhoods to exact $k$-nearest neighbourhoods.

%Since the skeleton graph plays a crucial role in our constructions, we first recall the definition of the skeleton graph below.

%\skeletonGraphDef*

\begin{lemma}[Approximate to Exact $k$-Nearest]\label{lemma:E_k}
    Let $G=(V,E)$ be a connected, weighted graph with $n$ vertices, and let integer $k > 0$ be a parameter. Suppose that for every vertex $u \in V$, we are given the set $\{(u, t) : t \in \tilde N_k(u)\}$, where $\tilde N_k(u)$ is an $a$-approximate $k$-nearest neighbourhood of $u$.
    
    Given that $k^2 \in O(n^\gamma)$, in the sublinear MPC model with $\tilde O(|E| + nk^2)$ total memory, there exists a deterministic algorithm for constructing the $k$-nearest edge set $E_k$ for $G$ in $O(\log\log d+\log a)$ rounds, where $d$ is the weighted diameter of $G$.
\end{lemma}
\begin{proof}
    We first apply \Cref{lemma:k-nearest-hopset} to obtain a $k$-nearest $(Ca\log d)$-hopbound hopset $H$ in $O(1)$ rounds, where $C$ is a suitable constant. We now apply \Cref{lemma:exact-k-nearest} on the graph $G\cup H$, setting the parameter $r$ such that $2^r=h\geq C a\log d$. This constructs the $h$-hop $k$-nearest neighbourhoods $N^{(h)}_{G\cup H, k}(u)$ for all vertices $u\in V$ in $G\cup H$. We recall by the definition of $k$-nearest hopsets (\Cref{definition:hopsets}) that for all vertices $u,v\in V$ such that $v\in N_{G,k}(u)$, we have
    $$
        d_{G\cup H}^{(h)}(u,v) \leq d_{G\cup H}^{(Ca\log d)}(u,v) = d_G(u,v).
    $$
    Moreover, $H$ does not shorten distances in $G$, i.e., $d_{G\cup H}(u,v)=d_G(u,v)$ for vertex pairs $u,v\in V$. This implies that the $k$-nearest neighbourhood of any vertex in $G\cup H$ is the same as that in $G$ (noting that in both cases, we are breaking ties by vertex ID).
    
    Altogether, this implies $N^{(h)}_{G\cup H,k}(u)=N_{G,k}(u)$ for all $u\in V$, and therefore we have equivalently constructed the $k$-nearest edge set $E_k$, as desired.
    
    Constructing the hopset $H$ requires $\tilde O(|E| + nk^2)$ total memory, while extracting out the $k$-nearest neighbourhoods requires $\tilde O(|G\cup H| + nk^2)$ total memory. Therefore, our entire algorithm only requires $\tilde O(|E|+nk^2)$ total memory, and runs in $O(r)=O(\log (a\log d))=O(\log a + \log\log d)$ rounds.
\end{proof}

\subsection{From Approximate $k$-Nearest Neighbourhoods to SSSP Oracles}
In this section, we aim to prove the following lemma. Notice that while the lemma is stated explicitly for the heterogeneous regime, it trivially extends to the near-linear regime with the same total memory and round complexity. In subsequent sections, we show how to relax our input assumptions depending on the MPC regime being considered.

We remark that the restriction on the edge set $|E|\in \tilde O(n^{1+{c/2}})$ is only to ensure that the additional total memory as stated in \Cref{theorem:near-linear} and \Cref {theorem:hetero} is represented as an \emph{additive} term $n^{1+c}$ as opposed to a \emph{multiplicative} term $n^c$, i.e., $\tilde O(|E| + n^{1+c})$ instead of $\tilde O(|E|n^c)$ total memory. Namely by choosing $k=n^{c/4}$ as the parameter for the $k$-nearest neighbourhoods, we ensure that the skeleton graph construction of \Cref{lemma:skeleton-implement} requires only $\tilde O(|E|k^2)=\tilde O(n^{1+c})$ total memory. 

\begin{lemma}\label{lemma-central-construct}
    Let $G=(V,E)$ be a connected, weighted graph with $n$ vertices, and let $c > 0$ be a constant such that $|E|\in \tilde O(n^{1+{c/2}})$. Suppose that for every vertex $u \in V$, we are given the set $\{(u,t) : t \in \tilde N_k(u)\}$, where $\tilde N_k(u)$ is an $a$-approximate $k$-nearest neighbourhood of $u$ with $k=n^{c/4}$.
    
    Given that $c\leq \gamma$, in the heterogeneous MPC model with $\tilde O(n^{1+c})$ total memory, there exists a randomised algorithm that w.h.p.~constructs an $O(1/c)$-approximate SSSP oracle in $O(\log a + \log\log n)$ rounds. Every subsequent SSSP query can then be answered in $O(1)$ rounds.
\end{lemma}

The key idea is to construct the skeleton graph $G_S=(V_S,E_S)$ for the input graph $G$, and subsequently construct an $O(1/c)$-stretch spanner $F$ for $G_S$. We will show by an appropriate choice of parameters that $F$ can be made linear-sized, and therefore can be stored on a single near-linear memory machine. The spanner $F$ and edge set $E_k$ are then sufficient to handle any future SSSP queries in $O(1)$ rounds. 

\subsubsection{Algorithm Description}\label{section:main-construction-algo}
As input, we are given the weighted $n$-vertex graph $G=(V,E)$, whose edge weights are polynomially-bounded (i.e., $\omega_G(u,v)\in \poly(n)$ for all $\set{u,v}\in E)$.

\begin{description}
    \item [Step 1: Get Exact $k$-Nearest.] Denote $k=n^{c/4}$. By the input assumptions, we have $a$-\emph{approximate} $k$-nearest neighbourhoods $\tilde N_k(u)$ for all vertices $u\in V$. Since $k^2=n^{c/2}\leq n^\gamma$, this allows us to find the \emph{exact} $k$-nearest neighbourhoods for all vertices by applying \Cref{lemma:E_k} on $G$. This constructs the $k$-nearest edge set $E_k$ for $G$.
    
    Denoting the weighted diameter of $G$ as $d$, we observe that this step runs in $O(\log a + \log\log d)$ rounds. The total memory required is $\tilde O(|E|+nk^2)=\tilde O(n^{1+c/2})$, which fits within our total memory constraints. 

    \item [Step 2: Construct Skeleton Graph.] We apply \Cref{lemma:skeleton-implement} on $G$ with $k = n^{c/4}$ to obtain the skeleton graph $G_S = (V_S, E_S)$. This step runs in $O(1)$ rounds. We note that the total memory needed is $\tilde O(|E|k^2)=\tilde O(n^{1+c/2}\cdot n^{c/2})=\tilde O(n^{1+c})$, which fits within our total memory constraints. 

    \item [Step 3: Sparsify $G_S$.] We further sparsify $G_S$ into a spanner $F=(V_S,E_F)$. Namely, we apply \Cref{lemma:near-linear-spanner} to obtain an $O(4/c)$-spanner of size $|E_F|\in \tilde O(|V_S|^{1+c/4})$ in $O(1)$ rounds. For now, we may assume that $|E_F|\in \tilde O(n)$; we justify this assumption in \Cref{lemma:main-construct-F-size}.

    \item [Step 4: Storing the Spanner.] We store the spanner $F$ on the single near-linear memory machine $M^*$. This requires $O(1)$ rounds. 
\end{description}
Our final SSSP oracle consists of the following:
\begin{itemize}
    \item The edge set $E_k$, stored arbitrarily across the machines.
    \item The graph $F$, stored entirely on the machine $M^*$.
\end{itemize}
This concludes our construction.

\subsubsection{Correctness}
We first justify our claim that the spanner $F$ is sufficiently small to store in a single near-linear memory machine.
\begin{lemma}\label{lemma:main-construct-F-size}
    For any constant $c > 0$, the size of the graph $F$ is bounded by $|E_F|\in \tilde O(n)$.
\end{lemma}
\begin{proof}
    By the definition of the skeleton graph in \Cref{definition:skeleton-graph}, we may denote the number of vertices in the skeleton graph $G_S$ as $|V_S| \leq (C_S\cdot n\log k)/k=C_S(c/4)\left(n^{1-(c/4)}\cdot \log n\right)$, where $k=n^{c/4}$ and $C_S$ is the constant defined in \Cref{definition:skeleton-graph}. Additionally by Step 3 of the algorithm, we may denote the number of edges in the graph $F$ as $|E_F| \leq C_1\cdot |V_S|^{1+c/4}\log^{C_2}n$, where $C_1,C_2$ are suitable constants. We find that
    \begin{align*}
        |E_F| &\leq  C_1\cdot |V_S|^{1+c/4}\log^{C_2}n\\
        &\leq  C_1\cdot \left[\frac{C_S\cdot c\log n}{4}\right]^{1+c/4}\cdot \left(n^{1-(c/4)}\right)^{1+c/4}\log^{C_2}n\\
        &=  C_1\cdot \left[\frac{C_S\cdot c\log n}{4}\right]^{1+c/4}\cdot n^{1-(c/4)^2}\log^{C_2}n.
    \end{align*}
    Since we treat all other variables as constants except for $n$, we have that $|E_F|\in n\cdot \polylog (n)\subseteq \tilde O(n)$ as desired.
\end{proof}

We next show the round complexity of the algorithm.
\begin{lemma}\label{lemma:main-construct-time}
    The above algorithm runs in $O(\log a + \log\log n)$ rounds.
\end{lemma}
\begin{proof}
    From the algorithm description, we observe that except for Step 1, all remaining steps run in $O(1)$ rounds. It therefore suffices to show that Step 1 runs in $O(\log a + \log\log n)$ rounds. 
    
    Denoting $g(n)\in \poly(n)$ as the upper bound of the edge weights in $G$, we find that the weighted diameter of $G$ is at most $d=(n-1)\cdot g(n)\in \poly(n)$ as any shortest path in $G$ has at most $n-1$ edges. Therefore, Step 1 runs in $O(\log a + \log\log d)=O(\log a + \log\log n)$ rounds as claimed. 
\end{proof}

\paragraph{Answering Queries.} For any given source vertex $u\in V$ and any other vertex $v\in V$, we consider all paths constructed as follows: 
\begin{itemize}
    \item $u\rightarrow x\in V_S$ using an edge in $E_k$,
    \item $x\rightsquigarrow y\in V_S$ using a path in $F$, and lastly,
    \item $y\rightarrow v$ using another edge in $E_k$.
\end{itemize}
Given that $v\notin N_k(u)$ and $u\notin N_k(v)$, we show by properties (P1) and (P2) in \Cref{definition:skeleton-graph} there must exist such a path that witnesses an $O(1/c)$-approximation of the distance $d_G(u,v)$. More concretely, for any source vertex $u\in V$, consider the distance estimate $\delta_u:V\rightarrow \R$, defined as follows: 
$$
    \delta_u(v)=\min_{x,y\in V_S}\left(\omega_{E_k}(u,x)+d_F(x,y)+\omega_{E_k}(y,v)\right).
$$

We first show an upper bound for the estimate $\delta_u$.
\begin{lemma}\label{lemma:main-construct-ub}
    For all vertices $v\in V$ such that $v\notin N_k(u)$ and $u\notin N_k(v)$, the estimate $\delta_u(v)$ is at most $O(1/c)\cdot d_G(u,v)$.
\end{lemma}
\begin{proof}
    We first consider the path $P$ from $u$ to $v$ which has the properties (P1) and (P2) as defined in \Cref{definition:skeleton-graph}, and denote $P$ as $u\rightarrow u^*\rightsquigarrow v^*\rightarrow v$. Additionally, we denote the edges $u\rightarrow u^*$ and $v^*\rightarrow v$  as $e_1$ and $e_2$ respectively, and denote the subpath from $u^*\rightsquigarrow v^*$ in $P$ as $Q$. 
    
    Recall from property (P1) that $\omega(P) \leq 7\cdot d_{G}(u,v)$. Moreover, recall from property (P2) that $e_1,e_2\in E_k$, and that $Q$ is contained entirely in the graph $G_S$. In particular, there must exist a path $Q'$ from $u^*$ to $v^*$ in $F$ such that $\omega(Q')\leq C(4/c)\cdot d_{G_S}(u^*,v^*)$ (where $C$ is a suitable constant). We are therefore able to bound $\delta_u(v)$ as
    \begin{align*}
        \delta_u(v) &\leq \omega(e_1) + \omega(Q') + \omega(e_2) \\
        &\leq \omega(e_1) + C(4/c)\cdot d_{G_S}(u^*,v^*) + \omega(e_2)\\
        &\leq \omega(e_1) + C(4/c)\cdot \omega(Q) + \omega(e_2)\\
        &\leq C(4/c)\cdot (\omega(e_1) + \omega(Q) + \omega(e_2)) \\
        &=C(4/c)\cdot \omega(P)\\
        &\leq 7C(4/c)\cdot d_{G}(u,v)
    \end{align*}
    The desired claim follows.
\end{proof}

We next show how to compute the estimates $\delta_u$. Since the algorithm itself does not require any new techniques, we only emphasise the important points for the sake of presentation. 
\begin{lemma}\label{lemma:main-construct-pathfind}
    There exists a deterministic algorithm for constructing the values $\delta_u(v)$ for all $v\in V$ in $O(1)$ rounds.
\end{lemma}
\begin{proof}
    Computation of the values $\delta_u(v)$ is broken up into two parts: For every vertex $y\in V_S$, we first compute the values
    $$
        \Delta_u(y) = \min_{x\in V_S}\left(\omega_{E_k}(u,x)+d_F(x,y)\right),
    $$
    which describe the length of the path $u\rightarrow x\rightsquigarrow y$. Subsequently, for every vertex $v\in V$, we compute the values
    $$
        \Delta'_u(v)=\min_{y\in V_S}\left(\Delta_u(y)+\omega_{E_k}(y,v)\right),
    $$
    which now fully describe the length of the path $u\rightarrow x\rightsquigarrow y\rightarrow v$. By definition, we observe that $\Delta'_u(v) = \delta_u(v)$ for all vertices $v\in V$. 
    
    We now describe the MPC algorithm: For readability, we describe only the essential details. Below, we recall that $M^*$ refers to the single near-linear memory machine. 
    \begin{enumerate}
        \item For every vertex $x\in V_S$ such that $\set{u,x}\in E_k$, we send the value $(x,\omega_{E_k}(u,x))$ to the machine $M^*$. Observe that there are at most $n$ such values, and therefore we have sufficient local memory in $M^*$. This requires $O(1)$ rounds.
        \item For every vertex $y\in V_S$, $M^*$ locally computes the values $\Delta_u(y)$. This requires no communication rounds.
        \item For every vertex $y\in V_S$ and $v\in V$ such that $\set{y,v}\in E_k$, we label the edge $\set{y,v}$ with the value $\Delta_u(y)$. Using the MPC tools in \Cref{lemma:mpc-toolbox} (namely sorting and broadcasting), this requires $O(1)$ rounds.
        \item For every vertex $v\in V$, we now compute the values $\Delta'_u(v)$. Using the MPC tools in \Cref{lemma:mpc-toolbox}, this also requires $O(1)$ rounds.
    \end{enumerate}
    Overall, we require only $O(1)$ rounds as claimed.
\end{proof}

We now summarise the previous two results in the following lemma.
\begin{lemma}\label{lemma:main-construct-queries}
    Given the edge set $E_k$ stored across the machines and the spanner $F$ stored entirely on the machine $M^*$, we can answer $O(1/c)$-approximate SSSP queries in $O(1)$ rounds.
\end{lemma}
\begin{proof}
    For a given SSSP query $u\in V$, we first construct the values $\delta_u(v)$ for all vertices $v\in V$ in $O(1)$ rounds using \Cref{lemma:main-construct-pathfind}. For every vertex $v\in V$, we separate them into two cases:
    \begin{description}
        \item[Case 1: $v\in N_k(u)$ or $u\in N_k(v)$.] By definition of $E_k$, there must exist an edge $\set{u,v}\in E_k$. It therefore suffices to return the edge weight $\omega_{E_k}(u,v)=d_G(u,v)$.

        \item[Case 2: $v\notin N_k(u)$ and $u\notin N_k(v)$.] We return the value $\delta_u(v)$. By \Cref{lemma:main-construct-ub}, we have that $\delta_u(v)\leq O(1/c)\cdot d_G(u,v)$. Additionally, for any vertices $x,y\in V_S$, we have
        $$
            d_F(x,y)\geq d_{G_S}(x,y)\geq d_G(x,y)
        $$
        by definition of the skeleton graph in \Cref{definition:skeleton-graph}. Thus, we necessarily have that $\delta_u(v) \geq d_G(u,v)$ as well. 
    \end{description}
    In other words, it suffices to return the value $\min(\omega_{E_k}(u,v),\delta_u(v))$ for all vertices $v\in V$, which always gives an $O(1/c)$-approximation of the distance $d_G(u,v)$ as desired.
\end{proof}

\paragraph{The Final Proof.} We summarise this section on correctness with the following proof for \Cref{lemma-central-construct}.
\begin{proof}[Proof of \Cref{lemma-central-construct}]
    We execute the algorithm described in \Cref{section:main-construction-algo} to construct the edge set $E_k$ and the spanner graph $F$, the latter of which is stored on the single near-linear machine $M^*$. By \Cref{lemma:main-construct-time}, this requires $O(\log a + \log\log n)$ rounds. Furthermore, from the algorithm description, observe that the algorithm requires only $\tilde O(n^{1+c})$ total memory. Lastly applying \Cref{lemma:main-construct-queries}, we find that this construction suffices as an $O(1/c)$-approximate SSSP oracle with a query time of $O(1)$ rounds.
\end{proof}

\subsection{The Near-Linear MPC Model}
In this section, we aim to prove \Cref{theorem:near-linear}. Namely, we will describe an algorithm for constructing $O(\log n)$-approximate $k$-nearest neighbourhoods for all vertices, from which a direct application of \Cref{lemma-central-construct} will give our desired claim.

\paragraph{Finding Approximate Neighbourhoods.} In order to construct $O(\log n)$-approximate neighbourhoods, we construct an $O(\log n)$-spanner of near-linear size which can be done in $O(1)$ rounds. The entire spanner can be stored on every machine, which enables us to determine the $k$-nearest neighbourhoods of every vertex \emph{in the spanner} -- these suffice as $O(\log n)$-approximate $k$-nearest neighbourhoods for the \emph{original graph}. We formalise this idea in the following lemma.

\begin{lemma}\label{lemma:near-linear-oracle}
    Let $G=(V,E)$ be a connected, weighted graph with $n$ vertices and let integer $k > 0$ be a parameter. In the near-linear MPC model with $\tilde O(|E| + nk)$ total memory, there exists a randomised algorithm that w.h.p.~constructs the set $\set{(u,t):u\in V\wedge t\in \tilde N_k(u)}$ where $\tilde N_k(u)$ is an $O(\log n)$-approximate $k$-nearest neighbourhood of $u$. This algorithm runs in $O(1)$ rounds. 
\end{lemma}
\begin{proof}
Using \Cref{lemma:near-linear-spanner} with $t=\log n$, we construct an $O(\log n)$-spanner $H$ of size $\tilde O(n)$ in $O(1)$ rounds. We collect the spanner $H$ onto a single machine, which we then broadcast to \emph{every} machine using \Cref{lemma:mpc-toolbox}. In other words, after $O(1)$ rounds, every machine now knows the entire spanner $H$.

What remains is to construct the $k$-nearest neighbourhood $N_k(u)$ of every vertex $u\in V$ in the spanner $H$, which suffices as our $O(\log n)$-approximate $k$-nearest neighbourhoods in $G$. This is done by a delegation of tasks: The machine with ID 1 determines the neighbourhoods for the first $n/k$ vertices, the machine with ID 2 determines the neighbourhoods for the next $n/k$ vertices, and so on. Since the size of each set $|N_k(u)|\leq k$ by definition, this requires only $O(n)$ memory per machine. Moreover, we only need $O(nk)$ total memory to fully compute the sets $N_k(u)$ for all $n$ vertices, and therefore the total memory of $\tilde O(|E|+nk)$ suffices. Additionally, since only local computation is needed, this requires no rounds of communication.
\end{proof}

We are now ready to prove \Cref{theorem:near-linear}. By applying \Cref{lemma:near-linear-oracle} and subsequently \Cref{lemma-central-construct}, we get the desired result.
\begin{proof}[Proof of \Cref{theorem:near-linear}]
    By \Cref{lemma:connectivity-reduction}, we may assume that $G$ is a connected graph. We first apply the spanner construction in \Cref{lemma:near-linear-spanner} (setting $t=2/c$) to construct an $O(2/c)$-spanner $G'$ for $G$ of size $\tilde O(n^{1+c/2})$ in $O(1)$ rounds. We now apply \Cref{lemma:near-linear-oracle} on $G'$ to construct $a=O(\log n)$-approximate $k=n^{c/4}$-nearest neighbourhoods $\tilde N_k(u)$ for all vertices $u\in V$ in $G'$ using another $O(1)$ rounds. Lastly, we apply \Cref{lemma-central-construct}, which constructs an $O(1/c)$-approximate SSSP oracle for $G'$ in $O(\log a + \log\log n) = O(\log\log n)$ rounds. In total, we require $\tilde O(|E|+n^{1+c})$ memory across all steps. 
    
    Since $G'$ is an $O(2/c)$-spanner for $G$, an $O(1/c)$-approximate SSSP oracle for $G'$ directly gives an $O(1/c^2)$-approximate SSSP oracle for $G$, and we are done.
\end{proof}

\subsection{The Heterogeneous MPC Model}\label{subsec:Heterogeneous}
Similar to the near-linear MPC regime, in this section we propose a fast method for constructing approximate $k$-nearest neighbourhoods, which together with \Cref{lemma-central-construct}, gives the final oracle construction. However, instead of constructing an $O(\log n)$-spanner that is distributed to all machines, we instead construct a $(2/\gamma)^{O(\log\log n)}$-stretch $O(\log\log n)$-hopbound emulator, which we show can be done in $O(\log^2\log n)$ rounds. On this emulator we apply \Cref{lemma:exact-k-nearest} to obtain $(2/\gamma)^{O(\log\log n)}$-approximate $k$-nearest neighbourhoods. We will show that this approximation ratio is sufficient for us to apply \Cref{lemma-central-construct} and construct the SSSP oracle using $O(\log^2\log n)$ rounds overall.

\paragraph{High Level Idea.} We iteratively construct a series of skeleton graphs $G=G_0, G_1, G_2, \ldots$ until the final skeleton graph $G_{\ell}$ is small enough that all-pairs distances can be determined entirely locally. The edges of our emulator $H$ then consist of the following:
\begin{itemize}
    \item The $k_i$-nearest edge sets $E_{k_i}$ taken from each graph $G_i$.
    \item The all-pairs distances taken from the final graph $G_{\ell}$.
\end{itemize} 
By choosing an appropriate sequence of polynomially increasing $\set{k_i}$, we show that each iteration requires only $O(\log\log n)$ rounds and that we require at most $\ell\in O(\log\log n)$ iterations. Consequently, we show that $H$ is a $(2/\gamma)^{O(\log\log n)}$-stretch, $O(\log\log n)$-hopbound emulator for the input graph $G$ which can be constructed in $O(\log^2\log n)$ rounds.    

Before stating the lemma, there are some minor implementation details to take note of:
\begin{itemize}
    \item We sparsify the skeleton graph $G_i$ into a spanner $G_i'$ at the start of each iteration. This is to ensure that the number of edges in $G_i'$ is well-bounded. As long as the spanner also maintains a constant approximation of distances, this will not affect the stretch analysis of the algorithm up-to a constant factor (in the exponent of the stretch term).

    \item The sequence $\set{k_i}$ increases polynomially until it reaches a small polynomial in $n$. Namely, once $k\in n^{O(\gamma)}$, we maintain $k$ at this value until the termination condition. This is to ensure we never exceed the local computation capabilities of the machines. 
\end{itemize}

Following from the above discussion, we will now prove the following lemma. 
\begin{lemma}\label{lemma:hetero-good-enough-hopset}
    Let $G=(V,E)$ be a connected, weighted graph with $n$ vertices. In the heterogeneous MPC model with $\tilde O(|E| + n^{1+\gamma})$ total memory, there exists a randomised algorithm that w.h.p.~constructs a $(2/\gamma)^{O(\log\log n)}$-stretch, $O(\log\log n)$-hopbound emulator $H$ in $O(\log^2\log n)$ rounds.     
\end{lemma}

We first state the algorithm description in \Cref{section:good-enough-hopset-algo} before showing the correctness of the algorithm in \Cref{section:good-enough-hopset-correctness}.

\subsubsection{Algorithm Description}\label{section:good-enough-hopset-algo}
As input, we are given the weighted $n$-vertex graph $G=(V,E)$, whose edge weights are initially polynomially-bounded. We begin with $G_0=G$ and an empty edge set $H=\emptyset$. We first define the sequence $\set{k_i}$ for all integers $i\in \set{0,1,2,\ldots}$ recursively as follows: 
\begin{align*}
    k_0 &= \log n, \\
    k_i &= \min\left[\frac{k_{i-1}^{1.99}}{C_S},n^{\gamma/4}\right],
\end{align*}
where $C_S$ is the constant defined in \Cref{definition:skeleton-graph}.

During iteration $i\in\set{0,1,2\ldots}$, we start with the graph $G_i=(V_i,E_i)$ and parameter $k_i$. Assume that $k_i\cdot |V_i|\in O(n\log n)$; we justify this assumption in \Cref{lemma:hetero-kV-bound}. We also assume that $G_i$ is connected; this is justified in \Cref{lemma:hetero-connectivity}. Lastly, we may assume that $|V_i|\leq n$; this is clear in the following algorithm description. The edges of $G_i$ are stored arbitrarily across all the machines. 

\begin{description}
    \item[Step 1: Sparsify $G_i$.] We first sparsify $G_i$ into a spanner graph $G_i'=(V_i,E_i')$. Namely, we apply \Cref{lemma:near-linear-spanner} on $G_i$ with $t=2/\gamma$ to obtain a $C(2/\gamma)$-spanner of size $|E_i'|\in \tilde O(n^{1+\gamma/2})$ in $O(1)$ rounds, where $C$ is a suitable constant. In all subsequent steps, we will work directly with the graph $G_i'$ rather than $G_i$. 
    
    \item [Step 2: Get Approximate $k_i$-Nearest.] Our goal for this step is to obtain $O(\log n)$-approximate $k_i$-nearest neighbourhoods for all vertices in $V_i.$ We accomplish this in the following manner:
    \begin{enumerate}
        \item We first apply \Cref{lemma:near-linear-spanner} on $G_i'$ with $t=\log |V_i|$ to obtain an $O(\log |V_i|)\subseteq O(\log n)$-spanner $S$ of size $\tilde O(|V_i|)\subseteq \tilde O(n)$ in $O(1)$ rounds. 
        \item Since the size of $S$ is $\tilde O(n)$, we can collect the spanner $S$ onto the single near-linear memory machine $M^*$.
        \item Locally, $M^*$ constructs the $k_i$-nearest neighbourhoods $N_{k_i}(u)$ for every vertex $u\in V_i$ \emph{in the spanner $S$} -- these suffice as $O(\log n)$-approximate $k_i$-nearest neighbourhoods \emph{in the graph $G_i'$}. By our previous assumption, we have $k_i\cdot |V_i|\in O(n\log n)$, and this therefore fits within the local memory constraints of $M^*$. 
    \end{enumerate}
    At the end of this step, we now have $O(\log n)$-approximate $k_i$-nearest neighbourhoods $\tilde N_{k_i}(u)$ for all vertices $u\in V_i$. This step runs in $O(1)$ rounds.
    
    \item [Step 3: Get Exact $k_i$-Nearest.] Now that we have constructed $a=O(\log n)$-approximate $k_i$-nearest neighbourhoods, we can find the exact $k_i$-nearest neighbourhoods for all vertices in $V_i$. Namely, since $k_i^2\leq n^{\gamma/2}\leq n^\gamma$, we can apply \Cref{lemma:E_k} on $G_i'$ with parameter $k=k_i$. This constructs the $k_i$-nearest edge set $E_{k_i}$ for $G_i'$. 

    Denoting $d_i$ as the weighted diameter of $G_i'$, this step runs in $O(\log a + \log\log d_i)\subseteq O(\log\log n+ \log\log d_i)$ rounds. This step requires $\tilde O(|E_i'|+nk_i^2)=\tilde O(n^{1+\gamma/2})$ total memory.

    \item [Step 4: Add Edges to $H$.] We add the edge set $E_{k_i}$ to $H$. That is, $H:=H\cup E_{k_i}$.

    \item [Step 5: Construct Skeleton Graph.] We apply \Cref{lemma:skeleton-implement} on $G_i'$ (with $k = k_i$) to obtain the skeleton graph $G_S = (V_S, E_S)$. We set $G_{i+1} := G_S$ and proceed to the next iteration. This step runs in $O(1)$ rounds. This step requires $\tilde O(|E_i'|k_i^2)\subseteq\tilde O(n^{1+\gamma/2}\cdot n^{\gamma/2})=\tilde O(n^{1+\gamma})$ total memory. 
\end{description}

\paragraph{Termination.} We terminate at iteration $\ell$ once the graph $G_{\ell}=(V_{\ell}, E_{\ell})$ is sufficiently small such that $|V_{\ell}|^2\in O(n)$. We collect $G_{\ell}$ onto the single near-linear memory machine $M^*$, which determines all-pairs distances for vertices in $G_{\ell}$. An edge for every pair of vertices in $V_\ell$ is then added into the edge set $H$ with the shortest-path distances as the edge weights. (i.e.,~$H:=H\cup \set{\set{u,v}:u,v\in V_\ell}$ with edge weights $d_{G_\ell}(u,v)$). This concludes the construction, and the final edge set $H$ is our emulator.

\subsubsection{Correctness}\label{section:good-enough-hopset-correctness}
We first justify our claims that $k_i|V_i|=O(n\log n)$ and $G_i$ is connected for all iterations $i$. We may assume that $n$ is sufficiently large such that $k_i^{0.01} \geq \log k_i$ always.
\begin{lemma}\label{lemma:hetero-kV-bound}
    For all $i\in\set{0,1,2,\ldots}$, we have $k_i|V_i|\leq n\log n$.
\end{lemma}
\begin{proof}
    We prove by induction on $i$. For the base case $i=0$, we have $k_0=\log n$ and $|V_0|=n$, and therefore $k_0|V_0|=n\log n$ as desired. Proceeding to the inductive case: Suppose the claim holds for iteration $i=j$. Recalling Step 5 of the algorithm and the skeleton graph definition in \Cref{definition:skeleton-graph}, we observe that $|V_{j+1}|\leq(C_S|V_j|\log k_j)/k_j$. Additionally, we have that $k_{j+1}\leq k_j^{1.99}/C_S$ by construction. Therefore, applying our induction hypothesis,
    \[
    k_{j+1}|V_{j+1}| \leq \frac{k_j^{1.99}}{C_S}\left(\frac{C_S|V_j|\log k_j}{k_j}\right)\leq k_j^{1.99}\left(\frac{|V_j|k^{0.01}_j}{k_j}\right)=k_j|V_j|\leq n\log n.\qedhere
    \]
\end{proof}

\begin{lemma}\label{lemma:hetero-connectivity}
    For all $i\in \set{0,1,2,\ldots}$, the graph $G_i$ is connected.
\end{lemma}
\begin{proof}
    This is proved straightforwardly using induction. In the base case, we have that $G_0$ is connected by the input assumptions. In the inductive case, supposing $G_j$ is connected for some $j \geq 0$, the properties of spanners and skeleton graphs (\Cref{definition:skeleton-graph}) ensure that $G_j'$ is connected, and therefore $G_{j+1}$ is also connected.
\end{proof}

\paragraph{Organisation.} For the remainder of the section on correctness, it is convenient to describe the algorithm in terms of phases. In particular, we define the \emph{first phase} and \emph{second phase} as follows:
\begin{description}
    \item[First Phase:] This phase consists of the iterations where $k_i < n^{\gamma/4}$.
    \item [Second Phase:] This phase consists of the remaining iterations where $k_i = n^{\gamma/4}$, up until iteration $i=\ell$ when the termination condition is satisfied. In particular, denote $i=\ell^*$ as the first iteration for which $k_{\ell^*}=n^{\gamma/4}$.
\end{description}
We present the proof of correctness in three parts:
\begin{itemize}
    \item We first show that the algorithm terminates in $O(\log\log n)$ iterations. In particular, we will show that the first phase requires $O(\log\log n)$ iterations while the second phase only requires $O(1/\gamma)$ iterations.
    \item We next show that each iteration runs in $O(\log\log n)$ rounds. The crux of the argument is to show that for both the first and second phases, the edge weights of all graphs $G_i$ are bounded by $\poly(n)$. 
    \item Lastly, we show that $H$ is indeed a $(2/\gamma)^{O(\log\log n)}$-stretch and $O(\log\log n)$-hopbound emulator for the input graph $G$.
\end{itemize}
For convenience and ease of presentation, we assume that the constant $C_S$ defined in \Cref{definition:skeleton-graph} is at least unity, i.e.,~$C_S\geq 1$; the case in which $C_S < 1$ is argued similarly. 

\paragraph{Number of Iterations.} We first argue that the first phase requires only $O(\log\log n)$ iterations.
\begin{lemma}\label{lemma:hetero-phase1-iteration}
    $\ell^*\in O(\log\log n)$. 
\end{lemma}
\begin{proof}
    Let $p=1.99$. By the definition of $k_i$, observe that for all iterations $i < \ell^*$, we have 
    $$
        k_i =\left(\frac 1{C_S}\right)^{\frac{p^i-1}{p-1}}\cdot (\log n)^{p^i} \geq \left[\left(\frac 1{C_S}\right)^{\frac1{p-1}}\cdot(\log n)\right]^{p^i},
    $$
    using the fact that $C_S \geq 1$. Therefore, after
    $$
    i_{min}= \frac{1}{\log p}\cdot \log\left[\frac{\gamma\log n}{4\log(C_S^{1/(1-p)}\log n)}\right]\in O(\log\log n)
    $$
    iterations, we necessarily have that $k_{i_{min}} \geq n^{\gamma/4}$. It follows that $\ell^* \in O(\log\log n)$ as claimed.
\end{proof}

Proceeding to the second phase, observe that the number of vertices reduces roughly by a factor of $n^{\gamma/4}$ between iterations. Accounting for the additional log factor, we first provide an upper bound on the number of vertices $V_i$. Subsequently, we show that this phase requires only $O(1/\gamma)$ iterations. 

\begin{lemma}\label{lemma:hetero-phase2-size}
    For all $i\in \set{0,1,2\ldots,\ell-\ell^*}$, we have that
    $$
        |V_{i+\ell^*}| \leq C_S^i\left(\frac \gamma4\right)^i\cdot n^{1-(i\gamma/4)}\log ^in,
    $$
    where $C_S$ is the constant as defined in \Cref{definition:skeleton-graph}.
\end{lemma}
\begin{proof}
    We prove by induction on $i$. In the base case of $i=0$, we have $|V_{\ell^*}| \leq n$ as expected. Consider now the inductive case: Suppose the claim holds for $i=j$. Recalling the skeleton graph definition in \Cref{definition:skeleton-graph}, we observe that $|V_{j+1+\ell^*}|\leq(C_S\cdot |V_{j+\ell^*}|/k_{j+\ell^*})\log k_{j+\ell^*}$. Since we are in the second phase, we also have $k_{j+\ell^*}=n^{\gamma/4}$. This allows us to invoke our induction hypothesis to obtain 
    \begin{align*}
        |V_{j+1+\ell^*}| &\leq C_S\left[\frac{|V_{j+\ell^*}|\log (n^{\gamma/4})}{n^{\gamma/4}}\right]\\
        &\leq C_S \left[\frac{\log (n^{\gamma/4})}{n^{\gamma/4}}\right]\cdot \left[C_S^j\left(\frac \gamma4\right)^j\cdot n^{1-(j\gamma/4)}\log ^jn\right]\\
        &= C_S^{j+1}\left(\frac \gamma4\right)^{j+1}\cdot n^{1-(j+1)\gamma/4}\log ^{j+1}n 
    \end{align*}
    as claimed.
\end{proof}

We are now ready to bound the number of iterations in the second phase. 
\begin{lemma}\label{lemma:hetero-phase2-iteration}
    After $i\in O(1/\gamma)$ iterations in the second phase, we have $|V_{i+\ell^*}|\in O(n^{1/2})$. 
\end{lemma}
\begin{proof}
    Observe that the terms $(\gamma/4)^i$ and $\log^in$ (in \Cref{lemma:hetero-phase2-size}) are both asymptotically smaller than $n^{\eps \cdot i}$, for \emph{any} constant $\eps > 0$. More concretely, we have 
    $$
        O(C_S^i\cdot (\gamma/4)^i\cdot \log^in)\subseteq O(n^{\eps \cdot i})
    $$
    for any $\epsilon > 0$. Therefore, we may rewrite the bound in \Cref{lemma:hetero-phase2-size} as 
    $$
        |V_{i+\ell^*}| \in O(n^{1-(\gamma/4-\eps)\cdot i}).
    $$
    Fixing say, $\eps=\gamma/8$, we get that after $i\in O(1/[2(\gamma/4-\eps)])=O(1/\gamma)$ iterations in the second phase, we necessarily have $|V_{i+\ell^*}|\in O(n^{1/2})$. 
\end{proof}

We are now ready to show that the entire algorithm terminates in $O(\log\log n + 1/\gamma)$ iterations. Namely, since the first phase requires $O(\log\log n)$ iterations (\Cref{lemma:hetero-phase1-iteration}) and the second phase requires $O(1/\gamma)$ iterations (\Cref{lemma:hetero-phase2-iteration}), we get the desired result.
\begin{lemma}\label{lemma:hetero-iterations}
    The entire algorithm runs in $O(\log\log n)$ iterations.
\end{lemma}
\begin{proof}
    By \Cref{lemma:hetero-phase1-iteration}, we observe that the first phase runs in $O(\log\log n)$ iterations. Additionally, by \Cref{lemma:hetero-phase2-iteration}, we observe that after $O(1/\gamma)$ iterations in the second phase, the termination condition is satisfied. Therefore in total, we require $O(\log\log n + 1/\gamma)\subseteq O(\log\log n)$ iterations across the entire algorithm.
\end{proof}

\paragraph{Round Complexity.} We first bound the edge weights of the graphs $\set{G_i}$ in the first phase. We do this in the following manner: We first derive a crude bound for the edge weights (\Cref{lemma:hetero-bad-weight-bound}), and thereafter a bound for the values $\set{k_i}$ (\Cref{lemma:hetero-k-ub}). Finally, this enables us to obtain a good bound for the edge weights in the first phase, where we show that $d_i\in \poly(n)$ (\Cref{lemma:hetero-phase1-weight}). 

Suppose the initial weights of the graph $G=G_0$ are bounded above by the function $g_0(n)=\poly(n)$. We begin with a crude, almost trivial bound for the edge weights. Note that in the following lemma, our induction begins with $i=1$ instead of $i=0$ for convenience. 
\begin{lemma}\label{lemma:hetero-bad-weight-bound}
    For all $i\in \set{1,2,\ldots,\ell^*}$, the weights of the graph $G_i$ are bounded above by the function 
    $$
        g_i(n) = \left(\prod_{j=0}^{{i-1}}3k_j\right)\cdot g_0(n).
    $$
\end{lemma}
\begin{proof}
    We prove by induction on $i$. In the base case $i=1$, recall that $G_1$ is a skeleton graph of $G_0'$. Since the edge weights of $G_0'$ are bounded above by $g_0(n)$, recalling \Cref{definition:skeleton-graph}, we find that the function $g_1(n)=3k_0\cdot g_0(n)$ is a suitable upper bound for the weights of $G_1$. The inductive case works similarly: Suppose the claim holds for $i=j$. Since $G_{j+1}$ is a skeleton graph of $G_j'$, we have that the weights of $G_{j+1}$ are bounded above by the function $g_{j+1}(n)=3k_j\cdot g_j(n)$. Invoking our induction hypothesis, we therefore get the desired claim. 
\end{proof}

We now provide a bound for the values $\set{k_i}$ in the first phase. Roughly speaking, we work backwards: Rather than starting our induction with $k_0=\log n$, we start our induction instead with $k_{\ell^*-1}\leq n^{\gamma/4}$. This allows us to express our bound for $k_i$ in terms of a small polynomial in $n$, which will be useful later on. 
\begin{lemma}\label{lemma:hetero-k-ub}
Let $p=1.99$, and denote $S_i$ as the geometric sum of $i$ terms with common ratio and initial term $1/p$. For $i\in \set{0,1,2\ldots,\ell^*-1}$, we have
    $$
        k_i \leq C_S^{S_{\ell^*-1-i}}\cdot (n^{\gamma/4})^{(1/p)^{\ell^*-1-i}}.
    $$
    As a direct consequence, this implies that $k_i \leq C_S^{1/(p-1)}\cdot (n^{\gamma/4})^{(1/p)^{\ell^*-1-i}}$.
\end{lemma}
\begin{proof}
    We prove by induction on $i$, starting with $i=\ell^*-1$. For the base case $i=\ell^*-1$, we get that $k_{\ell^*-1}\leq n^{\gamma/4}$ as $S_0=0$, which agrees with the definition of $\set{k_i}$. Proceeding to the inductive case: Suppose the claim holds for $i=j$. We recall from our definition of $k_{j}$ that $k_{j-1}=(C_Sk_j)^{1/p}$. Therefore invoking our induction hypothesis, we have that
    $$
        k_{j-1}=(C_Sk_j)^{1/p}\leq C_S^{1/p}\cdot \left[C_S^{S_{\ell^*-1-j}} (n^{\gamma/4})^{(1/p)^{\ell^*-1-j}}\right]^{1/p}=C_S^{S_{\ell^*-1-(j-1)}} (n^{\gamma/4})^{(1/p)^{\ell^*-1-(j-1)}},
    $$
    as $C_S \geq 1$ by assumption. Since $1/p < 1$, we have that as $i\rightarrow \infty$, $S_i$ converges to $1/(p-1)$ by property of the geometric sum. In other words, we have that 
    $$
                k_i \leq C_S^{S_{\ell^*-i-1}}\cdot (n^{\gamma/4})^{(1/p)^{\ell^*-1-i}}\leq C_S^{1/(p-1)}\cdot (n^{\gamma/4})^{(1/p)^{\ell^*-1-i}}
    $$
    for all $i\in \set{0,1,2,\ldots,\ell^*-1}$.
\end{proof}

In the following lemma, we show that the edge weights in the first phase are polynomial in $n$. This makes use of our bound for $k_i$ in \Cref{lemma:hetero-k-ub}, which we expressed in terms of $n^{\gamma/4}$. The key observation is that multiplying each $k_i$ successively results in an \emph{additive} effect in the exponent of $n$. In particular, the exponent of $n$ follows a geometric series that converges to a constant value, which shows the desired claim. With this lemma, we conclude our analysis of the first phase.
\begin{lemma}\label{lemma:hetero-phase1-weight}
    Let $p=1.99$. For all $i\in \set{0,1,2\ldots,\ell^*}$, the weights of the graph $G_i$ are bounded above by the function 
    $$
        g^*(n) = \left(3C_S^{1/(p-1)}\right)^{\ell^*}n^{p\gamma/[4(p-1)]}g_0(n)\in \poly(n).
    $$
\end{lemma}
\begin{proof}
    In order to bound the weights for \emph{all} the graphs in the first phase, it suffices to bound the function $g_{\ell^*}(n)$, as defined in \Cref{lemma:hetero-bad-weight-bound}. Applying our upper bounds on $k_i$ in \Cref{lemma:hetero-k-ub}, we have that 
    $$
    g_{\ell^*}(n)=\left(\prod_{j=0}^{{\ell^*-1}}3k_j\right)\cdot g_0(n) \leq \left(3C_S^{1/(p-1)}\right)^{\ell^*} \left(\prod_{i=0}^{{\ell^*-1}}(n^{\gamma/4})^{(1/p)^{\ell^*-1-i}}\right)\cdot g_0(n).
    $$
    Notice that the exponent for the term $(n^{\gamma/4})$ is a geometric sum with the common ratio $1/p$ and initial term of 1. Since $1/p < 1$, this sum converges to $p/(p-1)$ at infinity. Therefore, we may further bound $g_{\ell^*}(n)$ with
    $$
        g_{\ell^*}(n)\leq g^*(n)=\left(3C_S^{1/(p-1)}\right)^{\ell^*}n^{p\gamma/[4(p-1)]}g_0(n).
    $$
    Since $\ell^*\in O(\log\log n)$ by \Cref{lemma:hetero-phase1-iteration}, $\gamma \leq 1$ and the remaining variables (except for $n$) are treated as constants, we find that $g^*(n)\in \poly(n)$ as desired.
\end{proof}

We now proceed to the analysis of the edge weights in the second phase. The analysis of this phase is relatively simpler compared to the first: There are only a constant number of iterations in the second phase (\Cref{lemma:hetero-phase2-iteration}) and the weights increase by at most a factor of $O(n)$ between iterations. Since we have shown that the first graph in the second phase $G_{\ell^*}$ has polynomially-bounded weights (\Cref{lemma:hetero-phase1-weight}), we therefore have that the edge weights throughout the second phase must also be polynomially-bounded. 

\begin{lemma}\label{lemma:hetero-phase2-weight}
    Suppose the weights of the graph $G_{\ell^*}$ are bounded above by the function $g^*(n)$. Then, for all $i\in \set{0,1,2,\ldots,\ell-\ell^*}$, the weights of the graph $G_{i+\ell^*}$ are bounded above by the function
    $$
        g_{i+\ell^*}(n) = \left(3n^{\gamma/4}\right)^i\cdot g^*(n).
    $$
\end{lemma}
\begin{proof}
    We prove by induction on $i$. In the base case of $i=0$, we have $g_{\ell^*}(n)=g^*(n)$ as expected. Consider now the inductive case: Suppose the claim holds for $i=j$. Recalling \Cref{definition:skeleton-graph}, we observe that the function $g_{j+1+\ell^*}(n):=3k_{j+\ell^*}\cdot g_{j+\ell^*}(n)=3n^{\gamma/4}\cdot g_{j+\ell^*}(n)$ is a suitable upper bound for the weights in $G_{j+1+\ell^*}$. Applying our induction hypothesis, we then have $g_{j+1+\ell^*}(n)=(3n^{\gamma/4})^{j+1}\cdot g^*(n)$ as claimed. 
\end{proof}

With this next lemma, we claim the edge weights of \emph{all} the graphs $\set{G_i}$ are polynomially-bounded.
\begin{lemma}\label{lemma:hetero-poly-weights}
    For all $i\in \set{0,1,2,\ldots,\ell}$, the edge weights of the graph $G_i$ are bounded above by $\poly(n)$.
\end{lemma}
\begin{proof}
    By \Cref{lemma:hetero-phase1-weight}, we know that the weights of the graphs $\set{G_0,G_1,G_2,\ldots,G_{\ell^*}}$ are polynomially-bounded. To bound the edge weights of the remaining graphs $\set{G_{\ell^*+1},G_{\ell^*+2},\ldots,G_\ell}$, it suffices to show that the function $g_\ell(n)$ as defined in \Cref{lemma:hetero-phase2-weight} is bounded above by $\poly(n)$, i.e., that the last graph of the second phase has polynomially-bounded weights. 
    
    By \Cref{lemma:hetero-phase1-weight}, we know that the edge weights of $G_{\ell^*}$ are bounded above by the function $g^*(n)\in \poly(n)$. Additionally, by \Cref{lemma:hetero-phase2-iteration}, we know that the number of iterations in the second phase is at most $O(1/\gamma)$. Therefore, we conclude that the function $g_\ell(n)$ is bounded by
    $$
        g_\ell(n) \in \left(3n^{\gamma/4}\right)^{O(1/\gamma)}\cdot g^*(n) \in 3^{O(1/\gamma)} n^{O(1)}\cdot g^*(n)\subseteq \poly(n),
    $$
    as desired.
\end{proof}

We are now ready to show that the weighted diameters of the graphs $G_i'$ are always polynomial in $n$ for both the first and second phases. As a result, Step 3 of each iteration, and by extension the entirety of each iteration, requires only $O(\log\log n)$ rounds. Since we only have $O(\log\log n)$ iterations, this gives an overall round complexity of $O(\log^2\log n)$ rounds. 
\begin{lemma}\label{lemma:hetero-rounds}
    Each iteration runs in $O(\log\log n)$ rounds. As a direct consequence, the entire algorithm runs in $O(\log^2\log n)$ rounds.
\end{lemma}
\begin{proof}
    We first bound the weighted diameters $d_i$ of the graphs $G_i'$. By \Cref{lemma:hetero-poly-weights}, we know that the edge weights of the graphs $\set{G_i}$ for all iterations $i$ are bounded above by some function $g_i(n)\in \poly(n)$. Moreover, recall that $G_i'$ is a subgraph of $G_i$ by the definition of spanners. Therefore, we must have that $d_i\leq (n-1)\cdot g_i(n)\in \poly(n)$ as any shortest path in $G_i'$ has at most $n-1$ edges.  
    
    To assess the round complexity of each iteration, we break it down by steps.
    \begin{itemize}
        \item Steps 1, 2, 4 and 5 require $O(1)$ rounds.
        \item Step 3 requires $O(\log\log n + \log\log d_i)\subseteq O(\log\log n)$ rounds as $d_i\in\poly(n)$ for all iterations.
    \end{itemize}
    Therefore, we observe that each iteration requires $O(\log\log n)$ rounds. Since we have $O(\log\log n)$ iterations by \Cref{lemma:hetero-iterations}, the entire algorithm requires $O(\log\log n\cdot\log \log n) = O(\log^2\log n)$ rounds as claimed. 
\end{proof}

\paragraph{Emulator Properties.} We are now ready to bound the hopbound and stretch of the emulator $H$. To gain an intuitive picture, you can imagine that we are ``climbing a tower'' of skeleton graphs using edges in $H$, starting from $G=G_0$. We start with any pair of vertices $u_0,v_0\in V_0$ and the path $P_0$ from $u_0$ to $v_0$ with properties (P1) and (P2) as defined in \Cref{definition:skeleton-graph}. Denote the path $P_0$ as $u_0\rightarrow u_1\rightsquigarrow v_1\rightarrow v_0$, and denote the first and last edges of $P_0$ as $e_1=\set{u_0,u_1}$ and $e_2=\set{v_1,v_0}$ respectively. We make two observations: 
\begin{itemize}
    \item By property (P2), we know the edges $e_1$ and $e_2$ can be found in $H$.
    \item Also by property (P2), we know that the subpath from $u_1$ to $v_1$ in $P_0$ is contained \emph{entirely} in $G_1$.
\end{itemize}
Therefore, using the edges $e_1$ and $e_2$, we can ``climb up'' from $G_0$ to $G_1$. By property (P1) of path $P_0$, if we find a constant approximation for $d_{G_1}(u_1,v_1)$, this will directly yield a constant approximation for $d_{G_0}(u_0,v_0)$. This motivates us to repeat the same process: We now consider the path $P_1$ from $u_1$ to $v_1$ (as defined in \Cref{definition:skeleton-graph}) and using the same procedure, climb from $G_1$ to $G_2$. We repeat this until we reach the graph $G_\ell$, where we are finally able to ``connect'' the endpoints. Since we only have an $O(\log\log n)$ number of graphs to climb (by \Cref{lemma:hetero-iterations}), we only need $O(\log\log n)$ edges in $H$. Additionally, combining these edges yields a $(2/\gamma)^{O(\log\log n)}$-approximate length path from $u_0$ to $v_0$ in $H$. 

We now formalise the above intuition in the following lemma.
\begin{lemma}\label{lemma:hetero-hopset-correctness-ub}
    For all $i\in \set{0,1,2,\ldots,\ell}$, consider any pair of vertices $u,v\in V_i$. There exists a path $P$ from $u$ to $v$ in $H$ such that all of the following hold:
    \begin{itemize}
        \item $P$ has at most $1+2(\ell-i)$ edges, and
        \item the length of $P$ is at most $[7C(2/\gamma)]^{\ell-i}\cdot d_{G_i}(u,v)$, where $C$ is the same constant defined in Step 1 of the algorithm.
    \end{itemize}
\end{lemma}
\begin{proof}
    We prove by induction on $i$, starting from $i=\ell$. 
    \begin{description}
        \item[Base Case ($i=\ell$).] Recalling the termination stage of the algorithm, we know that the all-pairs distances of $G_\ell$ were added to $H$ as weighted edges. Therefore, the base case holds.

        \item[Inductive Case.] Suppose the claim holds for $i=j$. Consider any two vertices $u,v\in V_{j-1}$.  In the case where $u\in N_{k_{j-1}}(v)$ or $v\in N_{k_{j-1}}(u)$ (with respect to the graph $G'_{j-1}$), the claim trivially holds. So suppose not: Consider the path $P$ from $u$ to $v$ in $E_{k_{j-1}}\cup E_{j}$ with the properties (P1) and (P2) as defined in \Cref{definition:skeleton-graph}. We denote the path $P$ as $u\rightarrow u^*\rightsquigarrow v^*\rightarrow v$. Additionally, let $e_1 = \set{u,u^*}$ and $e_2 = \set{v^*,v}$ be the first and last edges of path $P$ respectively. 
        
        By property (P2), we know that $e_1,e_2\in H$. Moreover, the subpath from $u^*$ to $v^*$ in $P$ is contained \emph{entirely} in the graph $G_j$. We therefore apply our induction hypothesis on $G_j$ to obtain the path $P^*\subseteq H$ from $u^*$ to $v^*$, which has the following properties:
        \begin{itemize}
            \item $P^*$ has at most $1+2(\ell-j)$ edges.
            \item The length of $P^*$ is at most $[7C(2/\gamma)]^{\ell-j}\cdot d_{G_j}(u^*,v^*)$.
        \end{itemize}
        We now consider a new path $Q$ from $u$ to $v$, defined as follows: From $u$ to $u^*$ using edge $e_1$, $u^*$ to $v^*$ using path $P^*$, and lastly $v^*$ to $v$ using edge $e_2$. Observe that $Q$ is entirely contained in the edge set $H$. Moreover, we find that $Q$ has the following properties:
        \begin{itemize}
            \item $Q$ has at most $1 + 2(\ell - j) + 2=1+2(\ell-(j-1))$ edges.
            \item Considering $P$ is a $7$-approximation of $d_{G_{j-1}'}(u,v)$ and that $G_{j-1}'$ is a $C(2/\gamma)$-spanner for $G_{j-1}$, we have
            \begin{align*}
                \omega(Q)&= \omega(e_1) + \omega(P^*) + \omega(e_2)\\
                &\leq \omega(e_1) + [7C(2/\gamma)]^{\ell-j}\cdot d_{G_j}(u^*,v^*) + \omega(e_2)\\
                &= d_{G_{j-1}'}(u,u^*) + [7C(2/\gamma)]^{\ell-j}\cdot d_{G_j}(u^*,v^*) + d_{G_{j-1}'}(v^*,v) \\
                &\leq [7C(2/\gamma)]^{\ell-j}\left[d_{G_{j-1}'}(u,u^*) + d_{G_j}(u^*,v^*)+ d_{G_{j-1}'}(v^*,v)\right]\\
                &\leq [7C(2/\gamma)]^{\ell-j}\cdot \omega(P)\\
                &\leq [7C(2/\gamma)]^{\ell-j}\cdot 7\cdot d_{G'_{j-1}}(u,v) \\
                &\leq [7C(2/\gamma)]^{\ell-j}\cdot 7C(2/\gamma)\cdot d_{G_{j-1}}(u,v) \\
                &= [7C(2/\gamma)]^{\ell-(j-1)}\cdot d_{G_{j-1}}(u,v). \qedhere
            \end{align*}
        \end{itemize}
    \end{description}
\end{proof}

We now summarise the properties of our emulator $H$ in the following lemma.
\begin{lemma}\label{lemma:hetero-hopset-correctness-final}
    $H$ is a $(2/\gamma)^{O(\log\log n)}$-stretch $O(\log\log n)$-hopbound emulator of the graph $G=G_0$.
\end{lemma}
\begin{proof}
   By \Cref{lemma:hetero-iterations}, we have that the final iteration $\ell\leq T\cdot \log\log n$ for some constant $T$. Therefore, applying our stretch upper bound in \Cref{lemma:hetero-hopset-correctness-ub} with $i=0$, we have that for any pair of vertices $u,v\in V$,
    $$
        d_{H}^{(1 + 2T\cdot \log\log n)}(u,v)\leq [7C(2/\gamma)]^{(T\cdot \log\log n)}d_G(u,v)\in (2/\gamma)^{O(\log\log n)}d_G(u,v).
    $$
    To show the lower bound, we recall our definition of the skeleton graph in \Cref{definition:skeleton-graph}. Specifically, we observe that for all iterations $i$, the distances in $G$ are never contracted in the subsequent graphs $G_i$, i.e.,~$d_G(u,v)\leq d_{{G_i}}(u,v)$. Consequently, we have that $d_G(u,v)\leq d_{H}(u,v)$ for all pairs of vertices $u,v\in V$. Therefore, we conclude that $H$ is indeed a $(2/\gamma)^{O(\log\log n)}$-stretch $O(\log\log n)$-hopbound emulator for $G$.
\end{proof}

We now conclude with the proof of \Cref{lemma:hetero-good-enough-hopset}.
\begin{proof}[Proof of \Cref{lemma:hetero-good-enough-hopset}]
    We run the algorithm described in \Cref{section:good-enough-hopset-algo}. This constructs the edge set $H$, which by \Cref{lemma:hetero-hopset-correctness-final} is a $(2/\gamma)^{O(\log\log n)}$-stretch $O(\log\log n)$-hopbound emulator for $G$. By \Cref{lemma:hetero-rounds}, we only require $O(\log^2\log n)$ rounds.
    
    By the algorithm description, we observe that each iteration requires at most $\tilde O(|E| + n^{1+\gamma})$ total memory, and since there are $O(\log\log n)$ iterations by \Cref{lemma:hetero-iterations}, we therefore require at most $\tilde O(\log\log n \cdot (|E| + n^{1+\gamma}))\subseteq \tilde O(|E| + n^{1+\gamma})$ total memory. In particular, we have sufficient total memory to store the emulator $H$ across the machines.
\end{proof}

\subsubsection{Wrapping it Up}
Now having access to a $(2/\gamma)^{O(\log\log n)}$-stretch $O(\log\log n)$-hopbound emulator, we can apply \Cref{lemma:exact-k-nearest} on the emulator to construct $(2/\gamma)^{O(\log\log n)}$-approximate $k$-nearest neighbourhoods for all vertices. We state this in the following lemma. 
\begin{lemma}\label{lemma:hetero-oracle}
    Let $G=(V,E)$ be a connected, weighted graph with $n$ vertices and let integer $k > 0$ be a parameter. Given that $k^2\in  O(n^\gamma)$, in the heterogeneous MPC model with $\tilde O(|E| + n^{1+\gamma})$ total memory, there exists a randomised algorithm that w.h.p.~constructs the set $\set{(u,t):u\in V\wedge t\in \tilde N_k(u)}$ where $\tilde N_k(u)$ is a $(2/\gamma)^{O(\log\log n)}$-approximate $k$-nearest neighbourhood for $u$. This algorithm runs in $O(\log^2\log n)$ rounds.     
\end{lemma}
\begin{proof}
    We first invoke \Cref{lemma:hetero-good-enough-hopset} to construct a $(2/\gamma)^{O(\log\log n)}$-stretch $O(\log\log n)$-hopbound emulator $H$ for the graph $G$. This requires $\tilde O(|E|+n^{1+\gamma})$ total memory and runs in $O(\log^2\log n)$ rounds. Denote the hopbound of the emulator as $C_H \log\log n$, where $C_H > 0$ is some constant.
    
    We now apply \Cref{lemma:exact-k-nearest} on $H$, setting the parameter $r$ such that $2^r=h\geq C_H \log\log n$. This gives us  the $k$-nearest $h$-hop neighbourhoods $N_k^{(h)}(u)$ for all vertices $u\in V$ \emph{in the emulator} $H$. By the properties of the emulator $H$, these suffice as $(2/\gamma)^{O(\log\log n)}$-approximate $k$-nearest neighbourhoods \emph{in the graph} $G$, and we are done. This requires $\tilde O(|H| + nk^2)\subseteq \tilde O(|E| + n^{1+\gamma})$ total memory and $O(r)=O(\log\log\log n)$ rounds. 
    
    Overall, we require $\tilde O(|E| + n^{1+\gamma})$ total memory and $O(\log^2\log n + \log\log\log n)=O(\log^2\log n)$ rounds, as claimed.
\end{proof}

 It now remains to complete the proof for \Cref{theorem:hetero}, which follows similarly to the proof for \Cref{theorem:near-linear}. However, we now use \Cref{lemma:hetero-oracle} for constructing the approximate $k$-nearest neighbourhoods instead. 

\begin{proof}[Proof of \Cref{theorem:hetero}]
   By \Cref{lemma:connectivity-reduction}, we may assume that $G$ is a connected graph. We first apply the spanner construction in \Cref{lemma:near-linear-spanner} to construct an $O(2/\gamma)$-spanner $G'$ for $G$ of size $\tilde O(n^{1+\gamma/2})$ in $O(1)$ rounds. We then apply \Cref{lemma:hetero-oracle} to construct $a=(2/\gamma)^{O(\log\log n)}$-approximate $k=n^{\gamma/4}$-nearest neighbourhoods $\tilde N_k(u)$ for all vertices $u\in V$ in $G'$, which requires $O(\log^2\log n)$ rounds. 
   
    We now apply \Cref{lemma-central-construct} with the parameter $c=\gamma$, which constructs an $O(1/\gamma)$-approximate SSSP oracle for $G'$ in $O(\log a + \log\log n)\subseteq O(\log\log n)$ rounds -- the inclusion follows as we treat $\gamma$ as a constant. In total, we require $\tilde O(|E|+n^{1+\gamma})$ total memory across all steps. 
   
   Since $G'$ is an $O(2/\gamma)$-spanner for $G$, an $O(1/\gamma)$-approximate SSSP oracle for $G'$ directly yields an $O(1/\gamma^2)$-approximate SSSP oracle for $G$, and we are done. In total, we require $O(\log^2\log n + \log\log n)=O(\log^2\log n)$ rounds and $\tilde O(|E|+n^{1+\gamma})$ total memory, as desired.
\end{proof}

For the remainder of this work, we further discuss the building blocks as described in \Cref{section:building-blocks}. 

\section{Computing $k$-Nearest Hopsets}\label{section:k-nearest-hopset}
In this section, we aim to prove the following lemma, which constructs a $k$-nearest hopset for a given weighted graph $G$ in constant rounds. 
\kNearestHopset*
\paragraph{Construction.} We closely follow the construction of \cite[Lemma 2.2]{bui2024improved}. The central idea behind the construction is straightforward: Each vertex $u$ requests the $k$ shortest incident edges from every vertex $t\in \tilde N_k(u)$. This allows $u$ to learn the exact distances to a significant fraction of the vertices in $\tilde N_k(u)$; these exact distances are added as weighted edges to the hopset. Subsequently, one can show that for any vertex $v\in N_k(u)$, an appropriate choice of a hopset edge together with one edge from the original graph covers at least $1/a$-fraction of the distance from $u$ to $v$.  

More concretely, for every vertex $t\in V$, first denote $I_k(t)$ as the $k$ shortest edges incident to $t$ in $G$, breaking ties, say, by endpoint vertex ID. We define the edge set 
$$
T_k(u) = \bigcup_{t\in \tilde N_k(u)}I_k(t)
$$ 
and let $B_k(u)$ be the set of all vertices covered by at least one edge in $T_k(u)$. Lastly, for all $v\in B_k(u)$,  define $\delta_u(v)$ as the distance from $u$ to $v$ in the subgraph of $G$ induced by the edge set $T_k(u)$. We now define the weighted edge set $H$ as follows: 
$$
H=\set{\set{u,v}:u\in V\wedge v\in B_k(u)\wedge \delta_u(v)\neq\infty}, \ \ \ \text{where} \ \ \ 
\omega_H(u,v)=\delta_u(v).$$ If multiple candidates exist, choose the edge with the lowest weight. 
%We now claim that this edge set $H$ is indeed a $k$-nearest hopset for $G$. 
\paragraph{Correctness.} The hopset construction is essentially identical to \cite[Lemma 2.2]{bui2024improved}, except for a minor difference. In \cite{bui2024improved}, they defined $\delta_u(v)$ as the shortest distance from $u$ to $v$ in the induced subgraph of $T_k(u) \cup \set{\set{u,w}: \set{u,w}\in E}$, i.e., the edges incident to $u$ are also considered in the local shortest path computation. However, it is evident from their proof technique that this difference is not significant. For the sake of completeness, we make this explicit. 

We define the $d$-distance ball of a vertex $u$ as $\mathsf{Ball}_d(u):= \set{v: d_G(u,v)\leq  d}$. It is also convenient to define $\ell_k(u) := \max _{v\in N_k(u)}d_G(u,v)$, i.e., the minimum $d$ such that $N_k(u)\subseteq \mathsf{Ball}_d(u)$. Notice that for any $d < \ell_k(u)$, we have that $|\mathsf{Ball}_d(u)| < k$. 

We first invoke two claims from \cite{bui2024improved}. The first claim suggests that any approximate $k$-nearest neighbourhood $\tilde N_k(u)$ contains a ball of a certain distance radius around $u$. The second claim shows that this distance radius is sufficiently large such that short-cutting exact distances to vertices in this ball gives us the desired hopset. We note that the original statements in \cite{bui2024improved} only had $k=\sqrt{n}$, but it is clear that their proofs do not require any special value of $k$. 
\begin{claim}[\cite{bui2024improved}]\label{claim:hopset-ball}
    For any $u\in V$ and any $a$-approximate $k$-nearest neighbourhood $\tilde N_k(u)$ for $u$, we have $\mathsf{Ball}_{(\ell_k(u)-1)/a}(u)\subseteq \tilde N_k(u)$.
\end{claim}
\begin{claim}[\cite{bui2024improved}]\label{claim:hopset-shortcut-ball-distances}
    The edge set $H':= \set{\set{u,v}: u\in V\wedge v\in \mathsf{Ball}_{(\ell_k(u)-1)/a}(u)}$ with edge weights $\omega_{H'}(u,v) := d_G(u,v)$ is a $k$-nearest $O(a\log d)$-hopbound hopset for $G$.
\end{claim}
To connect the preceding two claims, we show in the following claim that the edges $T_k(u)$ are sufficient to determine the exact distances to vertices in  $\mathsf{Ball}_{(\ell_k(u)-1)/a}(u)$.
\begin{claim}\label{claim:hopset-ball-edges}
    Consider any vertices $u,v\in V$ such that $v\in \mathsf{Ball}_{(\ell_k(u)-1)/a}(u)$, and let $u = w_1 \rightarrow w_2 \rightarrow \ldots \rightarrow w_m = v$ be any shortest path from $u$ to $v$ in $G$. For each integer $1\leq i \leq  m - 1$, we have that  $\set{w_i,w_{i+1}}\in I_k(w_i)$.
\end{claim}
\begin{proof}
    Suppose not, i.e., there exists some vertex $w_i$ such that $\set{w_i,w_{i+1}}\notin I_k(w_i)$. Observe that this implies $|I_k(w_i)|=k$. Moreover, for every $z\in V$ such that $\set{w_i,z}\in I_k(w_i)$, we have $d_G(u,z)\leq d_G(u,v)$ -- to observe this, notice that $\omega_G(w_i,z)\leq \omega_G(w_i,w_{i+1}) \leq d_G(w_i,v)$ and therefore
    \begin{equation*}
        d_G(u,z) \leq d_G(u,w_i) + \omega_G(w_i,z) \leq d_G(u,w_i) + d_G(w_i,v) = d_G(u,v).
    \end{equation*}
    In particular, this implies that $d_G(u,z)\leq (\ell_k(u)-1)/a$ and $z\in \mathsf{Ball}_{(\ell_k(u)-1)/a}(u)$. Altogether, we find that the $k$ vertices $\set{z:\set{w_i,z}\in I_k(w_i)}$ are in $\mathsf{Ball}_{(\ell_k(u)-1)/a}(u)$ and hence $|\mathsf{Ball}_{(\ell_k(u)-1)/a}(u)| \geq k$. Contradiction. 
\end{proof}

\begin{claim}[\cite{bui2024improved}]\label{claim:k-nearest-hopset-correctness}
    The edge set $H$ is a $k$-nearest $O(a\log d)$-hopbound hopset for $G$.
\end{claim}
\begin{proof}
    By \Cref{claim:hopset-ball,claim:hopset-ball-edges}, we find that the edge set $T_k(u)$ is sufficient to learn the exact distances to each vertex $v\in \mathsf{Ball}_{(\ell_k(u) - 1)/a}(u)$; in particular, the edges of every shortest path from $u$ to $v\in \mathsf{Ball}_{(\ell_k(u) - 1)/a}(u)$ are contained in $T_k(u)$. \Cref{claim:hopset-shortcut-ball-distances} implies that $H$ is indeed a $k$-nearest $O(a\log d)$-hopbound hopset for $G$.
\end{proof}

For the rest of the section, to finish the proof of \Cref{lemma:k-nearest-hopset}, we show that $H$ can be constructed efficiently in the sublinear MPC model.
In the following implementation details, we use two MPC procedures: the duplication lemma and labelling lemma. For the sake of presentation, we defer the description of these lemmas and their proofs to \Cref{appendix:mpc-tools}.

%\subsection{Algorithm Description}
The high-level idea for the construction is as follows: Every vertex $u$ first requests the set $I_k(t)$ from every vertex $t\in \tilde N_k(u)$. After receiving the edge set $T_k(u)$, vertex $u$ locally computes the distances $\delta_u(v)$ and adds the edge $\set{u, v}$ to $H$ for every vertex $v\in B_k(u)$. Implementing this idea in the MPC model does not pose any significant challenges, although care has to be taken in the communication of edges from one vertex to another.

We remark that the size of the edge set $|T_k(u)|\leq k^2$ by definition, and therefore we require $k^2\in O(n^\gamma)$ for the distances $\set{\delta_u(v):v\in B_k(u)}$ to be computed entirely on a single sublinear memory machine. Moreover, we also require the additional total memory of $O(nk^2)$ to store the edge sets $T_k(u)$ for all vertices $u\in V$. 

\paragraph{Representing Information.} Before discussing implementation details, we first need to describe how we represent information within the MPC model. 
We categorize information into two types of information \emph{packets}:

\begin{description}
    \item[Edges:] For each edge $\set{a,b}\in E$ with weight $\omega_G(a,b)$ \emph{belonging} to vertex $c$, we represent it as $(\mathcal E, c, \set{a,b}, \omega_G(a,b))$.
    \item[$k$-Nearest:] For each vertex $b\in \tilde N_k(a)$, we represent it as $(\mathcal K, b, a)$.
\end{description}
The symbols $\mathcal E$ and $\mathcal K$ are used only to differentiate between the two types of information packets.

\begin{description}
    \item[Step 0: Initialisation.] At the start of the algorithm, we may assume for every vertex $u\in V$, we only have information of its incident edges and its approximate $k$-nearest neighbourhood $\tilde N_k(u)$. 
    \begin{itemize}
        \item For each edge $\set{u,v}\in E$, there exists a machine that contains $(\mathcal E, u, \set{u,v},\omega_G(u,v))$.
        \item For each vertex $t\in \tilde N_k(u)$, there exists a machine that contains $(\mathcal K,t,u)$.
    \end{itemize}
    Initially, this information is arbitrarily distributed across all the machines. For every edge $\set{u,v}\in E$, we may assume there exist two packets representing this edge: \[(\mathcal E, u, \set{u,v},\omega_G(u,v)) \ \ \  \text{and} \ \ \ (\mathcal E, v, \set{u,v},\omega_G(u,v)).\]
    
    \item[Step 1: Send Requests.] To simulate the sending of a request from one vertex to another, it suffices to sort the information packets. Namely, by applying the sorting procedure in \Cref{lemma:mpc-toolbox} using the second item in the tuple as the key, each vertex $t\in V$ now has the following information stored on a \emph{contiguous} set of machines $M(t)$. 
    \begin{itemize}
        \item $(\mathcal E, t, \set{t,v},\omega_G(t,v))$, for every edge $\set{t,v}\in E$ incident to $t$.
        \item $(\mathcal K, t, u)$, for every vertex $u$ such that $t\in \tilde N_k(u)$.
    \end{itemize}
    Denote the machine in $M(t)$ with the smallest ID as $\widehat M(t)$. This step runs in $O(1)$ rounds. Intuitively,  $(\mathcal K, t, u)$ corresponds to a request from $u$ asking $t$ for the edge set $I_k(t)$.
    
    \item[Step 2: Create Duplicates of Incident Edges.] For each vertex $t\in V$, we write $D_k(t) = \set{u:t\in\tilde N_k(u)}$ and $d_t=|D_k(t)|$. Recall that our goal is to let $t$ send  the edge set $I_k(t)$ to each vertex in $D_k(t)$. The messages to be sent are prepared as follows: We create $d_t$ duplicates of the edge set $I_k(t)$ and thereafter label each duplicate according to the vertex IDs in $D_k(t)$. To initiate this procedure, we require that each leader machine $\widehat M(t)$ learns the value of $d_t$ and the edge set $I_k(t)$. We do so in the following way: 
    \begin{itemize}
        \item Learning the value $d_t$ is done by counting the total number of packets $(\mathcal K, t, \cdot)$ stored across the machines $M(t)$. We execute this by calling the prefix sum sub-routine in \Cref{lemma:mpc-toolbox} and sorting such that the leader machine $\widehat M(t)$ knows the sum $d_t$.
        
        \item Learning the edge set $I_k(t)$ is done by sorting the edge packets in $M(t)$ by their edge weights (using \Cref{lemma:mpc-toolbox}). We then collect the edges $I_k(t)$ onto the leader machine $\widehat M(t)$ with one additional round (which can be done as $|I_k(t)|\leq k$ and $k\in O(n^\gamma)$ by assumption). 
    \end{itemize}
    Now, we are ready to duplicate and label the edges.
    \begin{enumerate}
        \item We first apply \Cref{lemma:mpc-duplicate}  to create $d_t$ duplicates of $I_k(t)$ in the form $(I_k(t), 1), (I_k(t), 2), \ldots, (I_k(t), d_t)$.
        \item We next apply \Cref{lemma:mpc-label} to relabel each duplicate with a vertex ID in $D_k(t)$. In other words, denoting $D_k(t)=\set{u_1,u_2,\ldots,u_{d_t}}$, we now have the duplicates in the form of $(I_k(t), u_1),(I_k(t), u_2),\ldots, (I_k(t),  u_{d_t})$.
        \item Lastly, for each duplicate $(I_k(t), u_j)$, we relabel every edge packet in $I_k(t)$ from $(\mathcal E, t,\set{t,v},\omega_G(t,v))$ to  $(\mathcal E, u_j,\set{t,v},\omega_G(t,v))$.
    \end{enumerate}  
    Crucially, since $\sum_{t\in V}d_t\leq nk$, we only require $O(nk^2)$ additional total memory to apply \Cref{lemma:mpc-duplicate,lemma:mpc-label} which is satisfied by our constraints. Overall, this entire step runs in $O(1)$ rounds. 
    
    \item[Step 3: Send and Receive Edges.] Similar to Step 1, we again sort the edge packets by the second item in the tuple using \Cref{lemma:mpc-toolbox}. At this point, we can discard the $k$-nearest packets. After this, for each vertex $u\in V$, we have the packets $(\mathcal E, u,\set{t,v},\omega_G(t,v))$ for the edges $\set{t,v}\in T_k(u)$ stored on a set of contiguous machines $M(u)$ (Note that these are not necessarily the same machines in Step 2). 
    
    Recalling that $|T_k(u)| \leq k^2$ and $k^2\in O(n^\gamma)$ by assumption, the entire set of edges $T_k(u)$ can be collected onto the leader $\widehat M(u)$ with an additional $O(1)$ rounds (note that we do not need the entire set of incident edges of $u$, which could have size larger than $k^2$). Overall, this step requires $O(1)$ rounds.  
    
    \item[Step 4: Add Edges to $H$.] Each leader machine $\widehat M(u)$ (from Step 3) now locally computes the distances $\delta_u(v)$ for all $v\in B_k(u)$. Subsequently, for $v\in B_k(u)$ (and $\delta_u(v)\neq \infty$), the machine $\widehat M(u)$ creates a new packet $(\mathcal E, u,\set{u,v},\delta_u(v))$ to be included into the edge set $H$. 
\end{description}

This concludes the algorithm description. We make a few concluding remarks which are not critical to the algorithm description, but are worth mentioning for clarity: 
\begin{itemize}
    \item In the last step of the algorithm, it suffices for each machine $\widehat M(u)$ to just mark which edges are included in $H$. We do not specify that the edges in $H$ need to be on a set of contiguous machines. 
    \item To remove duplicate edges in $H$, it suffices to perform one final sorting step using \Cref{lemma:mpc-toolbox} to group all edge weights corresponding to each edge $\set{u,v} \in H$ onto a contiguous set of machines, with the weights within each group sorted. Then, each group can locally retain the smallest edge weight using the ``finding minimum'' procedure in \Cref{lemma:mpc-toolbox}, and discard the rest.
\end{itemize}

We summarise the above discussion in the following proof of \Cref{lemma:k-nearest-hopset}.
\begin{proof}[Proof of \Cref{lemma:k-nearest-hopset}]
	We apply the algorithm as described above to construct the edge set $H$. By the algorithm description, the edge set $H$ constructed by the algorithm is exactly the edge set $H$ specified in \Cref{claim:k-nearest-hopset-correctness}, and is therefore a $k$-nearest $O(a\log d)$-hopbound hopset for $G$. Using $\tilde O(|E| + nk^2)$ total memory, each step of the algorithm runs in $O(1)$ rounds. Therefore, the entire algorithm runs in $O(1)$ rounds as well, as desired. 
\end{proof}

\section{Extracting $k$-Nearest Neighbourhoods}\label{section:extracting-neighbourhoods}
In this section we prove the following lemma, which gives us $k$-nearest $h$-hop neighbourhoods according to the input graph. In particular, if we have a $c$-stretch $h$-hopbound emulator for a graph $G$, we can obtain $c$-approximate $k$-nearest neighbourhoods for all vertices in $G$ using $O(\log h)$ rounds. Alternatively, if we have a $k$-nearest $h$-hopbound hopset $H$ for $G$, applying the lemma on $G\cup H$ gives exact $k$-nearest neighbourhoods for all vertices in $G$. That is, in conjunction with \Cref{lemma:k-nearest-hopset}, we have a useful tool for extracting exact $k$-nearest neighbourhoods from approximate $k$-nearest neighbourhoods in the sublinear MPC model.

\exactKNearest*

\paragraph{Notation.} Before describing the construction, it is convenient to define some notation for matrices. We call $C$ the \emph{adjacency matrix} of a graph $H=(V_H,E_H)$ if $C_{u,v}=\omega_H(u,v)$ for all vertices $u,v\in V_H$. We denote $\overline C$ as the \emph{filtered} matrix of $C$, where only the $k$ smallest elements of each row in $C$ are kept, breaking ties by ID (with the remaining elements set to $\infty$). We specifically denote $\mathcal A$ to be the adjacency matrix of the input graph $G$, i.e., $\mathcal A_{u,v}=\omega_{G}(u,v)$ for all vertices $u,v\in V$.

In this section, we only consider matrix multiplication under the min-plus semiring. Observe that for all integers $i \geq 1$ and for all vertices $u,v\in V$, we have $\mathcal A^i_{u,v}=d_G^{(i)}(u,v)$. In other words, learning the matrix $\overline{\mathcal A^i}$ is equivalent to learning the neighbourhoods $N^{(i)}_k(u)$ for all vertices $u\in V$. 

\paragraph{Construction.} The construction is an application of the well-known graph exponentiation technique: Each vertex $u$ maintains a vertex set which describes its current view of its $k$-nearest neighbourhood. Initially, this vertex set is simply $N^{(1)}_k(u)$, the $k$ nearest adjacent vertices to $u$ in $G$. During the $i$th iteration, the vertex $u$ starts with the $k$-nearest $2^i$-hop neighbourhood $N^{(2^i)}_k(u)$, and gathers the sets $N^{(2^i)}_k(t)$ for every vertex $t\in N^{(2^i)}_k(u)$. Using this new information, $u$ is able to exactly construct the $2^{i+1}$-hop neighbourhood $N^{(2^{i+1})}_k(u)$. We repeat this process until every $u\in V$ knows the $h$-hop neighbourhood $N_k^{(h)}(u)$.

More concretely, we define the sequences of matrices $\set{B_i}$ for $i\in\set{0,1,2,\ldots}$ as follows:
\begin{align*}
    B_0 &= \overline {\mathcal A},\\
    B_i &= \overline{B_{i-1}^2}.
\end{align*}
For every vertex $u\in V$, the entries of row $u$ in $B_{i+1}$ correspond to the vertices and values that $u$ learns at the end of the $i$th iteration. In other words, we want to show that $B_i=\overline{\mathcal A^{2^i}}$ for all integers $i\geq 0$.

\paragraph{Correctness.} We adopt a useful lemma from \citet{bui2024improved}, which claims that for any adjacency matrix $C$ representing a graph with positive integer weights, the $k$ smallest entries in each row of $(\overline C)^i$ are the $k$ smallest entries in the corresponding row of $C^i$.
\begin{lemma}[\cite{bui2024improved}]\label{lemma:filter-at-end}
    Let $C$ be the adjacency matrix of a graph with positive integer weights. Then we have $\overline{(\overline C)^i}=\overline{C^i}$ for all integers $i\geq 1$.
\end{lemma}
We now proceed to show that $B_i=\overline{\mathcal A^{2^i}}$. There is a slight difference between \Cref{lemma:filter-at-end} and the construction of the matrices $\set{B_i}$: For the former, the filtering occurs only \emph{at the end} of all $i$ iterations; for the latter, the filtering occurs \emph{between} each iteration. However, by applying \Cref{lemma:filter-at-end} between each iteration, we can show the desired claim.
\begin{lemma}\label{lemma:filter-every-iteration}
    For all integers $i\geq 0$, we have $B_i=\overline{\mathcal A^{2^i}}$.
\end{lemma}
\begin{proof}
    We prove this by induction. In the base case $i=0$, we have $B_0=\overline {\mathcal A} = \overline{\mathcal A^{2^0}}$ as claimed. Proceeding to the inductive case: Suppose the claim holds for $i=j$. Invoking our induction hypothesis, we have that
    $$
        B_{j+1}=\overline{B_j^2}=\overline{\left(\overline{\cA^{2^j}}\right)^2}.
    $$
    Observe that the matrix $\cA^{2^j}$ represents the adjacency matrix of the graph $G'=(V,E')$ with edge weights $\omega_{G'}(u,v)=d_G^{(2^j)}(u,v)$. Therefore, we can apply \Cref{lemma:filter-at-end} on the matrix $A^{2^j}$ with $i=2$ to obtain
    $$
    B_{j+1}=\overline{\left(\overline{\cA^{2^j}}\right)^2} = \overline{(\cA^{2^j})^2}=\overline{(\cA^{2^{j+1}})},
    $$
    as desired.
\end{proof}

\subsection{Algorithm Description}

We describe the algorithm in iterations: For each iteration $i\in\set{0,1,2,\ldots}$, we start with the matrix $B_{i}$, and by communicating minimal but sufficient information between vertices, we show that we can construct the matrix $B_{i+1}$ in a constant number of rounds. For ease of notation, for all $i\in\set{0,1,2,\ldots}$ and for all vertices $u\in V$, we define 
$$
B_i(u) = \left\{\left(u,t,d_G^{(2^{i})}(u,t)\right):t\in N_k^{(2^i)}(u)\right\}.
$$
In other words, $B_i(u)$ represents the entries of the row $u$ of $B_i$ (by \Cref{lemma:filter-every-iteration}). To perform the matrix multiplication, observe that to determine the values $B_{i+1}(u)$ for any vertex $u\in V$, it suffices to have the following information:
\begin{itemize}
    \item The values $B_{i}(u)$, and
    \item The values $B_{i}(t)$ for all vertices $t\in N_k^{(2^{i})}(u)$.
\end{itemize}
By extension, recalling that $|B_{i}(v)|\leq k$ for all vertices $v\in V$, we require $O(k^2)$ memory per machine in order to construct the values $B_{i+1}(u)$ entirely locally. We note that the algorithm description for each iteration is nearly identical to the algorithm for \Cref{lemma:k-nearest-hopset}. Therefore, for the sake of brevity, we describe only the essential details.

\paragraph{Initialisation.} To determine the values $B_0(u)$ for every vertex $u\in V$, it suffices to sort the edges such that the $k$ shortest edges incident on $u$ (breaking ties by ID) are stored on a single machine. Using \Cref{lemma:mpc-toolbox}, this requires $O(1)$ rounds.

\paragraph{A Single Iteration.} Our procedure for each iteration proceeds as follows: In iteration $i\in \set{0,1,2,\ldots}$, every vertex $u\in V$ initially has the values $B_{i}(u)$. 
\begin{description}
    \item[Step 1. Sending Requests.] For every vertex $u\in V$, it sends a request for the values $B_{i}(t)$ for each vertex $t\in N_k^{(2^{i})}(u)$. Using \Cref{lemma:mpc-toolbox}, this step requires $O(1)$ rounds.
    
    \item[Step 2. Duplicating Information.] For every vertex $t\in V$, denote $D_i(t)=\set{u:t\in N_k^{(2^{i})}(u)}$ and $d_t = |D_i(t)|$. Recall that our goal is to let $t$ send the values $B_{i}(t)$ to every vertex in $D_{i}(t)$. This proceeds as follows: 
    \begin{enumerate}
        \item We duplicate $d_t$ copies of the values $B_{i}(t)$.
        \item We label each duplicate with a vertex ID in $D_i(t)$.
    \end{enumerate}
    In other words, denoting $D_i(t)=\set{u_1,u_2,\ldots,u_{d_t}}$, we have constructed the values $(B_{i}(t),u_1),(B_{i}(t),u_2),\ldots,(B_{i}(t),u_{d_t})$. Using \Cref{lemma:mpc-duplicate} and \Cref{lemma:mpc-label}, this step requires $O(1)$ rounds. Moreover, recalling that $\sum_{t\in V}d_t\leq nk$, this step requires $O(nk^2)$ additional total memory.

    \item[Step 3. Performing Matrix Multiplication. ] For every vertex $u\in V$, we sort the values across the machines such that we have the following information stored on a \emph{single} machine $\widehat M(u)$:
    \begin{itemize}
        \item The values $B_{i}(u)$, and
        \item The values $B_{i}(t)$ for all vertices $t\in N_k^{(2^{i})}(u)$.
    \end{itemize}
    Finally, using its local information, $\widehat M(u)$ is able to construct the values $B_{i+1}(u)$ as specified in the earlier construction. Using \Cref{lemma:mpc-toolbox}, this step requires $O(1)$ rounds. All remaining information is then discarded.
\end{description}

\begin{claim}\label{claim:exact-k-algorithm}
    Running $T$ consecutive iterations of the above algorithm in the sublinear MPC model requires $O(T)$ rounds and $\tilde O(|E| + nk^2)$ total memory (independent of $T$).
\end{claim}
\begin{proof}
    We first argue about the round complexity. For each iteration, every step requires only $O(1)$ rounds. Since there is a constant (three) number of steps, each iteration only requires $O(1)$ rounds, and running $T$ consecutive iterations requires $O(T)$ rounds.

    We next argue about the total memory constraints. Observe that each iteration requires only $O(nk^2)$ additional memory: Step 1 requires $O(nk)$ additional memory, while Steps 2 and 3 require $O(nk^2)$ additional memory. Naively, we might assume that we therefore require $O(T\cdot nk^2)$ additional total memory overall. However, recall that after Step 3 in each iteration, we discard all remaining information except for the values of the matrix $B_{i+1}$. In other words, there is no accumulation of memory usage between iterations, and the total memory required remains as $\tilde O(|E| + nk^2)$, regardless of the number of iterations. 
\end{proof}

\paragraph{The Final Proof.} We summarise our discussion for this section with the proof of \Cref{lemma:exact-k-nearest}.
\begin{proof}
     We run the above algorithm for $r$ iterations, which requires $O(r)$ rounds and $\tilde O(|E|+nk^2)$ total memory by \Cref{claim:exact-k-algorithm}. By the end of the algorithm, we have constructed the entries of the matrix $B_{r}=\overline{A^h}$. Now applying \Cref{lemma:filter-every-iteration}, we observe that this is equivalent to the values $\set{(u,t,d_G^{(h)}(u,t)):u\in V\wedge t\in N_k^{(h)}(u)}$, and we are done.
\end{proof}

\section{Skeleton Graphs}\label{section:skeleton-graph}

In this section, we prove the following lemma, which provides a fast construction for the skeleton graph as defined in \Cref{definition:skeleton-graph}.

\skeletonGraphImp*

Our skeleton graph is based on the construction proposed by \cite{bui2024improved}, which has exactly the properties described in \Cref{definition:skeleton-graph}. However, the main (and only) issue with this construction is that the number of edges in the skeleton graph is poorly bounded. Therefore, we suggest a modified construction that maintains the same properties while bounding the number of edges by $O(|E|k^2)$. In doing so, we require adding additional vertices to the skeleton graph; however, we show that the number of additional vertices is at most $n/k$, and therefore the number of vertices is still bounded above by $\tilde O(n/k)$.

\paragraph{Original Construction.} We first discuss the construction proposed by \cite{bui2024improved}. We follow a clustering approach: The cluster centers are the vertex set $C\subseteq V$, chosen such that $|C| \in O((n\log k)/k)$ and $C\cap N_k(u)\neq\emptyset$ for all $u\in V$ (i.e., $C$ is a hitting set for the $k$-nearest neighbourhoods). The vertex set $C$ is constructed by first adding each vertex $w\in V$ to $C$ independently with probability $p \in \Theta((\log k)/k)$ and subsequently adding any remaining vertices $w'$ with $C\cap N_k(w')=\emptyset$. 

For all vertices $u\in V$, we add $u$ to the cluster with the nearest center. More concretely, we choose a $c(u)\in C$ such that the distance $d_G(u,c(u))$ is the smallest among all vertices in $C\cap N_k(u)$ (breaking ties arbitrarily). We now define the graph $G^*_S = (V^*_S,E^*_S)$ as follows:
\begin{itemize}
	\item $V_S^* = C$, and
	\item For all vertices $u,v\in V$, we add the edge $\set{c(u),c(v)}\in E^*_S$ if there exists a vertex $t\in V$ such that all of the following hold:
	\begin{itemize}
		\item $t\in N_k(u)$ (i.e., $t$ is in the $k$-nearest neighbourhood of $u$), and
		\item $\set{t,v}\in E$ (i.e., $t$ is adjacent to $v$ in $G$) or $t=v$.
	\end{itemize}
	We define the edge weight $\omega_{G_S^*}(c(u),c(v)) := d_G(c(u),u) + d_G(u,t) + \omega_G(t,v) + d_G(v,c(v))$. If there are multiple candidate edge weights, choose the minimum weight.
\end{itemize}
We now give a proof sketch of the correctness of $G_S^*$. Namely, for any two vertices that are not in the $k$-nearest of each other, there exists a path in $E_k\cup E_S^*$ that satisfies the two properties stated in \Cref{definition:skeleton-graph}.

\begin{claim}[\cite{bui2024improved}]\label{claim:skeleton-stretch-original}
$G_S^*$ is a skeleton graph for $G$ as defined in \Cref{definition:skeleton-graph}.
\end{claim}
\begin{proof}[Proof Sketch]
To show that $G_S^*$ is a skeleton graph for $G$, we need to verify the following properties:
\begin{itemize}
    \item $|V_S^*|\in O((n\log k)/k)$.
    \item $d_{G_S^*}(u,v)\geq d_G(u,v)$ for all $u,v\in V^*_S$.
    \item If $\omega_G(u,v)\leq g(n)$ for all $\set{u,v}\in E$ for some function $g$, then $\omega_{G_S^*}(x,y)\leq 3k\cdot g(n)$ for all $\set{x,y}\in E^*_S$.
    \item If $G$ is connected, then so is $G_S^*$.
    \item If $u\notin N_k(v)$ and $v\notin N_k(u)$, then there exists a path $P$ in $E_k\cup E_S^*$ with the properties (P1) and (P2). 
\end{itemize} 
In this proof sketch, we focus on showing the last property; the remaining properties are relatively easier to verify. 

Consider any two vertices $u,v\in V$ such that $u\notin N_k(v)$ and $v\notin N_k(u)$. We first consider the shortest path from $u$ to $v$ in $G$, which we denote as $u=u_0\rightsquigarrow t_0\rightarrow u_1\rightsquigarrow t_1\rightarrow u_2\rightsquigarrow t_2\rightarrow\ldots \rightarrow u_\ell\rightsquigarrow t_\ell =v$ chosen such that $t_i\in N_k(u_i)$ while $u_{i+1}\notin N_k(u_i)$. By definition, observe that there must exist a path $Q$ in $E_S^*$ of the form $c(u_0)\rightarrow c(u_1)\rightarrow c(u_2)\rightarrow\ldots\rightarrow c(u_\ell)\rightarrow c(t_\ell)$. In particular, we observe that for $i\leq \ell - 2$,
\begin{align*}
    \omega_{G_S^*}(c(u_i),c(u_{i+1}))&\leq d_G(c(u_i),u_i)+d_G(u_i,t_i)+\omega_G(t_i,u_{i+1})+d_G(u_{i+1}, c(u_{i+1}))\\
    &\leq d_G(u_i,u_{i+1}) + d_G(u_i,t_i) + \omega_G(t_i,u_{i+1}) + d_G(u_{i+1},u_{i+2})\\
    &=2\cdot d_G(u_i,u_{i+1}) + d_G(u_{i+1},u_{i+2}).
\end{align*}
To bound the edge weights $c(u_{\ell-1})\rightarrow c(u_\ell)$ and $c(u_{\ell})\rightarrow c(t_\ell)$, we observe that 
$$
d_G(u_\ell,c(u_\ell))\leq d_G(u_\ell,c(t_\ell))\leq d_G(u_\ell,t_\ell)+d_G(t_\ell,c(t_\ell)) \leq d_G(u_\ell,t_\ell)+ d_G(u,v).
$$
Therefore, we conclude that the length of path $Q$ is at most 
\begin{align*}
 &\sum_{i=0}^{\ell-1}\omega_{G_S^*}(c(u_i),c(u_{i+1}))+\omega_{G_S^*}(c(u_\ell),c(t_\ell))\\
 &\leq 2\cdot d_G(u_0,u_1)+3\left(\sum_{j=1}^{\ell-1}d_G(u_j,u_{j+1}) + d_G(u_\ell,t_\ell)+d_G(u,v)\right).   
\end{align*}
% $$
% \sum_{i=0}^{\ell-1}\omega_{G_S^*}(c(u_i),c(u_{i+1}))+\omega_{G_S^*}(c(u_\ell),c(t_\ell))\leq 2\cdot d_G(u_0,u_1)+3\left(\sum_{j=1}^{\ell-1}d_G(u_j,u_{j+1}) + d_G(u_\ell,t_\ell)+d_G(u,v)\right).
% $$
We now construct the final path $P$ from $u$ to $v$ in $E_k\cup E^*_S$: $u\rightarrow c(u)$ using an edge in $E_k$, $c(u)\rightsquigarrow c(v)$ using path $Q$ (contained entirely in $E_S^{*}$), and $c(v)\rightarrow v$ using an edge in $E_k$. By construction, we observe that $P$ obeys property (P2) of \Cref{definition:skeleton-graph}. Additionally, noting that $d_G(u,c(u))\leq d_G(u_0,u_1)$ and $d_G(v,c(v))\leq d_G(u,v)$, we observe the following:
\begin{itemize}
    \item The length of the path $u\rightarrow c(u)\rightsquigarrow c(v)$ is at most $6\cdot d_G(u,v)$.
    \item The length of the path $c(v)\rightarrow v$ is at most $d_G(u,v)$.
\end{itemize}
Therefore, we conclude that the length of $P$ is at most $7\cdot d_G(u,v)$. Moreover, using the fact that $d_{G_S^*}(u,v)\geq d_G(u,v)$ for all $u,v\in V^*_S$, we have that the length of $P$ is indeed a $7$-approximation of the distance $d_G(u,v)$, and thus obeys property (P1) of \Cref{definition:skeleton-graph}.
\end{proof}

We now discuss our modified construction, which seeks to provide a better upper bound for the size of the skeleton graph while preserving the essential properties as stated in \Cref{definition:skeleton-graph}.

% \paragraph{The Problem, and the Fix.} They key barrier to implementing the above construction is the size of the skeleton graph: It is unfortunately trivial to construct graphs where the size of $E_S$ is larger than $O(|E|k^2)$. For example, consider the case of a vertex $t\in V$ such that
% \begin{itemize}
% 	\item $t$ is in the $k$-nearest neighbourhood of $n$ vertices, and
% 	\item $t$ is adjacent to $n$ vertices.
% \end{itemize}
% Such a vertex would add $O(|V_S|^2) = \tilde O(n^2/k^2)$ edges, and thus requires more memory than we can afford. Therefore, directly implementing this construction directly is not feasible. To resolve this, we present a surprisingly simple fix: For any vertex that would add an excessive number of edges (such as the vertex $t$ in the aforementioned example), we can instead add this vertex to the skeleton graph (and in doing so, add a significantly smaller number of edges). This comes at the cost of a larger vertex set, but by choosing such vertices carefully, we show that the number of vertices is still bounded by $O((n\log k)/k)$ while the number of edges is now reduced to $O(|E|k^2)$.  We make this idea more concrete in the next paragraph.\yijun{note: this has been discussed in the technical overview - remove or shorten (check also other places)}

\paragraph{Modified Construction.} For every vertex $t\in V$, we call $t$ a \emph{bad vertex} if there exist more than $k^2$ distinct vertices $u\in V$ such that $t\in N_k(u)$ (i.e., $|\set{u:t\in N_k(u)}| > k^2$); otherwise, we call $t$ a \emph{good vertex}. Denote $B\subseteq V$ as the set of bad vertices in $V$. We also define the vertex set $C$ and the vertices $c(u)$ for all $u\in V$ identically to the original construction. Before defining the modified skeleton graph, we first define two weighted edge sets $E_C$ and $E_B$. 

\begin{description}
\item [Edge Set $E_C$.] Defined over the vertices $C$. For all vertices $u,v\in V$, we add the edge $\set{c(u),c(v)}$ to $E_C$ if there exists a vertex $t\in V$ such that all of the following hold:
\begin{itemize}
	\item $t\in N_k(u)$ (i.e., $t$ is in the $k$-nearest neighbourhood of $u$),
	\item $\set{t,v}\in E$ (i.e., $t$ is adjacent to $v$ in $G$) or $t=v$, and
	\item $t\notin B$ (i.e., $t$ is a good vertex). 
\end{itemize}
We define the edge weight $\omega_{E_C}(c(u),c(v)) := d_G(c(u),u) + d_G(u,t) + \omega_G(t,v) + d_G(v,c(v))$. If there are multiple candidate edge weights, choose the minimum weight. 

Note that the definition of $E_C$ is almost identical to the earlier definition of $E_S^*$, with the additional requirement that $t\notin B$.

\item [Edge Set $E_B$.] Defined over the vertices $C\cup B$. For all bad vertices $t\in B$, we add the edge $\set{s,t}$ to $E_B$ if either one of the following conditions holds:
\begin{itemize}
	\item There exists a vertex $u\in V$ such that $s = c(u)$ and $t\in N_k(u)$. In this case, we set the edge weight $\omega_{E_B}(s,t):=d_G(s,u)+d_G(u,t)$.
	\item There exists a vertex $v\in V$ such that $s = c(v)$ and either $\set{t,v}\in E$ or $t=v$. In this case, we set the edge weight $\omega_{E_B}(s,t):=d_G(s,v)+\omega_G(v,t)$.
\end{itemize}
If both conditions hold simultaneously or if there are multiple candidate edge weights, we always choose the minimum weight. 
\end{description}

We are now ready to define the graph $G_S = (V_S,E_S)$ as follows: We define $V_S = C\cup B$ and $E_S = E_C\cup E_B$ (such that for any parallel edges, choose the edge with the smaller edge weight). This concludes the construction, and we claim that $G_S$ is a skeleton graph for $G$ with the additional property that $|E_S|\in O(|E|k^2)$.

\subsection{Correctness} 
We separate the proof of correctness into multiple smaller claims. We first argue on the size of the vertex set $V_S$. Namely, we show that the number of bad vertices is upper bounded by $n/k$. 
\begin{claim}\label{claim:skeleton-vertex-size}
	$|B|\leq n/k$. Consequently, we have $|V_S|\in O((n\log k)/k)$. 
\end{claim}
\begin{proof}
By definition, we have $|C|\in O((n\log k)/k)$. Thus to show that $|V_S|\in O((n\log k)/k)$, it suffices to obtain a good bound for $|B|$. Consider the following ``token distribution'' procedure: For every vertex $u\in V$, $u$ gives the vertex $t\in V$ a token if $t\in N_k(u)$. Observe that the set $B$ is equivalent to the set of vertices with more than $k^2$ tokens. Additionally, we observe that there are at most $nk$ tokens across all vertices. Therefore by the pigeonhole principle, we can only have at most $nk/k^2=n/k$ bad vertices. In other words, $|B|\leq n/k$ as claimed. 
\end{proof}

We next argue on the size of the edge set $E_S$. In particular, we will show that both $E_C$ and $E_B$ have at most $O(|E|k^2)$ edges, and therefore $E_S$ has $O(|E|k^2)$ edges as well.
\begin{claim}\label{claim:skeleton-edge-size}
	$|E_S|\in O(|E|k^2)$.
\end{claim}
\begin{proof}
A useful observation here is that in any simple undirected graph, the number of non-isolated vertices (i.e., vertices with at least one adjacent edge) is at most $2|E|$; in particular, $|B|\leq 2|E|$.

\begin{description}
    \item[Bounding $|E_B|$.] For every bad vertex $t\in B$, define the integers $a_t = |\set{v:\set{v,t}\in E\text{ or } v=t}|$ and $b_t = |\set{u:t\in N_k(u)}|$ (i.e., $a_t$ is the number of vertices adjacent to $t$ in $G$ (and itself) and $b_t$ is the number of vertices that $t$ is in the $k$-nearest of). 
    
    We find that $\sum_{t\in B}a_t \leq 2|E|+|B|$ and $\sum_{t\in B}b_t\leq 2|E|k$ -- to observe the latter, sum the $k$-nearest neighbourhoods of each non-isolated vertex in $V$. Additionally, we observe from the definition of $E_B$ that for every bad vertex $t\in B$, we add at most $(a_t+b_t)$ edges. Therefore, the size of $|E_B|$ is bounded by 
    \begin{equation*}
        |E_B| \leq \sum_{t\in B}(a_t+b_t) \leq 2|E| +|B| + 2|E|k \in O(|E|k^2).
    \end{equation*}
    \item[Bounding $|E_C|$.] Fix any two non-isolated vertices $t,v\in V$ such that $t$ is a good vertex, and either $\set{t,v}\in E$ or $t=v$. By definition of a good vertex, there are at most $k^2$ vertices $u$ such that $t\in N_k(u)$. Therefore, it follows from the definition of $E_C$ that we only construct at most $k^2$ edges $\set{c(u),c(v)}$ with edge weight $d_G(c(u),u)+d_G(u,t)+\omega_G(t,v)+d_G(v,c(v))$. Since there are at most $O(|E|)$ possible choices for the vertices $t$ and $v$, we can therefore conclude that there are at most $O(|E|k^2)$ edges in $E_C$. \qedhere
\end{description}
\end{proof}

We now argue on the value of the edge weights of the edge set $E_S$. 
\begin{claim}\label{claim:skeleton-weights}
	Suppose that the edge weights $\omega_G$  are bounded above by the function $f(n)$. Then the edge weights $\omega_{G_S}$ are bounded above by the function $g(n) = 3k\cdot f(n)$. 
\end{claim}
\begin{proof}
    Before bounding the edge weights, we first make a preliminary but important observation: For any pair of vertices $u,v\in V$ such that $v\in N_k(u)$, we must have that $d_G(u,v)\leq (k-1) f(n)$. To see this, note that the shortest path from $u$ to $v$ in $G$ necessarily has at most $(k-1)$ edges; otherwise the first $k$ vertices along this path, including $u$ and excluding $v$, form a better $k$-nearest neighbourhood for $u$.

    We now recall the definitions of the edge weights for $E_C$ and $E_B$.
    \begin{itemize}
        \item For any edge $\set{s_a,s_b}\in E_C$, recall that the edge weight 
        $$
        \omega_{E_C}(s_a,s_b)=d_G(s_a,u)+d_G(u,t)+\omega_G(t,v)+d_G(v,s_b),
        $$ such that $t,s_a\in N_k(u)$ and $s_b\in N_k(v)$. By the earlier observation, we have that $d_G(s_a,u) + d_G(u,t) + d_G(v,s_b)\leq (3k-3) f(n).$ Including the weight $\omega_G(t,v)$, we find that $\omega_{E_C}(s_a,s_b)\leq (3k-2)\cdot f(n)\leq 3k\cdot f(n)$.

        \item For any edge $\set{s,t}$ in $E_B$ where $t\in B$, recall that there exists some vertex $x\in V$ such that $c(x)=s$, and the edge weight $\omega_{E_B}(s,t)$ is either one of the following:
        \begin{itemize}
            \item $\omega_{E_B}(s,t)=d_G(s,x)+d_G(x,t)$ such that $t\in N_k(x)$; or
            \item $\omega_{E_B}(s,t)=d_G(s,x)+\omega_G(x,t)$ such that $\set{x,t}\in E$ (or $t=x$).
        \end{itemize}
        In the former case, the weight is bounded by $(2k-2)\cdot f(n)\leq 2k\cdot f(n)$, while in the latter case, the weight is bounded by $k\cdot f(n)$. Therefore, we find that $\omega_{E_B}(s,t)\leq \max[2k,k]\cdot f(n)=2k\cdot f(n)$.
    \end{itemize}
    Overall, we find that for all edges $\set{x,y}\in E_S$, we must have that $\omega_{G_S}(x,y)\leq \max[3k,2k]\cdot f(n)=3k\cdot f(n)$.
\end{proof}

The following claim is a lower bound on the distances in $G_S$. Namely, for any pair of vertices in $G$, we show that their distances do not shrink in $G_S$. 
\begin{claim}\label{claim:skeleton-no-shrink}
	For any pair of vertices $u,v\in V_S$, we have $d_G(u,v)\leq d_{G_S}(u,v)$.
\end{claim}
\begin{proof}
We recall the definitions of the edge weights for $E_C$ and $E_B$.
\begin{itemize}
    \item Recalling the definition of $E_C$, observe that for any edge $\set{x,y}\in E_C$, the edge weight $\omega_{E_C}(x,y)$ is equivalent to the length of the path $x\rightsquigarrow u\rightsquigarrow t\rightarrow v\rightsquigarrow y$ in $G$ (for some vertices $u,v,t$ as per the definition of $E_C$). Namely, the length of this path is necessarily at least $d_G(x,y)$.

    \item Similarly, recalling the definition of $E_B$, observe that for any edge $\set{s,t}\in E_B$ and fixing $t$ to be the bad vertex, the edge weight $\omega_{E_B}(s,t)$ is equivalent to the length of the path $s\rightsquigarrow u\rightsquigarrow t$ or $s\rightsquigarrow v\rightarrow t$ (for some vertices $u,v$ as per the definition of $E_B$). Namely, the length of this path is necessarily at least $d_G(s,t)$.
\end{itemize}
Therefore, for any edge $\set{x,y}\in E_S$, we necessarily have that $\omega_{E_S}(x,y) \geq d_G(x,y)$. By extension, for any vertices $u,v\in V_S$, the length of any path from $u$ to $v$ contained in $E_S$ is at least $d_G(u,v)$. From this observation, the claim directly follows.
\end{proof}

We now show the most important claim: Between any pair of vertices in $G$ that are not in the $k$-nearest neighbourhood of each other, there exists a path in $E_k\cup E_S$ with the properties (P1) and (P2) described in \Cref{definition:skeleton-graph}. Before proving this claim, it helps to prove the following result which gives us a way to transform a path in the original construction $G_S^*$ to a path in the modified construction $G_S$ of at most the same path distance. 
\begin{claim}\label{claim:original-to-modified-skeleton-path-transform}
    For any edge $\set{x,y}\in E_S^*$ in the original construction $G_S^*$, the shortest path from $x$ to $y$ in the modified construction $G_S$ is at most the edge weight in $G_S^*$, i.e., $d_{G_S}(x,y) \leq \omega_{G_S^*}(x,y)$.
\end{claim}
\begin{proof}
    Recall the original construction $G_S^*$: For any edge $\set{x,y}\in E_S^*$, the edge weight is given by $\omega_{G_S^*}(x,y)=d_G(x,u)+d_G(u,t)+\omega_G(t,v)+d_G(v,y)$ for some vertices $t,u,v\in V$ such that $c(u)=x$ and $c(v)=y$. Now consider these same vertices but in the modified construction $G_S$. There are two possible cases:
    \begin{description}
        \item[Case 1: $t$ is a good vertex.] Recalling the construction for $E_C$, we must have that the edge $\set{x,y}$ is also added into $E_C$. Since we always take the minimum weight of parallel edges, we have that the edge weight $\omega_{E_C}(x,y)\leq \omega_{G_S^*}(x,y)$.
        \item[Case 2: $t$ is a bad vertex.]  Recalling the construction for $E_B$, we must have that the edges $\set{x,t}$ and $\set{t,y}$ are added into $E_B$. Similarly, since we always take the minimum weight of parallel edges, we have that the weights $\omega_{E_B}(x,t)\leq d_G(x,u)+d_G(u,t)$ and $\omega_{E_B}(t,y)\leq \omega_G(t,v)+d_G(v,y)$. Namely, this implies that $\omega_{E_B}(x,t)+\omega_{E_B}(t,y)\leq \omega_{G_S^*}(x,y)$.\qedhere
    \end{description}
\end{proof}

We now adapt the earlier proof of \Cref{claim:skeleton-stretch-original} and apply \Cref{claim:original-to-modified-skeleton-path-transform} to prove the following claim.
\begin{claim}\label{claim:skeleton-stretch-modified}
	For any pair of vertices $u,v\in V$ such that $u\notin N_k(v)$ and $v\notin N_k(u)$, there exists a path $P$ in $E_S\cup E_k$ that satisfies the properties (P1) and (P2) given in \Cref{definition:skeleton-graph}.
\end{claim}
\begin{proof}
    First recall the proof of \Cref{claim:skeleton-stretch-original}: For the vertices $u,v$ as specified in the hypothesis, there exist edges $e_1=\set{u,c(u)}$, $e_2=\set{c(v), v}$ in $E_k$ and a path $Q$ from $c(u)$ to $c(v)$ in $G_S^*$ such that the path from $u$ to $v$ formed by $e_1$, $e_2$ and $Q$ has distance at most $7\cdot d_G(u,v)$. 
    
    By applying \Cref{claim:original-to-modified-skeleton-path-transform}, we construct the path $Q'$ from $c(u)$ to $c(v)$ in $G_S$ such that $\omega(Q')\leq \omega (Q)$. We now define the path $P$ in $E_k\cup E_S$ as follows: $u\rightarrow c(u)$ using the edge $e_1$, $c(u)\rightsquigarrow c(v)$ using $Q'$, and $c(v)\rightarrow v$ using the edge $e_2$. We now show that $P$ obeys properties (P1) and (P2) as defined in \Cref{definition:skeleton-graph}.
    \begin{description}
        \item[Proof of (P1).] We have that the length of $P$ is at most $7\cdot d_G(u,v)$ by the proof of \Cref{claim:skeleton-stretch-original}. By \Cref{claim:skeleton-no-shrink}, we also have that the length of $P$ is at least $d_G(u,v)$. Therefore, the length of $P$ is indeed a $7$-approximation of the distance $d_G(u,v)$.
        \item[Proof of (P2).] By construction, the first and last edges $u\rightarrow c(u)$ and $c(v)\rightarrow v$ are contained in $E_k$, while the remaining edges in $Q'$ are contained entirely in $E_S$.\qedhere
    \end{description}
\end{proof}

Lastly, connectivity is now straightforward to show.
\begin{claim}\label{claim:skeleton-connectivity}
    If $G$ is connected, then $G_S$ is also connected.
\end{claim}
\begin{proof}
    We first show that the vertices in $C$ are in the same connected component in $G_S$. Suppose not, i.e., there exist vertices $x,y\in C\subseteq V$ such that $x$ and $y$ are not in the same connected component in $G_S$. In particular, observe that $x = c(x)$ and $y = c(y)$ necessarily. There are two possible cases:
    \begin{description}
        \item[Case 1: $x\notin N_k(y)$ and $y\notin N_k(x)$.] Since $G$ is connected, we find that $x$ and $y$ are connected by a path, i.e., $d_G(x,y)\neq \infty$. By the proof of \Cref{claim:skeleton-stretch-modified}, we find that there is necessarily a path from $c(x)$ to $c(y)$ in $G_S$. 
        
        \item[Case 2: $x\in N_k(y)$ or $y\in N_k(x)$.] Suppose that $x\in N_k(y)$; the other case is similar. Notice that this necessarily satisfies the condition for an edge to be constructed between $c(x)$ and $c(y)$. Namely, if $x$ is a good vertex, then $\set{c(x),c(y)}\in E_C$; else,  $\set{x,c(y)}\in E_B$. 
    \end{description}
    In either case, we get a contradiction. Considering now any vertex $t\in V_S\setminus C$, note that $t$ is necessarily a bad vertex. In other words, we have that the edge $\set{t,c(t)}\in E_B$ by definition. This suffices to show the claim. 
\end{proof}

To summarise the above results, we present this final claim on correctness. 
\begin{claim}\label{claim:skeleton-correctness}
	The graph $G_S=(V_S,E_S)$ is a skeleton graph for $G$ as defined in \Cref{definition:skeleton-graph}. Additionally, we have $|E_S| \in O(|E|k^2)$. 
\end{claim}
\begin{proof}
By \Cref{claim:skeleton-vertex-size,claim:skeleton-no-shrink,claim:skeleton-weights,claim:skeleton-stretch-modified,claim:skeleton-connectivity}, it follows directly that $G_S$ is a skeleton graph for $G$ according to \Cref{definition:skeleton-graph}. Additionally by \Cref{claim:skeleton-edge-size}, we have $|E_S| \in O(|E|k^2)$ as desired.
\end{proof}

\subsection{Algorithm Description}
The most critical part of the algorithm description is constructing the edge sets $E_C$ and $E_B$ (as defined in the modified construction). The high-level idea is for every vertex $u\in V$ to communicate the following values:
\begin{itemize}
    \item $(c(u),d_G(u,c(u))+d_G(u,t))$ to every vertex $t\in N_k(u)$.
    \item $(c(u),d_G(u,c(u))+\omega_G(u,t))$ to every vertex $t$ adjacent to $u$ in $G$.
\end{itemize}
We then delegate the construction of $E_C$ to the good vertices, and the construction of $E_B$ to the bad vertices. More concretely, every $t\notin B$ constructs the edges $\set{c(u),c(v)}$ for every $u,v\in V$ such that $t\in N_k(u)$ and $\set{t,v}\in E$ (or $t=v$), which follows the familiar duplication--labelling procedure (as described previously in the algorithms for \Cref{lemma:k-nearest-hopset} and \Cref{lemma:exact-k-nearest}). For every vertex $t\in B$, it constructs the edges $\set{c(u),t}$ for every $u$ such that $t\in N_k(u)$, and the edges $\set{c(v),t}$ for every incident edge $\set{v,t}\in E$ (or $t=v$). Overall, we show that this requires $O(1)$ rounds with $\tilde O(|E|k^2)$ total memory. Similar to the previous algorithm descriptions, we use the term \textit{packets} to describe discrete units of information.

\begin{description}
\item[Step 0: Initialisation.] We may assume that for every edge $\set{u,v}\in E$, there exist two copies of this edge across the machines in the form $(\mathcal E_{\mathsf{adj}}, u, \set{u,v},\omega_{G}(u,v))$ and $(\mathcal E_\mathsf{adj}, v, \set{u,v},\omega_{G}(u,v))$. Likewise for every edge $\set{u,w}\in E_k$, there exist two copies of this edge in the form $(\mathcal {E}_\mathsf{k-near}, u, \set{u,w},d_G(u,w))$ and $(\mathcal E_\mathsf{k-near}, w, \set{u,w}, d_G(u,w))$. 

As a preliminary step, we first sort the edges using \Cref{lemma:mpc-toolbox} such that for every vertex $u\in V$, we have the packets
\begin{align*}
    \set{(\mathcal E_\mathsf{adj}, u, \set{u,v},\omega_G(u,v)):\set{u,v}\in E} && \text{and} && \set{(\mathcal {E}_\mathsf{k-near}, u, \set{u,w},d_G(u,w)): w\in N_k(u)} 
\end{align*}
stored on a contiguous set of machines $M(u)$ with leader $\widehat M(u)$. With another $O(1)$ rounds, we allow every leader $\widehat M(u)$ to know the set $N_k(u)$ (recall that $k\in O(n^\gamma)$ by assumption). 

\item[Step 1: Constructing the Hitting Set $C$.] The objective in this step is to construct the packets $(u,c(u))$ for all vertices $u\in V$. We first describe how to construct the hitting set $C$, which we execute via a random sampling process. The high-level idea is as follows: 
\begin{enumerate}
    \item We start with $C=\emptyset$. Each vertex selects itself into $C$ with probability $(\ln k)/k$.
    \item For every vertex $u\in V$ such that $N_k(u)\cap  C=\emptyset$, we add $u$ to $C$.
\end{enumerate}
By construction, $C$ must be a hitting set for the neighbourhoods $\set{N_k(u):u\in V}$. It can be shown that by repeating this procedure $\Theta(\log n)$ times, we will have at least one execution where $|C|\in \Theta((n\log k)/k)$ w.h.p., see~\cite[Lemma 6.1, arXiv2405.02695v2]{bui2024improved} or~\cite[Lemma 4.1]{dory2021}. 

We now translate this idea into the language of the MPC model. 
\begin{enumerate}
    \item For every vertex $u\in V$, the leader $\widehat M(u)$ locally selects $u$ into $C$ with probability $(\ln k)/k$. Denoting 
    $$
        i_u=
        \begin{cases}
            1 &\text{if $u$ is added in $C$} \\
            0 &\text{otherwise},
        \end{cases}
    $$
    we relabel every $k$-nearest edge packet into the form $(\mathcal E_{\mathsf{k-near}}, v, \set{u,v},d_G(u,v),i_u,i_v)$ which can be done in $O(1)$ rounds via the MPC tools in \Cref{lemma:mpc-toolbox}.
    \item Every leader $\widehat M(u)$ locally determines whether $N_k(u)\cap C=\emptyset$; if this is true, it adds $u$ to $C$. 
    
    \item The edge packets are re-labelled with the updated values for $i_u$ (for each vertex $u\in V$) in $O(1)$ rounds. Subsequently, each leader $\widehat M(u)$ chooses $c(u)$ as the vertex in $C\cap N_k(u)$ that is closest to $u$.  
\end{enumerate}
We execute the above procedure $\Theta(\log n)$ times \emph{in parallel}. We then choose an execution where $|C|\in O((n\log k)/k)$ -- such an execution occurs w.h.p.~-- and construct the packets $(u,c(u))$ corresponding to that particular execution, for all vertices $u\in V$. Overall, this step requires $O(1)$ rounds in total.

\item[Step 2: Communicating Preliminary Information.] Before constructing the edge sets $E_C$ and $E_B$, we first require the communication of some preliminary information. Namely, the objective of this step is to construct the values $\set{(c(u),d_G(c(u),u)+d_G(u,t)):t
\in N_k(u)}$ and $\set{(c(u),d_G(c(u),u)+\omega_{G}(u,t)):\set{u,t}\in E}$ for all vertices $u\in V$, and to communicate these values to the respective vertices. 

We do so in the following way: Following from Step 0, for every vertex $u\in V$, there exists a contiguous segment of machines $M(u)$ that contains the edges incident to $u$ in $E\cup E_k$. In particular, the leader machine $\widehat M(u)$ knows the $k$-nearest neighbourhood $N_k(u)$ and the value $(u,c(u),d_G(u,c(u)))$.
\begin{enumerate}
    \item $\widehat M(u)$ broadcasts the value $(u,c(u),d_G(u,c(u)))$ to the machines in $M(u)$ (using \Cref{lemma:mpc-toolbox}).
    \item For every vertex $t\in N_k(u)$, we construct the \emph{request packet} 
    $$
        (\mathcal R_\mathsf{k-near},t,c(u),d_G(u,c(u))+d_G(u,t)),
    $$
    which is done locally on the machines $M(u)$.
    \item For every edge $\set{u,t}\in E$ (and for $u=t$), we construct the request packet
    $$
        (\mathcal R_\mathsf{adj},t,c(u),d_G(u,c(u))+\omega_{G}(u,t)),
    $$
    which is also done locally on the machines $M(u)$.
\end{enumerate}
By applying a sorting procedure (\Cref{lemma:mpc-toolbox}), every vertex $t\in V$ now has the packets
\begin{itemize}
    \item $(\mathcal R_\mathsf{k-near},t,c(u),d_G(u,c(u))+d_G(u,t))$ for every vertex $u\in V$ such that $t\in N_k(u)$, and
    \item $(\mathcal R_\mathsf{adj},t,c(v),d_G(v,c(v))+\omega_{G}(v,t))$ for every vertex $v$ such that $\set{v,t}\in E$ (or $v=t$). 
\end{itemize}
stored on a (possibly different) contiguous segment of machines $M(t)$ with leader $\widehat M(t)$. This entire step runs in $O(1)$ rounds.

\item[Step 3: Identifying Bad Vertices $B$.] Recalling the definition of a bad vertex, it suffices for each vertex $t\in V$ to count the number of packets of the form $(\mathcal R_\mathsf{k-near},t,\ldots)$, which we execute by performing the prefix sum and sorting sub-routines in \Cref{lemma:mpc-toolbox}. At the end of this procedure, the leader $\widehat M(t)$ knows the number of packets $(\mathcal R_\mathsf{k-near},t,\ldots)$. More concretely, defining
$$
b_t = \begin{cases}
    1&\text{if }|\set{(\mathcal R_\mathsf{k-near},t,\ldots)}| > k^2\\
    0&\text{otherwise,}
\end{cases}
$$
$\widehat M(t)$ determines the value $(t,b_t)$ and distributes this value to the other machines in $M(t)$ (using \Cref{lemma:mpc-toolbox}). Overall, this entire step runs in $O(1)$ rounds.

\item[Step 4: Building Edges in $E_C$.] For every vertex $t\notin B$, our objective is the following: For every pair of request packets $(\mathcal R_\mathsf{k-near},t,s_a,w_a)$ and $(\mathcal R_\mathsf{adj},t,s_b,w_b)$, create a new edge $\set{s_a,s_b}$ with the edge weight $(w_a+w_b)$. For ease of notation, denote 
\begin{align*}
    K_t &= \set{(c(u),d_G(u,c(u))+d_G(u,t)):t\in N_k(u)},\\
    A_t&=\set{(c(v),d_G(v,c(v))+\omega_G(v,t)):\set{t,v}\in E\vee t=v},
\end{align*}
with $a_t=|A_t|$. Our plan is to create $a_t$ duplicates of the set $K_t$, and ``label'' each duplicate with a value from $A_t$. Therefore, we perform the following: 
\begin{enumerate}
        \item We first apply \Cref{lemma:mpc-duplicate} to create $a_t$ duplicates of $K_t$ in the form $(K_t, 1), (K_t, 2), \ldots, (K_t, a_t)$.
        \item We next apply \Cref{lemma:mpc-label} to relabel each duplicate with a value in $A_t$. In other words, denoting $A_t=\set{(s_1,w_1),(s_2,w_2),\ldots,(s_{a_t},w_{a_t})}$, we now have the duplicates in the form of $(K_t, (s_1,w_1)),(K_t, (s_2,w_2)),\ldots, (K_t, (s_{a_t},w_{a_t}))$.
        \item Lastly, for each duplicate $(K_t, (s_j,w_j))$, we re-construct each packet in $K_t$ from $(\mathcal R_\mathsf{k-near},t,s,w)$ into the edge $(\set{s,s_j},w+w_j)$.
\end{enumerate}  
Since $t\notin B$, we have that $|K_t| \leq k^2$. Additionally, we also observe that $\sum_{t\notin B}a_t\in O(|E|)$. Lastly recalling that $k^2\in O(n^\gamma)$ by assumption, we find that the above procedure can be executed in $O(1)$ rounds with $\tilde O(|E|k^2)$ additional memory. This completes the construction of the edge set $E_C$.

\item[Step 5: Building Edges in $E_B$.] For every vertex $t\in B$, we perform a simple procedure across the machines in $M(t)$ (from Step 3) to construct edges of $E_B$:
\begin{itemize}
    \item For every request packet $(\mathcal R_\mathsf{k-near},t,s_a,w_a)$ received, construct the value $(\set{t,s_a},w_a)$.
    \item For every request packet $(\mathcal R_\mathsf{adj},t,s_b,w_b)$ received, construct the value $(\set{t,s_b},w_b)$.
\end{itemize}
This can be done entirely locally, and completes the construction of the edge set $E_B$.
\end{description}

\begin{claim}\label{claim:skeleton-implement}
    There exists a randomised algorithm to construct the graph $G_S$ (as defined in the earlier modified construction) w.h.p.~in the sublinear MPC model in $O(1)$ rounds with $\tilde O(|E|k^2)$ total memory. 
\end{claim}
\begin{proof}
    We apply the algorithm as described above to construct the edge sets $E_C$ (in Step 4) and $E_B$ (in Step 5), which completes the construction of $G_S$. By the description of the algorithm, we observe each step only requires $O(1)$ rounds, and therefore the entire algorithm runs in $O(1)$ rounds as well. Moreover, only Step 4 requires $\tilde O(|E|k^2)$ additional memory, with the remaining steps only needing $\tilde O(|E|+nk)$ total memory (with the extra $O(nk)$ memory needed to store the edge set $E_k$). Therefore, it suffices to have $\tilde O(|E|k^2)$ total memory overall.
\end{proof}

\paragraph{The Final Proof.} We summarise this section on the skeleton graph with the proof for \Cref{lemma:skeleton-implement}.
\begin{proof}[Proof of \Cref{lemma:skeleton-implement}]
    By \Cref{claim:skeleton-implement}, we construct the graph $G_S=(V_S,E_S)$ for $G$ as defined in the modified construction in $O(1)$ rounds and $\tilde O(|E|k^2)$ total memory. By \Cref{claim:skeleton-correctness}, we find that $G_S$ is indeed a skeleton graph for $G$, with the added property that $|E_S|\in O(|E|k^2)$ as claimed.
\end{proof}

\printbibliography

\appendix
\section{MPC Tools}\label{appendix:mpc-tools}
We describe two tools, the duplication lemma and the labelling lemma, that together enable efficient communication of information between vertices. As a motivating example, suppose each vertex holds $k$ edges that need to be delivered to $d$ other vertices. To facilitate this, we first create $d$ duplicates of each of the $k$ edges and then label each copy with its intended recipient. The labelled messages can then be routed through a sorting procedure. We show that this entire process can be implemented in a constant number of rounds.

\begin{lemma}[Duplication]\label{lemma:mpc-duplicate}
    Suppose we have $n$ sets $S_1,S_2,\ldots,S_n$ of items such that each set $S_i$ occupies at most $O(k)$ space for some integer $k > 0$. For every set $S_i$, the task is to create $d_i\geq 1$ duplicates of the set in the form $(S_i,1),(S_i,2),\ldots,(S_i,d_i)$. 
    
    Given that $k\in O(n^\gamma)$, in the sublinear MPC model with $O\left(nk+\sum_{i=1}^nkd_i\right)$ total memory, there exists a deterministic algorithm to accomplish this task for all sets $S_i$. This algorithm runs in $O(1)$ rounds.    
\end{lemma}
\begin{proof}
    For all $i\in\set{1,2,\ldots,n}$, we may assume that $S_i$ and $d_i$ are appropriately labelled such that there exists a single machine $\widehat M(i)$ that completely stores the tuple $(S_i,d_i)$. 

    To facilitate the duplication procedure, our goal is to first create ``gaps'' of memory between the machines; we then fill in the gaps by duplicating the sets. For simplicity, we may assume that each machine has \emph{exactly} enough local memory to store a single tuple of the form $(S,d)$, where $S$ is a set of $k$ items and $d$ is an integer; the procedure can be generalised for the case of larger local memory. 
    \begin{enumerate}
        \item Define the integers $p_0,p_1,p_2,\ldots, p_n$ as follows:
        \begin{align*}
            p_0 &= 0,\\
            p_i &= d_i + p_{i-1}.
        \end{align*}
        Using the prefix sum sub-routine in \Cref{lemma:mpc-toolbox}, these integers $\set{p_i}$ can be determined in $O(1)$ rounds. We then sort the values such that each machine $\widehat M(i)$ knows the value $p_{i-1}$.

        \item The machine $\widehat M(i)$ sends the tuple $(S_i,d_{i})$ to the machine with ID $(p_{i-1}+1)$. By our total memory constraints (and our simplifying assumption on the local memory), we have sufficiently many machines to perform this communication.
    \end{enumerate}
        Denote $M(i)$ as the contiguous set of machines with IDs $(p_{i-1}+1),(p_{i-1}+2),\ldots,(p_{i-1}+d_i)$, and $\widehat M(i)$ as the machine in $M(i)$ with the smallest ID. After the communication procedure in Step 2, $\widehat M(i)$ knows the value $(S_i,d_{i})$ while the remaining machines in $M(i)$ are empty. Subsequently, $\widehat M(i)$ broadcasts the value $(S_i,p_{i-1})$ to the machines in $M(i)$ in $O(1)$ rounds (\Cref{lemma:mpc-toolbox}). This enables us to create the duplicates $(S_i, 1),(S_i,2),\ldots,(S_i,d_i)$ locally across the machines in $M(i)$, and we are done. 
        
        Overall, this requires $O(1)$ rounds as claimed.
\end{proof}

\begin{lemma}[Labelling]\label{lemma:mpc-label}
Suppose we have $n$ sets $S_1,S_2,\ldots,S_n$ of items such that each set $S_i$ occupies at most $O(k)$ space for some integer $k > 0$. For every set $S_i$, there also exist $d_i\geq 1$ duplicates of $S_i$ in the form $(S_i,1),(S_i,2),\ldots,(S_i,d_i)$. Our task is to label each duplicate $(S_i,j)$ with a label $u_{i,j}$, where each label $u_{i,j}$ occupies $O(1)$ space. 

Given that $k\in O(n^\gamma)$, in the sublinear MPC model with $O\left(nk+\sum_{i=1}^nkd_i\right)$ total memory, there exists a deterministic algorithm to accomplish this task for all sets $S_i$. This algorithm runs in $O(1)$ rounds.   
\end{lemma}
\begin{proof}
    We assume that the labels $u_{i,j}$ and the sets $(S_i,j)$ are appropriately labelled such that by performing a sorting procedure (\Cref{lemma:mpc-toolbox}), there exists a single machine $M(i,j)$ that completely stores the tuple $(S_i,j,u_{i,j})$. Subsequently, by means of local computation, we can then relabel $(S_i,j)$ into $(S_i,u_{i,j})$, and we are done.
\end{proof}

\section{Reporting the Path}\label{appendix:paths}
In this section, our task is the following: Given a weighted undirected graph $G=(V,E)$, construct a distance oracle that receives a pair of vertices $(u,v)$ as query and returns in constant rounds a path $P$ in $G$ from $u$ to $v$ such that $\omega(P)$ is a constant approximation of $d_G(u,v)$. We call this a \emph{path-reporting distance oracle}. We show that this can be accomplished by modifications to our distance oracle algorithms. Namely, our aim is to show the following lemma, which can be seen as a modification of \Cref{lemma-central-construct}. We remark that when applying this lemma, the connectivity assumption can be relaxed using \Cref{lemma:connectivity-reduction} as before (noting that the technique of \Cref{lemma:connectivity-reduction} almost immediately applies to path-reporting oracles as well).

\begin{lemma}\label{lemma:path-distance-oracle}
    Let $G=(V,E)$ be a connected, weighted graph with $n$ vertices, and let $c > 0$ be a constant such that $|E|\in \tilde O(n^{1+{c/2}})$. Suppose that for every vertex $u \in V$, we are given the set $\{(u,t) : t \in \tilde N_k(u)\}$, where $\tilde N_k(u)$ is an $a$-approximate $k$-nearest neighbourhood of $u$ with $k=n^{c/4}$.
    
    Given that $c\leq \gamma$, in the heterogeneous MPC model with $\tilde O(n^{1+c})$ total memory, there exists a randomised algorithm that w.h.p.~constructs an $O(1/c)$-approximate path-reporting distance oracle in $O(\log a + \log\log n)$ rounds. Every subsequent query can then be answered in $O(1)$ rounds.
\end{lemma}

The algorithms in this section are presented primarily as modifications to algorithms discussed previously in the main body of this work, and will be nearly identical in proof to those stated in the main body. Therefore for the sake of clarity and presentation, we only sketch these modifications while providing sufficient exposition to motivate them.  

We begin by recalling the distance oracle construction of \Cref{lemma-central-construct}, which consists of the following two components; here, $k$ is some small polynomial of $n$ determined by the input parameters.
\begin{itemize}
    \item The $k$-nearest edge set $E_k$ of $G$.
    \item A constant-stretch spanner $F=(V_S,E_F)$ of the skeleton graph $G_S=(V_S,E_S)$.
\end{itemize}
A distance query $(u,v)$ is answered as follows. If $u\in N_k(v)$ or $v\in N_k(u)$, then return the edge weight $\omega_{E_k}(u,v)$. Otherwise, we return 
\begin{equation*}
    \delta(u,v) := \min_{x,y\in V_S}(\omega_{E_k}(u,x) + d_F(x,y) + \omega_{E_k}(y,v)).
\end{equation*}
We had previously shown that this can be computed in $O(1)$ rounds and yields a constant approximation of $d_G(u,v)$. Our aim is to output a path in the original graph that has length at most $\delta(u,v)$. Specifically, there are two components to this task: 
\begin{enumerate}
    \item For each edge $\set{u,v}\in E_k$, we should have knowledge of, or be able to compute, some shortest path $P$ from $u$ to $v$. 
    \item Likewise for each edge $\set{x,y}\in E_F\subseteq E_S$, we should have knowledge of, or be able to compute, a path $P$ from $x$ to $y$ such that $\omega(P)\leq \omega_{E_S}(x,y)$.
\end{enumerate}
We discuss each component separately before combining them to form the proof of \Cref{lemma:path-distance-oracle}.

\paragraph{Recovering paths in $E_k$.} The most immediate attempt would be to modify the building block algorithms \Cref{lemma:k-nearest-hopset} and \Cref{lemma:exact-k-nearest} such that each time we shortcut a path (i.e., construct a new edge from a path in the original graph), we also maintain the full path that was shortcut. While this works, we will consider a slightly more space-efficient method that uses a property of $k$-nearest neighbourhoods. Specifically, let $(u,t)$ be any vertex pair with $t\in N_k(u)$, and let $P$ be any shortest path from $u$ to $t$. Observe that for any vertex $v$ along $P$, we necessarily have that $v\in N_k(u)$ as well, i.e., \emph{all the vertices in $P$ are in the $k$-nearest neighbourhood of $u$}. Therefore for each edge $\set{u,t}\in E_k$, instead of storing the entire shortest path from $u$ to $t$, we just store the last edge along this path. 

Concretely, alongside the $k$-nearest edge set $E_k$, we also desire some set $L_k\subseteq V^3$ that satisfies the following properties: 
\begin{itemize}
    \item (\emph{Size-bounded}) $|L_k| \leq nk$.
    \item (\emph{Soundness}) $(u,z,t)\in L_k$ only if $t\in N_k(u)$ and either (1) the edge $\set{t,z}$ is on some shortest path from $u$ to $t$ or (2) $t = z = u$.
    \item (\emph{Completeness}) For all pairs $u,t\in V$ with $t\in N_k(u)$, there exists some vertex $z\in V$ such that $(u,z,t)\in L_k$. 
\end{itemize}
Note that for a fixed graph $G$ and integer $k$, there may be multiple candidate choices for $L_k$. We first show for any pair of vertices $u,t\in V$ such that $t\in N_k(u)$, the set $L_k$ is sufficient to recover some shortest path from $u$ to $t$. 
\begin{lemma}\label{lemma:path-L_k}
    Let $G=(V,E)$ be a connected, weighted graph with $n$ vertices, and let integer $k > 0$ be a parameter. Consider any set $L_k$ with the properties specified above. For any pair of vertices $u,t\in V$ such that $t\in N_k(u)$, there exists a path $u=v_1 - v_2 - v_3 -\ldots - v_m = t$ in $G$ with $m\leq k$ such that (1) the length of this path is $d_G(u,t)$ and (2) for all integers $1\leq i \leq m - 1$ we have $(u,v_{i},v_{i+1})\in L_k$. 
\end{lemma}
\begin{proof}
    We first note that any shortest path from $u$ to $t$ necessarily has at most $k$ distinct vertices (including $u$ and $t$); otherwise, the first $k$ vertices along the path, counting from $u$, have strictly shorter distances to $u$ than that of $t$, hence contradicting $t\in N_k(u)$ (recall that we consider only positive edge weights). Consequently, for any shortest path from $u$ to $t$, all vertices along the path are necessarily in $N_k(u)$ as well. 
    
    We will prove the claim by induction on $d_G(u,t)$. In the base case, we have $d_G(u,t)=0$ and $t=u$ necessarily; the claim immediately follows by the properties of $L_k$. 
    
    In the inductive case, suppose the claim holds for all distances at most $\ell$ and that $d_G(u,t) = \ell+1$. By the properties of $L_k$, observe that there exists some vertex $z\in N_k(u)$ such that $(u,z,t)\in L_k$. In particular, notice that $d_G(u,z) < d_G(u,t)$ -- this follows from the observation that the edge $\set{z,t}$ is on some shortest path from $u$ to $t$, and hence $d_G(u,t) = d_G(u,z) + \omega_G(z,t)$. Therefore, applying the induction hypothesis on $(u,z)$, there exists a path $u=v_1 - v_2 - \ldots - v_m = z$ of length $d_G(u,z)$ such that $(u,v_i,v_{i+1})\in L_k$ for all integers $1\leq i \leq m - 1$. The claim now follows by directly appending the edge $\set{z,t}$ to the end of this path and noticing that the new path indeed has length exactly $d_G(u,t)$; by an earlier observation, we find that the new path has at most $k$ distinct vertices as well. 
\end{proof}

We now show that given approximate $k$-nearest neighbourhoods for all vertices, we can efficiently construct a set $L_k$ satisfying the above properties alongside constructing the $k$-nearest set $E_k$. In particular, it suffices for us to introduce small modifications to the algorithm of \Cref{lemma:E_k}.
\begin{lemma}\label{lemma:path-k-nearest}
    Let $G=(V,E)$ be a connected, weighted graph with $n$ vertices, and let integer $k > 0$ be a parameter. Suppose that for every vertex $u \in V$, we are given the set $\{(u, t) : t \in \tilde N_k(u)\}$, where $\tilde N_k(u)$ is an $a$-approximate $k$-nearest neighbourhood of $u$.
    
    Given that $k^2 \in O(n^\gamma)$, in the sublinear MPC model with $\tilde O(|E| + nk^2)$ total memory, there exists a deterministic algorithm for constructing the $k$-nearest edge set $E_k$ and some set $L_k$ (satisfying the above properties) for $G$ in $O(\log\log d+\log a)$ rounds, where $d$ is the weighted diameter of $G$.
\end{lemma}
\begin{proof}[Proof Sketch]
    Recall that the original algorithm of \Cref{lemma:E_k} was simply a direct combination of two previous algorithms: \Cref{lemma:k-nearest-hopset} and \Cref{lemma:exact-k-nearest}, the algorithms for constructing a $k$-nearest hopset and extracting exact $k$-nearest neighbourhoods respectively. Separately, we show how to modify these algorithms to yield the desired output. 
    \begin{description}
        \item[Modifying \Cref{lemma:k-nearest-hopset}.] Recall that each hopset edge represents some path in the original graph $G$. We modify the output of the algorithm such that for each edge $e$ in the hopset $H$, we also include the first and last edges of the path represented by $e$. 

        \item[Modifying \Cref{lemma:exact-k-nearest}.] We modify the algorithm such that in each iteration $i$ of the graph exponentiation, we update both $E_k$ and $L_k$ corresponding to the new ``local'' view of the $k$-nearest neighbourhoods for each vertex. Crucially for $L_k$, we want to maintain paths in the \emph{original} graph $G$ instead of the graph $G\cup H$ -- this is enabled by the above modification to \Cref{lemma:k-nearest-hopset}. After sufficient iterations, the algorithm outputs both $E_k$ and $L_k$ corresponding to the true $k$-nearest neighbourhoods.\qedhere
    \end{description}
\end{proof}

\paragraph{Recovering paths in $G_S$.} Again, the most obvious attempt seems to be modifying the algorithm of \Cref{lemma:skeleton-implement} such that each time we shortcut a path in $G$ to form an edge in $G_S$, we also maintain the original path in $G$. However, we will utilise our specific implementation of the skeleton graph algorithm to suggest a more space-efficient procedure. 

Consider a skeleton graph $G_S=(V_S,E_S)$ constructed using \Cref{lemma:skeleton-implement} from a graph $G=(V,E)$ and parameter $k$. We recall the algorithm description for \Cref{lemma:skeleton-implement}. Namely, each edge $\set{x,y}\in E_S$ (oriented appropriately) satisfies one of the following cases.
\begin{description}
    \item[Case 1.] There exists a vertex $u\in V$ such that $\omega_{E_S}(x,y) = d_G(x,u) + d_G(u,y)$  satisfying that $x,y\in N_k(u)$.
    \item[Case 2.] There exists a vertex $v\in V$ such that $\omega_{E_S}(x,y) = \omega_G(x,v) + d_G(v,y)$ satisfying that $y\in N_k(v)$, and $\set{x,v}\in E$ or $x=v$. 
    \item[Case 3.] There exist vertices $u,t,v\in V$ such that all of the following hold:
    \begin{itemize}
        \item $x\in N_k(u)$ and $y\in N_k(v)$.
        \item $t\in N_k(u)$, and either $\set{t,v}\in E$ or $t=v$.
        \item $\omega_{E_S}(x,y) = d_G(x,u) + d_G(u,t) + \omega(t,v) + d_G(v,y)$.
    \end{itemize}
\end{description}
Specifically, the edges in the edge set $E_B$ fall under Cases 1 and 2 while the edges in the edge set $E_C$ fall under Case 3 (recalling the definition of $E_B$ and $E_C$ from \Cref{section:skeleton-graph}). This motivates us to modify the algorithm for \Cref{lemma:skeleton-implement} in the following way: When we construct any edge $\set{x,y}\in E_S$, we maintain which of the three cases this edge falls under as well as the vertices in $G$ that were used to construct this edge. 

Concretely, alongside the skeleton graph $G_S = (V_S,E_S)$, we also desire some set $L_S\subseteq E_S\times (V^3\sqcup V^5)\times\set{1,2,3}$ satisfying the following properties.
\begin{itemize}
    \item (\emph{Size-bounded}) $|L_S| = |E_S|$.
    \item (\emph{Soundness}) For any $(\set{x,y},\tau, i)\in L_S$, we have the following; assume that $\set{x,y}$ is oriented appropriately such that it satisfies at least one of the above cases.
    \begin{itemize}
        \item If $\set{x,y}$ satisfies Case 1, then $i = 1$ and $\tau = (x,u,y)$ for some vertex $u\in V$ such that $\omega_{E_S}(x,y) = d_G(x,u) + d_G(u,y)$ and $x,y\in N_k(u)$. 
        \item If $\set{x,y}$ satisfies Case 2, then $i = 2$ and $\tau = (x,v,y)$ for some vertex $v\in V$ such that $\omega_{E_S}(x,y) = \omega_G(x,v) + d_G(v,y)$, $y\in N_k(v)$, and either $\set{x,v}\in E$ or $x=v$. 
        \item If $\set{x,y}$ satisfies Case 3, then $i = 3$ and $\tau = (x,u,t,v,y)\in V^5$, where $u,t,v\in V$ are vertices that satisfy the definition of Case 3. 
    \end{itemize}
    If $\set{x,y}$ satisfies multiple cases, then $\tau$ and $i$ may be assigned any value that matches at least one of these specifications.
    \item (\emph{Completeness}) For any $\set{x,y}\in E_S$, we have that $(\set{x,y},\tau, i)\in L_S$ for some $\tau \in V^3\sqcup V^5$ and $i\in \set{1,2,3}$.
\end{itemize}
As before, we remark that for a fixed skeleton graph $G_S$, there may be multiple candidate choices for $L_S$. It is not difficult to see that given any path $P$ of the form $x = v_1 - v_2 - v_3 -\ldots - v_m = y$ in $G_S$, one can use $L_S$ to ``sketch'' a walk in the original graph $G$ of length at most $\omega(P)$; we refer to the proof sketch of \Cref{lemma:path-distance-oracle} for the details. 

In the following lemma, we show that the skeleton graph algorithm of \Cref{lemma:skeleton-implement} can be slightly modified to also output $L_S$ alongside $G_S$.
\begin{lemma}\label{lemma:path-skeleton}
    Let $G = (V, E)$ be a connected, weighted graph with $n$ vertices, and let $k > 1$ be an integer. Suppose the $k$-nearest edge set $E_k$ for $G$ is given. If $k^2 \in O(n^\gamma)$, then in the sublinear MPC model with $\tilde{O}(|E|k^2)$ total memory, there exists a randomised algorithm that w.h.p.~constructs a skeleton graph $G_S = (V_S, E_S)$ and some set $L_S$ (satisfying the above properties) in $O(1)$ rounds.
\end{lemma}
\begin{proof}[Proof Sketch]
    We make a simple modification to the algorithm description of \Cref{lemma:skeleton-implement}. Namely in Step 2, for each vertex $u\in V$ we modify the request packets to also include the vertex ID of $u$ before communicating these packets to the respective vertices. This provides the vertices $t\in V$ in Steps 4 and 5 with sufficient information to compute $L_S$ alongside the edges $E_S$. Notice that the total space needed to store $L_S$ is at most $O(|E_S|)\subseteq O(|E|k^2)$, and therefore we have sufficient total memory.
\end{proof}

\paragraph{The path-reporting distance oracle.} We now combine \Cref{lemma:path-k-nearest} and \Cref{lemma:path-skeleton} in a nearly identical fashion to how we combined \Cref{lemma:E_k} and \Cref{lemma:skeleton-implement} to form the algorithm of \Cref{lemma-central-construct}. 
\begin{proof}[Proof Sketch of \Cref{lemma:path-distance-oracle}]
    We apply the same preprocessing algorithm technique as the proof of \Cref{lemma-central-construct}, but replace \Cref{lemma:E_k} with \Cref{lemma:path-k-nearest} and \Cref{lemma:skeleton-implement} with \Cref{lemma:path-skeleton}. At the end of this procedure, we have constructed the following:
    \begin{itemize}
        \item The $k$-nearest edge set $E_k$.
        \item The constant-stretch spanner $F = (V_S, E_F)$ of the skeleton graph $G_S=(V_S,E_S)$, stored on the single near-linear machine $M^*$.
        \item Sets $L_k$ and $L_S$ with the properties as specified previously.
    \end{itemize}
    We may also assume that for each vertex $u\in V$, the edge set $\set{\set{u,t}:t\in N_k(u)}$ and the set $\set{(u,z,t)\in L_k: t,z\in V}$ are stored locally on a single machine $M(u)$. We now perform two last additional preprocessing steps: 
    \begin{enumerate}
        \item We filter out $L_S$ to only contain information about the edges in $F$. Specifically, we construct the set $L_F := (E_F \times (V^3\sqcup V^5)\times\set{1,2,3}) \cap L_S$; we proceed to collect $L_F$, which is now near-linear-sized, onto the machine $M^*$. This can be done in $O(1)$ rounds.
        \item For each vertex $u\in V$, the machine $M(u)$ locally computes some shortest path from $u$ to each vertex $t\in N_k(u)$ (\Cref{lemma:path-L_k}). Notice that this requires at most $O(k^2)=O(n^{c/2})\subseteq O(n^\gamma)$ space per vertex $u$, and thus we have sufficient local memory. This requires no communication rounds.
    \end{enumerate}
    We now describe the procedure for answering path-reporting queries. For each query $(u,v)$, the case where $u\in N_k(v)$ or $v\in N_k(u)$ is handled by simply returning the path constructed by the local machine $M(v)$ or $M(u)$ respectively. For the remainder of the proof sketch, we focus on the case when $u\notin N_k(v)$ and $v\notin N_k(u)$.

    Following the same answering-query procedure of \Cref{lemma-central-construct}, we first determine the vertices $x,y\in V_S$ that minimise the distance estimate $\delta(u,v)$. The vertices $(u,v,x,y)$ are then communicated to the near-linear machine $M^*$. Using $F$ and $L_F$, $M^*$ now locally computes a sequence of vertex pairs $(u,w_2)=(w_1,w_2), (w_2,w_3),\ldots, (w_{\ell-1},w_\ell)=(w_{\ell-1}, v) \in V^2$ and a sequence $b_1,b_2,\ldots,b_{\ell - 1}\in \set{1,2}$ with the following properties:
    \begin{itemize}
        \item (\emph{No repeat vertices}) $w_i \neq w_j$ for any $i\neq j$.
        \item (\emph{Size-bounded}) $\ell\in O(|V_S|)\subseteq O((n\log k) / k)$.
        \item (\emph{Pairs are $k$-nearest apart or share an edge}) For each pair $(w_i,w_{i+1})$, we have that (1) if $b_i = 1$ then either $w_i\in N_k(w_{i+1})$ or $w_{i+1}\in N_k(w_{i})$; (2) if $b_i = 2$ then $\set{w_i,w_{i+1}}\in E$.
        \item (\emph{Good approximation}) For each integer $1\leq i \leq \ell - 1$, we denote
        \begin{equation*}
            \Delta_{i} := 
            \begin{cases}
                \omega_{E_k}(w_i,w_{i+1}) & \text{if }b_i = 1 \\
                \omega_G(w_i,w_{i+1}) & \text{if }b_i = 2, 
            \end{cases}
        \end{equation*}
        and let $\Delta := \sum_{i=1}^{\ell-1}\Delta_{i}$. We have that $\Delta$ is at most $\delta(u,v)$.
    \end{itemize}
    It can be quickly verified that such a sequence of vertex pairs indeed exists and can be computed entirely using $F$ and $L_F$. 
    
    Using the paths computed in the last preprocessing step, we use $O(1)$ rounds to construct a \emph{walk} between $u$ and $v$ in $G$ of length at most $\Delta$ and with at most $O(k\cdot \ell)\subseteq O(n\log k)\subseteq \tilde O(n)$ edges; notice that the lack of repeat vertices in the above sequence allows this walk construction to be executed without any added complications. With one last $O(1)$ rounds, the entire walk can be collected onto the near-linear machine $M^*$. Local computation on $M^*$ simplifies this walk into a \emph{path} between $u$ and $v$ of length at most $\Delta$ by removing cycles.
\end{proof}

\end{document}